\documentclass{CSML}
\pdfoutput=1

\usepackage{lastpage}
\newcommand{\dOi}{13(4:9)2017}
\lmcsheading{}{1--\pageref{LastPage}}{}{}%
{Jun.~01, 2015}{Nov.~15, 2017}{}

\usepackage{hyperref}
\hypersetup{hidelinks}
\usepackage{lmcs-preamble}
\renewcommand{\Pl}{\mathrm{Agents}}
\renewcommand{\Behfun}{\BB}
\renewcommand{\Beh}[1]{\BB_{#1}}

\theoremstyle{plain}

\begin{document}

\author[C.\ Eberhart]{Clovis Eberhart\rsuper a}

\author[T.\ Hirschowitz]{Tom Hirschowitz\rsuper b}
\author[T.\ Seiller]{Thomas Seiller\rsuper c}

\address{{\lsuper{a,b}}Univ. Savoie Mont Blanc, CNRS, LAMA, F-73000 Chambéry, France}
\address{\lsuper{c}Department of Computer
  Science, University of Copenhagen, Copenhagen, Denmark and CNRS, UMR
  7030, Laboratoire d’Informatique de Paris Nord Université Paris 13,
  Sorbonne Paris Cité, F-93430 Villetaneuse, France}

\thanks{The authors acknowledge support from French ANR
    projet blancs PiCoq ANR 2010 BLAN 0305 01 and Récré
    ANR-11-BS02-0010.}%

\title{An intensionally fully-abstract sheaf model for $\pi$ \\ {\tiny (expanded version)}}


\keywords{Programming languages; categorical semantics;presheaf
    semantics; game semantics; concurrency; process algebra}
\subjclass{MANDATORY list of acm classifications}


\begin{abstract}
  Following previous work on CCS, we propose a compositional model for
  the $\pi$-calculus in which processes are interpreted as sheaves on
  certain simple sites.  Such sheaves are a concurrent form of
  innocent strategies, in the sense of Hyland-Ong/Nickau game
  semantics.  We define an analogue of fair testing equivalence in the
  model and show that our interpretation is intensionally fully
  abstract for it. That is, the interpretation preserves and reflects
  fair testing equivalence; and furthermore, any innocent strategy is
  fair testing equivalent to the interpretation of some process. The
  central part of our work is the construction of our sites, relying
  on a combinatorial presentation of $\pi$-calculus traces in the
  spirit of string diagrams.
\end{abstract}

\maketitle

\tableofcontents

\section{Introduction}
\subsection{Causal models and beyond}
Operational semantics of programming languages standardly model the
execution of programs as paths in a certain labelled transition system
(\lts{}). Under this interpretation, different possible interleavings
of parallel actions yield different paths. Verification on \ltss{}
thus incurs a well-known state explosion problem. Similarly, causality
between various actions, visible in the syntax, is lost in the \lts{},
thus making, e.g., error diagnostics
difficult~\cite{DBLP:conf/rv/GosslerMR10}.

\emph{Causal models}, originally designed for Petri
nets~\cite{eventstructures} and Milner's
CCS~\cite{DBLP:conf/icalp/Winskel82}, intend to remedy both problems,
but have yet to be applied to full-scale programming languages.  They
have recently been extended in two different directions: (1) by Crafa
et al.~\cite{DBLP:conf/fossacs/CrafaVY12} to Milner's $\pi$-calculus,
and (2) by Melliès~\cite{DBLP:conf/lics/Mellies05} to Girard's linear
logic.  The former extension accounts for the subtle interaction of
channel creation with synchronisation in $\pi$, a significant
technical achievement, 30 years after the first causal semantics for
CCS.  The latter is the first causal model for functional languages
(inspired by Hyland-Ong's and Nickau's \emph{games} models for
PCF~\cite{DBLP:conf/lfcs/Nickau94,DBLP:journals/iandc/HylandO00}).  An
important challenge is now the search for a causal model of
full-fledged languages with both concurrent and functional features.
Winskel and collaborators are currently working in this direction,
using extensions of Melliès's
approach~\cite{RideauW,DBLP:conf/fossacs/Winskel13,PPCF}.

In previous
work~\cite{HP11,DBLP:conf/calco/Hirschowitz13,HirschoDoubleCats}, we
have proposed a causal model for CCS based on a different approach.
We here push this approach further by applying it to the
$\pi$-calculus.

\subsection{Traces and naive concurrent strategies}
In standard causal models, execution traces essentially consist of
partially ordered sets of atomic `events'. Our approach relies on a
new notion of \trace, which we now briefly sketch. There is first a
(straightforward) notion of \emph{\position}, which is
essentially a finite hypergraph whose nodes are thought of as
\emph{\agents}, and whose hyperedges between nodes $x_1, \ldots, x_n$
are thought of as communication channels shared by $x_1, \ldots, x_n$.
There is then a notion of \emph{atomic action} from one \position
to another. The collection of atomic \actions{} is thought of as a
`rule of the game'.  Examples of atomic \actions are: \anagent creates
a new, private communication channel; \anagent forks into two new
\agents connected to the same channels; \anagent sends some channel
$a$ over some channel $b$ to some other \agent.  We finally have a
notion of \trace which allows several atomic actions to occur, in a
way that only retains some minimal causality information between them.
We here mean, e.g., information such as: `such agent outputs on such
channel only after having created such other channel'.




The main purpose of our notion of \trace is to interpret
$\pi$-calculus processes as some kind of strategies over them.  Most
naively, a strategy on some \position $X$ is a prefix-closed set
of `accepted' \traces from $X$. But what should prefix mean in our
setting?  Well, we may view \traces with initial \position $X$
and final \position $Y$ as morphisms $Y \proto X$.  Sequential
composition of \traces, denoted by $\vrond$, yields an analogue of
prefix ordering, defined by $t \leq t \vrond w$.  Strategies as
prefix-closed sets of \traces{} however fail to suit our needs on
three counts. First, such naive strategies may not be stable under
isomorphism of \traces; second, they are bound to model coarse
behavioural equivalences, at least as coarse as may testing
equivalence (a.k.a.\ trace equivalence); and third, they permit
undesirable interaction between players. Let us examine these issues
in more detail.

The first, easy one is that there is an obvious notion of isomorphism
between \traces, under which strategies should be closed.  The second
problem is more serious: until now, these too naive strategies are not
concurrent enough to adequately model CCS or the $\pi$-calculus.
\begin{example}[Milner's coffee machines]\label{ex:coffee} 
  Consider the CCS processes $P = (a.b + a.c)$ and $Q = a.(b+c)$. The
  process $P$ has two ways of inputting on $a$ and then, depending on
  the chosen way, inputs on either $b$ or $c$. The process $Q$ inputs
  on $a$ and then has both possibilities of inputting on $b$ or
  $c$. They hence exhibit significantly different interactive
  behaviour. Both processes, however, accept exactly the same traces
  (in the standard sense), namely $\ens{\epsilon, a, ab, ac}$, where
  $\epsilon$ denotes the empty trace.
\end{example}
Thus, taking strategies to be prefix-closed sets of traces would
prevent us from directly modelling any reasonably fine behavioural
equivalence on processes.  Inspired by \emph{presheaf
  models}~\cite{DBLP:conf/lics/JoyalNW93}, we remedy both problems at
once by passing from prefix-closed sets of traces to presheaves (of
finite sets) on traces. Indeed, in the simple case where \traces on
$X$ form a mere poset $\T(X)$ by prefix ordering, a prefix-closed set
of traces is nothing but a contravariant functor from $\T (X)$ to the
ordinal $\two$, viewed as a category. The latter has two objects $0$
and $1$ and just one non-trivial morphism $0 \to 1$. The idea is that
a functor $S \colon \op{\T (X)} \to \two$ maps any trace to $1$ when
it is accepted, and to $0$ otherwise. Furthermore, if $t \leq t'$,
i.e., $t$ is a prefix of $t'$, then we have a morphism $t \to t'$
which should be mapped by $S$ to some morphism $S (t') \to S (t)$. If
$t'$ is accepted then $S(t') = 1$, so this has to be a morphism
$1 \to S(t)$. Because there are no morphisms $1 \to 0$, this entails
$S(t) = 1$, hence prefix-closedness of the corresponding strategy.
Now in the case where \traces{} form a proper category $\T(X)$, whose
morphisms encompass both prefix ordering and isomorphism between
\traces, considering functors $\op{\T (X)} \to \two$ retains
prefix-closedness and solves our first problem: for any $t \iso t'$,
functoriality imposes $S (t) \iso S (t')$.  Our second problem is then
solved by replacing such functors with presheaves, i.e., functors
$\op{\T (X)} \to \Set$.
\begin{example}
  In Example~\ref{ex:coffee}, the two ways that $P$ has to accept
  inputting on $a$ may be reflected by mapping the trace $a$ to some
  two-element set. More precisely, $P$ may be modelled by the presheaf
  $S$ defined on the left and pictured on the right:
  \vspace{1em}
  \begin{center}
    \begin{minipage}[c]{0.26\linewidth}
      \begin{itemize}
      \item $S (\epsilon) = \ens{\star}$,
      \item $S (a) = \ens{x,x'}$,
      \item $S (ab) = \ens{y}$,
      \item $S (ac) = \ens{y'}$,
      \end{itemize}
    \end{minipage}
    \hfil
    \begin{minipage}[c]{0.44\linewidth}
      \begin{itemize}
      \item $S$ empty otherwise,
      \item $S (\epsilon \into a) = \ens{x \mapsto \star, x' \mapsto
          \star}$,
      \item $S (a \into ab) = \ens{y \mapsto x}$,
      \item $S (a \into ac) = \ens{y' \mapsto x'}$,
      \end{itemize}
    \end{minipage}
    \hfil
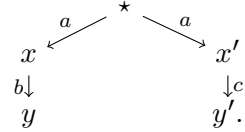
    \begin{minipage}[c]{0.28\linewidth}
      \vspace*{-1.5em} \diag(.15,.8){%
        \& |(root)| \star \& \\
        |(x)| x \& \& |(x')| x' \\
        \\
        |(y)| y \& \& |(y')| y'. %
      }{%
        (root) edge[labelal={a}] (x) %
        edge[labelar={a}] (x') %
        (x) edge[labell={b}] (y) %
        (x') edge[labelr={c}] (y') %
      }
    \end{minipage}
  \end{center}
  \vspace{1em}
  Presheaves thus may `accept a trace in several ways': the trace $a$
  is here accepted in two ways, $x$ and $x'$.  The process $Q$ is of
  course modelled by equating $x$ and $x'$.
\end{example}
As it turns out, we actually only need \emph{finitely many} ways of
accepting each trace. Thus, we arrive at a first sensible notion of
strategy given by presheaves of finite sets, i.e., functors $\op{\T
  (X)} \to \set$, where $\set$ denotes the category with as objects
all finite subsets of $\Nat$, with all maps between them. We call them
\emph{(naive) strategies} in the sequel.  (Please note that $\set$ is
equivalent to the category of all finite sets.)
\begin{notation}\label{not:psh}
  For any $\CCC$, let $\FPsh{\CCC} 
  $ denote
  the category of presheaves of finite sets over $\CCC$.
\end{notation}

\subsection{Innocence as a sheaf condition} 
The third problem evoked above is that functors $\op{\T(X)} \to \set$
allow some undesirable behaviours.  Intuitively, in $\pi$ just as in
CCS, \agents should not have any control over the routing of messages.
\begin{example}\label{ex:inno}
  Consider \aposition $X$ with three agents $x,y$, and $z$
  sharing a communication channel $a$, and a strategy $S$ accepting
  (1) the \trace where $x$ outputs on $a$, (2) the \trace where $y$
  inputs on $a$, and (3) the \trace where $z$ inputs on $a$. Then, both
  synchronisations should be accepted by $S$.  However, one easily
  constructs a naive strategy in which one is refused (see
  Example~\ref{ex:noninnocent}).
\end{example}
In order to rectify this deficiency, we enrich strategies with `local'
information. The idea is that a strategy should not only accept or
refuse \traces on the whole \position $X$, but also \traces on
all possible sub\positions of $X$. Moreover, this local
information should fit together coherently. 
\begin{example}
  Consider the \position $X$ of Example~\ref{ex:inno}. Any
  strategy on $X$ should now in particular include independent
  strategies for each of the three \agents $x$, $y$, and $z$.
  Coherence means that in order for a trace to be accepted, it should
  be enough for it to be `locally accepted', i.e., at every stage in
  the trace, each agent should approve what they see of the next
  action.  E.g., if the next action is a synchronisation $x
  \shortrightarrow y$ with $x$ outputting and $y$ inputting on some
  channel $a$, then all that is required for the synchronisation to be
  accepted is that $x$ accepts to output and $y$ accepts to input.
  Consequently, if some other agent $z$ also accepts to input on $a$
  at this stage, then the synchronisation $x \shortrightarrow z$ is
  also accepted.
\end{example}
We call this putative coherence condition \emph{innocence} by analogy
with Hyland and Ong's notion~\cite{DBLP:journals/iandc/HylandO00}.  In
order to formalise it, we first extend our category of \traces $\T(X)$
on $X$ with new objects representing \traces on sub\positions of
$X$.  We also add new morphisms, which are about `locality'.  Indeed,
standardly, plays form a poset for the prefix ordering, but here we want
to enrich this, e.g., by embedding \traces on sub\positions of
$X$ into \traces on $X$.
\begin{example}
  Consider a \position $X$ with two \agents $x_1$ and
  $x_2$.  There is \atrace $t$ on $X$ in which both \agents fork.
  Consider now the sub\position $Y$ of $X$ consisting solely of
  $x_1$ and the \trace $t'$ on $Y$ in which $x_1$ merely forks.  There
  is a morphism $t' \to t$ in our new category.
\end{example}
This extended category, $\T_X$, yields an intermediate notion of
strategy, given by functors $\op{\T_X} \to \set$. Among the new
objects, we have in particular \traces{} on just one \agent{} of $X$
obtained by sequentially composing atomic actions whose final
\position again consists of one \agent. We call this particular
kind of \trace a \emph{view}. Views are the most `local' kind of
objects in $\T_X$. They form a full subcategory $\V_X$ of $\T_X$.
\begin{example}
  If $X$ merely consists of an agent $x$ linked to $n$ communication
  channels, consider the atomic action given by $x$ forking into two
  new agents, say $x_1$ and $x_2$. This action, viewed as an object of
  $\T_X$ has three subobjects which are views: (1) the `identity'
  view, in which nothing happens, (2) $\paralof{n}$, which represents
  the left-hand branch (to $x_1$), (3) and $\pararof{n}$, which
  represents the right-hand branch (to $x_2$).
\end{example}
The inclusion $\V_X \into \T_X$ induces a simple Grothendieck
topology~\cite{MM} on $\T_X$, which amounts to decreeing that any
trace is covered by its views.  We finally call any $S \colon
\op{\T_X} \to \set$ \emph{innocent} precisely when it is a
\emph{sheaf} for this Grothendieck topology.  In particular, giving an
innocent presheaf on $\T_X$ is equivalent (up to isomorphism) to
separately giving an innocent presheaf for each agent of $X$, which
rules out the undesirable behaviour described in
Example~\ref{ex:inno}.

Sheaves on $\T_X$ form a category $\SS_X$, which is small thanks to
our use of $\set$ instead of $\Set$. They furthermore map back to
naive strategies, i.e., presheaves on $\T(X)$, by forgetting the local
information.  Finally, because the considered topology is particularly
simple, sheaves are equivalent to presheaves on views, i.e., $\SS_X
\equi \FPsh{\V_X}$ (recalling Notation~\ref{not:psh}).  In summary, we
have three categories of strategies: \emph{naive} strategies are
presheaves on the `global' category of \traces $\T (X)$,
\emph{innocent} strategies $\SS_X$ are sheaves on the extended
category of \traces $\T_X$, and so-called \emph{behaviours} $\Beh{X}$
are presheaves on the category of \threads $\Views_X$. The last two
are equivalent, and we furthermore have an adjunction
\adj[.5]{\FPsh{\T(X)}}{\SS_X\rlap{.}}{}{}

We use both sides of the equivalence: behaviours directly lead to our
compositional interpretation $\translfun \colon \picalc \to \SS$ of
$\pi$-calculus processes, and innocent strategies are used below as
the basis for our semantic definition of fair testing equivalence.

\subsection{Main result}
What should we do in order to demonstrate adequacy of our model?  By
definition, causal models expose some intensional information. Hence,
equality is generally much finer than any reasonable behavioural
equivalence, so we should not base our main result on it. On the other
hand, causal models are supposed to be `compositional', i.e., to come
equipped with an interpretation of syntactic operations in the
model. A natural thing to do is thus to choose some behavioural
equivalence defined from syntactic operations, use compositionality to
transpose it to the model, and prove that the two coincide. More
precisely, the considered equivalence induces by quotienting two
`extensional collapses', one syntactic and the other semantic, and we
want to prove that the translation induces a bijection between both
extensional collapses.  Following~\cite{ajm}, we call this
\emph{intensional full abstraction} for the considered equivalence. In
fact, all behavioural equivalences end up relying on some notion of
observation, which we will also need to transpose to the model.

We here focus on so-called \emph{testing
  equivalences}~\cite{DBLP:journals/tcs/NicolaH84,DBLP:conf/icalp/NatarajanC95,DBLP:conf/concur/BrinksmaRV95,DBLP:journals/iandc/RensinkV07},
which are defined in two stages.  First, one chooses a `mode of
interaction'. That is, one defines what the relevant \emph{tests} are
for a given process and specifies how the two should
interact. Typically, tests for $P$ are other processes $T$ with the
same free communication channels as $P$, and interaction is just
parallel composition $P \para T$.  This part will be easy to transpose
to the model by compositionality.  The second stage amounts to
choosing when $P \para T$ is \emph{successful}. E.g., in may testing
equivalence $P \para T$ is successful just when there exists a
transition $(P \para T) \xTo{\tick} P'$ (that is, a $\tick$
transition, possibly surrounded by silent transitions), where $\tick$
is some `tick' action fixed in advance. In must testing equivalence,
success is when all maximal (possibly infinite) transition sequences
contain at least one $\tick$ transition.  In fair testing equivalence,
one requires that all silent sequences $(P \para T) \xTo{} P'$ extend
to some sequence $P' \xTo{} P'' \xto{\tick} P'''$ ending with a
$\tick$ transition.  These ideas transpose to the model by observing
whether a given \trace contains a $\tick$ \action.  In this paper, we
focus on fair testing equivalence, i.e., we prove
(Theorem~\ref{thm:main}) that our model is intensionally
fully-abstract for fair testing equivalence.  We finally show
(Section~\ref{subsec:gen}) that the result generalises to a wide range
of testing equivalences, obtained by varying the notion of success.

In order to fix intuitions, let us quickly motivate must and fair
testing, using \emph{barbed
  congruence}~\cite{DBLP:books/daglib/0004377} as a standard starting
point.  Barbed congruence equates processes $P$ and $Q$, roughly, when
for all contexts $C$, $C[P]$ and $C[Q]$ are weakly bisimilar w.r.t.\
\emph{reduction} (i.e., only $\tau$-actions are allowed), and
furthermore they have the same interaction capabilities at all
stages. Barbed congruence is sometimes perceived as too discriminating
w.r.t.\ guarded choice. Consider, e.g., the following CCS processes.
\begin{center}
  $P_1 \quad = \quad \diag(.3,.3){%
    |(a)| \bullet \& |(b)| \bullet \& |(d)| \bullet
    \& |(f)| \bullet \\
    \& |(c)| \bullet \& |(e)| \bullet \& |(g)| \bullet %
  }{%
    (a) edge[labela={\tau}] (b) %
    edge[labelbl={\tau}] (c) %
    (b) edge[labela={\tau}] (d) %
    edge[labelbl={a}] (e) %
    (d) edge[labela={\tau}] (f) %
    edge[labelbl={b}] (g) %
  }$
\hfil
    $P_2 \quad = \quad \diag(.3,.3){%
      |(a)| \bullet \& |(b)| \bullet \& |(d)| \bullet
      \& |(f)| \bullet \\
      \& |(c)| \bullet \& |(e)| \bullet \& |(g)| \bullet %
    }{%
      (a) edge[labela={\tau}] (b) %
      edge[labelbl={\tau}] (c) %
      (b) edge[labela={\tau}] (d) %
      edge[labelbl={b}] (e) %
      (d) edge[labela={\tau}] (f) %
      edge[labelbl={a}] (g) %
    }$
\end{center}
Both processes may disable both actions $a$ and $b$, the only
difference being that $P_1$ disables $a$ \emph{before} disabling $b$.
Barbed congruence distinguishes $P_1$ from $P_2$ (consider the trivial
context $C = \square$), which some view as a deficiency.

Another possibility would be \emph{must testing}
equivalence~\cite{DBLP:journals/tcs/NicolaH84}.  Recall that $P$
\emph{must pass} a test process $R$ iff all maximal executions of
$P \para R$ perform, at some point, the `tick' action $\tick$.  Then,
$P$ and $Q$ are must testing equivalent iff they must pass the same
tests.  Must testing equivalence does equate $P_1$ and $P_2$ above,
but is sometimes perceived as too discriminating w.r.t.\ divergence.
E.g., consider $Q_1 = {!\tau} \para a$ and $Q_2 = a$. Perhaps
surprisingly, $Q_1$ and $Q_2$ are \emph{not} must testing
equivalent. Indeed, $Q_2$ must pass the test $\abar.\tick$, but $Q_1$
does not, due to an infinite, silent reduction sequence.

Fair testing equivalence was originally introduced (for CCS-like
calculi) to rectify both the deficiency of barbed congruence w.r.t.\
choice and that of must testing equivalence w.r.t.\ divergence.  The
idea is that two processes are equivalent when they \emph{should} pass
the same tests. A process $P$ should pass the test $T$ iff their
parallel composition $P \para T$ never loses the ability of performing
the special `tick' action $\tick$, after any $\tick$-free reduction
sequence. Fair testing equivalence thus equates $P_1$ and $P_2$ above,
as well as $Q_1$ and $Q_2$.  Cacciagrano et
al.~\cite{DBLP:journals/corr/abs-0904-2340} provide an excellent
survey of fair testing for $\pi$.

\begin{example}[\cite{DBLP:journals/corr/abs-0904-2340}]
  The $\pi$-calculus features a well-known encoding of internal choice
  using channel creation and parallel composition. Mixing this with
  replication leads to intriguing examples of fair testing.  Consider
  the following subtly different encodings of $!(b \oplus c)$, where
  $\oplus$ denotes internal choice and $!$ denotes replication: let
  $R_1 = !\nu a.( \abar \para a.b \para a.c)$ and
  $R_2 = \nu a.!( \abar \para a.b \para a.c)$.  These are clearly fair
  testing equivalent. However, each encoding has an execution that
  always makes, say, the left choice, and Cacciagrano et al.\ argue
  that for $R_1$ this execution is fair, as the involved channel is
  different each time. They use similar examples to argue that fair
  testing is in fact too coarse, and instead propose alternative
  notions (which lie beyond the scope of this paper).
\end{example}

\subsection{Contributions}\label{sec:contrib}
Since this paper follows the same approach as previous work on
CCS~\cite{HP11,DBLP:conf/calco/Hirschowitz13,HirschoDoubleCats}, we
should explain in which sense extending the approach to $\pi$ is more
than an easy application.

A first contribution comes from the fact that, in order to even define
composition in our category of traces, we need to show that traces
form the total category of a \emph{fibration}~\cite{Jacobs} over
\positions. In previous work, this was done in an \emph{ad hoc}
way. We here introduce a more satisfactory approach based on
\emph{factorisation systems}~\cite{dualityforgroups,FK}.

A second significant contribution is prompted by the interplay between
synchronisation and private channels in $\pi$, which is notoriously
subtle to handle. And indeed, our proof method for CCS fails miserably
on $\pi$. One reason for this, we think, is that our notion of trace
for $\pi$, though simple and natural, is not `modular' enough, in the
sense that a trace contains strictly more information than the
collection of all `local' information accessible to agents (i.e., of
all of its views, in the above sense).  Thus, adapting our proof
technique from CCS would have required us to define a much more
complex but more modular notion of trace.  Instead, we here take a
somewhat rougher route.

Finally, our proof now applies not only to fair testing equivalence,
but also to a whole class of testing equivalences.

\subsection{Related work}\label{sec:rw}
Beyond the obviously closely related, already mentioned work of
Winskel et al., we should mention other causal models for
$\pi$~\cite{DBLP:conf/concur/BusiG95,DBLP:journals/entcs/MontanariP95,DBLP:journals/tcs/Engelfriet96,DBLP:conf/ctcs/CattaniSW97,DBLP:journals/acta/BorealeS98,DBLP:journals/tcs/DeganoP99,DBLP:conf/lics/CattaniS00,DBLP:conf/concur/BruniMM06,DBLP:conf/fossacs/CrafaVY12,DBLP:journals/jlp/BusiG09},
as well as interleaving
models~\cite{DBLP:conf/lics/FioreMS96,DBLP:conf/lics/FioreT01,DBLP:conf/lics/Stark96,DBLP:conf/lics/CattaniS00,DBLP:journals/tcs/MontanariP05,DBLP:journals/tcs/Hennessy02}
and the early approach~\cite{DBLP:conf/amast/JagadeesanJ95} based on
Girard's Geometry of Interaction.  All of these models are based on
some \lts{} for $\pi$.  Instead, ours is rather based on
\emph{reduction rules}.  The subtleties usually showing up in \ltss{},
related to mixing synchronisation and private channels, do resurface
in our proof of intensional full abstraction, but not in the
definition of our model. Indeed, it merely goes by describing the
`rule of the game' in $\pi$, and applying the general framework of
\emph{playgrounds}~\cite{HirschoDoubleCats}.

Another general framework relating operational and denotational
descriptions of programs is Kleene
coalgebra~\cite{DBLP:conf/concur/BonchiBRS09}, which is mainly
designed for automata theory.  Playgrounds may be viewed as adapting
ideas from Kleene coalgebra to the process-algebraic setting.

We should also mention Laird's games model of (a fragment of)
$\pi$~\cite{DBLP:conf/fsttcs/Laird06}, which accounts for \emph{trace}
(a.k.a.\ \emph{may testing}) equivalence.  Standard game models view
strategies as \emph{sets} of traces (with well-formedness conditions),
so, as we have seen, lend themselves better to modelling trace
equivalence. In a non-deterministic, yet not concurrent setting,
Harmer and McCusker~\cite{DBLP:conf/lics/HarmerHM07} resort to an
explicit action for divergence, which allows them to recover a finer
behavioural equivalence. We feel that the presheaf-based approach is
more general.

Furthermore, recent work by Tsukada and Ong~\cite{OngTsukada} adapts
and extends some ideas of~\cite{HP11,DBLP:conf/calco/Hirschowitz13} to
nondeterministic, simply-typed $\lambda$-calculus. In particular, they
show that innocent strategies as sheaves are compatible with the
\emph{hiding} operation of standard game semantics.  Eberhart and
Hirschowitz further establish~\cite{TOtoEH} a formal link between
Tsukada and Ong's notion of innocence and ours: they construct a model
of nondeterministic, simply-typed $\lambda$-calculus in our style, and
then a morphism of Grothendieck sites, which entails that both models
are equivalent.

Let us moreover mention less closely related work: Girard's
\emph{ludics}~\cite{DBLP:journals/mscs/Girard01}, Melliès's reworking
of game semantics~\cite{Mellies04,DBLP:conf/lics/Mellies05}, the part
of it rediscovered by Levy~\cite{PBLGalop13} with a different
presentation, Melliès's game semantics in string
diagrams~\cite{DBLP:conf/lics/Mellies12}, Harmer et al.'s categorical
combinatorics of innocence~\cite{DBLP:conf/lics/HarmerHM07}, McCusker
et al.'s graphical foundation for
schedules~\cite{DBLP:journals/entcs/McCuskerPW12}, and Winskel's
strategies as profunctors~\cite{DBLP:conf/fossacs/Winskel13}. Finally,
Hildebrandt's work~\cite{DBLP:journals/tcs/Hildebrandt03} also uses
sheaves, though as a tool to correctly handle infinite behaviour, as
opposed to their use here to force reactions of agents to depend only
on their views.

\subsection{Plan}
In previous work~\cite{HirschoDoubleCats}, we have defined an
algebraic notion called \emph{playground}, which provides a sufficient
framework for sheaf-based innocence to make sense.  Namely, it
organises \positions, atomic actions, and traces into a
\emph{pseudo double
  category}~\cite{Ehresmann:double,Ehresmann:double2,GrandisPare,GrandisPareAdjoints,LeinsterHC,GarnerPhD}
with additional structure. Any playground $\D$ automatically gives
rise, among other things, to
\begin{itemize}
\item categories of innocent strategies $\SSX$ on each \position
  $X$, organised into a pseudo double functor from $\op\D$ to small
  categories;
\item a simple, yet complete syntax for innocent strategies, together
  with \anlts{} $\SSS_\D$ for them over an alphabet built from atomic
  actions.
\end{itemize}

After introducing some notation, the considered variant of
$\pi$-calculus, fair testing equivalence (Section~\ref{sec:prelims}),
and recalling the notion of playground, we construct a playground $\D$
for the $\pi$-calculus in Sections~\ref{sec:double}
to~\ref{sec:playground:pi}.  This is a lot of work, and not all
aspects of playgrounds are used in defining the model and proving the
main result. The reason we devote so much energy to it is that
playgrounds provide a really helpful setting, in fact a calculus, for
reasoning about \positions{}, \traces{}, views and the various notions
of strategies.  The underlying pseudo double category is constructed
in Section~\ref{sec:double}.  The main playground axiom, asserting
that a certain functor is a Grothendieck fibration, is established in
Section~\ref{sec:fib}. Finally, the remaining axioms are proved in
Section~\ref{sec:playground:pi}.

We then continue in Section~\ref{sec:model} by applying results
from~\cite{HirschoDoubleCats} to define our sheaf model and semantic
fair testing equivalence, as well as our translation $\translfun$ of
$\pi$.  We then state the main result (Theorem~\ref{thm:main}).  In
Section~\ref{sec:translation}, after introducing the basic notion of
definite residual, we reduce our main theorem to an analogous
statement about \anlts{} $\SSS$ for strategies (derived from
$\SSS_\D$). The advantage of the latter statement is that it lies
entirely in the realm of \ltss{}.  We then define a further, more
syntactic \lts{} $\MMM$ which we prove equivalent to $\SSS$, thus
further reducing the main result to an analogous statement about
$\MMM$. We finally prove the latter, which entails the main result.

\section{Notation and preliminaries}\label{sec:prelims}
We start in this section with a few reminders.  In
Section~\ref{subsec:lts}, after fixing some basic notation, we recall
\ltss{}. We use a slightly more general, graph-based notion than the
standard, relation-based one.  In Section~\ref{subsec:pi}, we
introduce the considered $\pi$-calculus, which is mostly standard
except that (1) we use a presentation in the style of Berry and
Boudol's \emph{chemical abstract
  machine}~\cite{DBLP:conf/popl/BerryB90}, and (2) we consider
infinite terms, thus sparing us the need for recursion or replication
constructs. We then go on and recall fair testing equivalence for
$\pi$ in Section~\ref{subsec:fair}. In fact, because we will also need
to define fair testing for other \ltss{}, we introduce a general
framework in which it makes sense, called \emph{graphs with
  testing}. We further provide sufficient conditions for a relation
between the vertices of two graphs with testing to preserve and
reflect fair testing equivalence (Lemma~\ref{lem:fairness} and
Corollary~\ref{cor:fairness:weak}).  Finally, in
Section~\ref{subsec:playgrounds}, we recall and briefly explain the
definition of playgrounds.

\subsection{Basic notation and labelled transition systems}\label{subsec:lts}
First of all, we adopt the notation
of~\cite[Section~2]{HirschoDoubleCats}, with the slight modification
that $\set$ now denotes the category with finite subsets of $\Nat$ as
objects, and all maps as morphisms.  (This category is equivalent to
what we used in~\cite{HirschoDoubleCats}, but slightly easier to work
with for our purposes.)  For all $n \in \Nat$, we often abuse notation
and let $n$ denote the finite set $\ens{1, \ldots, n}$.  We denote by
$\Chat$ the category of presheaves on $\C$, and by $\FPsh{\C}$ the
category of presheaves of finite sets, i.e., of contravariant functors
to $\set$. For any category $\CCC$, let $\CCC_f$ denote the full
subcategory of \emph{finitely presentable} objects~\cite{Adamek}, or
\emph{finite} objects for short. In the only case where we'll use
this, $\CCC$ will be a presheaf category $[\op{\C},\set]$ and
furthermore due to the special form of $\C$, finite presentability of
$F \in \Chat$ will be equivalent to the category of elements of $F$
being finite, and further equivalent to the set of elements of $F$
being finite, i.e., $\sum_{c \in \ob (C)} F (c)$ is finite.

To recall some bare minimum: we often confuse objects $C$ of a
category $\C$ with the corresponding representable presheaves
$\yoneda_C \in \Chat$. $\Gph$ denotes the category of reflexive
graphs, and all our graphs are reflexive so we often omit mentioning
it. We think of morphisms $p \colon G \to A$ in $\Gph$ as \ltss{} over
the \emph{alphabet} $A$, except that for reasons specific to
playgrounds our convention is that a transition from $x$ to $y$ is
represented as an edge $x \ot y$. Using graphs as alphabets
generalises the standard approach based on sets of labels: indeed, in
order to model any set of labels, take for $A$ the graph with one
vertex and one endo-edge for each label. The extra generality is
useful, e.g., to add some typing information on labels.  Finally,
using graphs as alphabets provides us with standard tools for
transporting \ltss{} across morphisms (by pullback, resp.\
post-composition).

For any graph $G$, $G^\star$ denotes the graph with the same vertices
and all paths between them; on the other hand, $\freecat{G}$ denotes
the free category on $G$, i.e., the category with the same vertices
and identity-free paths between them. Both $(-)^\star$ and
$\freecatfun$ extend to functors, i.e., act on morphisms. We often
silently coerce $\freecat{G}$ into a reflexive graph, and denote by
$\idfree{-}$ the obvious morphism $G^\star \to \freecat{G}$.

For any graph $p \colon G \to A$ over $A$, $x,y \in \ob (G)$, and edge
$a \colon p(y) \to p(x)$ in $A$, we denote by $x \xot{a} y$ the
existence of an edge $e \colon y \to x$ in $G$ such that $p(e) =
a$. When $a = \id$, we just write $x \ot y$.
We denote strong bisimilarity over $A$ by $\bisima$.

For any graph $p \colon G \to A$ over $A$, $x,y \in \ob (G)$, and path
$\rho \colon p(y) \to p(x)$ in $A^\star$, we denote by $x \xOt{\rho}
y$ the existence of a path $r \colon y \to x$ in $G^\star$ such that
$\idfree{p^\star (r)} = \idfree{\rho}$. When $\rho$ is the empty path
we just write $x \xOt{} y$.  We denote weak bisimilarity over $A$ by
$\wbisima$.

\subsection{A \texorpdfstring{$\pi$}{π}-calculus}\label{subsec:pi}
We now present our variant of $\pi$, which features a chemical
abstract machine presentation and infinite terms. Also, we keep track
of the channels known to the considered process, i.e., we work with a
`natural deduction' presentation of terms.

\emph{Processes} are infinite terms coinductively generated by the grammar
\begin{mathpar}
  \inferrule{\gamma \vdashg P_1 \\ \ldots \\ \gamma \vdashg P_n 
    }{ %
      \gamma \vdash \sum_{i \in n} P_i
    }
    \and
    \inferrule{\gamma \vdash P \\ \gamma \vdash Q}{\gamma \vdash P \para Q}
    \\
    \inferrule{\gamma,a \vdash P}{\gamma \vdashg \nu a.P}
    \and
    \inferrule{\gamma \vdash P}{\gamma \vdashg \tick.P}
    \and
    \inferrule{\gamma \vdash P}{\gamma \vdashg \tau.P}
    \and
    \inferrule{a \in \gamma \\ \gamma,b \vdash P}{\gamma \vdashg a(b).P}
    \and
    \inferrule{a,b \in \gamma \\ \gamma \vdash P}{\gamma \vdashg \send{a}{b}.P}~,
\end{mathpar}
where
\begin{itemize}
\item we have two judgements, $\vdash$ for \emph{processes} and
  $\vdashg$ for \emph{guarded processes};
\item $\gamma$ ranges over finite
  sets of natural numbers, and
\item $\gamma,a$ is defined iff $a \notin \gamma$ and then denotes
  $\gamma \uplus \ens{a}$.
\end{itemize}

\begin{notation}
  Let $\picalc$ be the set of all such (non-guarded) processes.
  Let $\picalc_\gamma$ denote the set of processes $\gamma \vdash P$.
\end{notation}

As usual, $a$ is bound in $\nu a.P$ and $b$ is bound in $a(b).P$.
In the following, processes are considered equivalent up to renaming
of bound channels. Capture-avoiding substitution extends the
assignment $\gamma \mapsto \picalc_\gamma$ to a functor $\set \to
\Set$ mapping $\sigma \colon \gamma \to \gamma'$ to $P \mapsto
P[\sigma]$.

Let us now describe the dynamics of our $\pi$-calculus.  They are
slightly unusual, in that they are presented in the style of the
chemical abstract machine.  In particular, there are silent
transitions for `heating' both parallel composition and name
creation. A further slight peculiarity, which we adopt for its
convenience in the chemical abstract machine presentation, is that
name creation is a guard. E.g., we have some processes of the form
$(\nu a. P) + b(x).Q$. This is hardly significant. E.g., the previous
process is strongly bisimilar to $(\tau.\nu a.P) + b(x).Q$ in more
standard settings, and our results are about equivalences coarser than
weak bisimilarity anyway.

\begin{notation}
  For any $\gamma \vdashg P$, $\gamma \vdash Q$ of the form $\sum_{i
    \in n} Q_i$, and injection $h \colon n \into n+1$, we denote by
  $P+_hQ$ the sum $\sum_{j \in n+1} P_j$, where $P_{h(i)} = Q_i$ for
  all $i \in n$ and $P_k = P$, for $k$ the unique element of $(n+1)
  \setminus \im(h)$.
\end{notation}
\begin{defi}\label{def:multisets}
  Let $\Multiset{-}$ denote the finite multiset monad on sets. 
\end{defi}
\begin{defi}
  A \emph{configuration} is an element of $\Conf=\sum_{\gamma \in
    \powfin(\Nat)} \Multiset{\picalc_\gamma}$.
\end{defi}
\begin{notation}\label{not:multisets}
  Configurations $(\gamma,S)$ will be denoted by $\conf{\gamma}{S}$,
  and we will use list syntax $[P_1,\ldots,P_n]$ for multisets,
  sometimes dropping brackets, e.g., as in
  $\conf{\gamma}{P_1,\ldots,P_n}$.  We sometimes resort to a hopefully
  clear `multiset comprehension' notation $[P \aalt \varphi(P)]$.  We
  use $\cup$ for multiset union and $x \cons M = [x] \cup M$.

  Just as $\picalc$, $\Conf$ extends to a functor $\set \to \Set$ by
  capture-avoiding substitution.  
\end{notation}

Let us now extend $\Conf$ to \anlts{} over the alphabet $\ens{\tick,\tau}$.
This means that we need to construct a graph morphism
$\Conf \to \Sierp$, where $\Sierp$ denotes
\begin{center}
\begin{tikzpicture}[baseline=(V.base)]
	\node (V) at (0,0) {$\bullet$};
	\draw[->] (V) .. controls (-1,1) and (-1,-1) .. (V) node [midway,left] {$\tick$};
	\draw[->] (V) .. controls (1,1) and (1,-1) .. (V) node [midway,right] {$\tau$};
\end{tikzpicture},
\end{center}
$\tau$ being the chosen identity edge.

This is done in Figure~\ref{fig:redpi}, omitting identity edges.
There, we let $R$ and $R'$ range over processes of the form
$\sum_{i \in n} P_i$.  The last rule makes sense because each
transition as in the premise implicitly comes with an inclusion
$h \colon \gamma_1 \into \gamma_2$, and the second occurrence of $S$
is implicitly $S[h]$.

\begin{figure}[ht]
  \begin{mathpar}
    \inferrule{ }{\conf{\gamma}{P\para Q} \xtrans{}{\tau}
      \conf{\gamma}{P,Q}} 
\and %
    \inferrule{ }{\conf{\gamma}{\tau.P +_h R} \xtrans{}{\tau}
      \conf{\gamma}{P}} 
\and %
 \inferrule{ }{\conf{\gamma}{\tick. P +_h
        R} \xtrans{}{\tick} \conf{\gamma}{P}} 
\and %
 \inferrule{ }{\conf{\gamma}{\nu a . P +_h R} \xtrans{}{\tau}
      \conf{\gamma,a}{P}} 
\and %
 \inferrule{
    }{\conf{\gamma}{a(b).P +_h R,\send{a}{c}.Q +_{h'} R'} \xtrans{}{\tau}
      \conf{\gamma}{P[b \mapsto c],Q}} 
\and %
    \inferrule{\conf{\gamma_1}{S_1} \xtrans{}{\alpha}
      \conf{\gamma_2}{S_2}}{ \conf{\gamma_1}{S \cup S_1} \xtrans{}{\alpha}
      \conf{\gamma_2}{S \cup S_2}}~(\alpha\in\{\tau,\tick\})
  \end{mathpar}
  \caption{Reduction rules for $\Conf$}
\label{fig:redpi}
\end{figure}

\subsection{Fair testing equivalence}\label{subsec:fair}
Let us now define fair testing equivalence for $\pi$, together with
our general framework of graphs with testing.  Graphs with testing are
essentially De Nicola and Hennessy's original
framework~\cite{DBLP:journals/tcs/NicolaH84}, adapted to our
graph-based presentation of \ltss{}. We derive a few results about
general graphs with testing, notably sufficient conditions for a
relation between two graphs with testing to preserve and reflect fair
testing equivalence.

\begin{rem}
In~\cite{HirschoDoubleCats}, an abstract framework was defined
for studying fair testing equivalence and its relationship with weak
bisimilarity. We won't use this here, and instead work in a simpler
setting.
\end{rem}

We first cover $\pi$-calculus, and then generalise.  The starting
point is that we need to be able to test processes against other
processes, and more generally configurations against
configurations. Because configurations carry their sets of free
channels, it makes sense to consider a partial parallel composition
operation:
\begin{defi}\label{def:atpi}
  For any $\conf{\gamma}{S}, \conf{\gamma'}{S'} \in \Conf$, let
  $\conf{\gamma}{S} @ \conf{\gamma'}{S'}$ denote $\conf{\gamma}{S \cup
    S'}$ if $\gamma = \gamma'$ and be undefined otherwise.
  Let furthermore  $\epsilon_\gamma = \conf{\gamma}{}$.
\end{defi}

\begin{lem}
  The domain of $@$, i.e., the set of pairs $(C,C')$ such that
  $C @ C'$, is an equivalence relation.
\end{lem}
Let us denote by $\coh_\Conf$ this equivalence relation.

Here is the standard definition of fair testing equivalence:
\begin{defi}\label{def:faireq:pi}
Let $\bot^\Conf$ denote the set of
  configurations $C$ such that for all $C \xOt{} C'$, there exists $C'
  \xOt{\tick} C''$.

  Any two configurations $C$ and $C'$ are \emph{fair testing
    equivalent} iff $C \coh_\Conf C'$ and for all $D \coh_\Conf C$,
  $(C @ D) \in \bot^\Conf$ iff $(C' @ D) \in \bot^\Conf$.
\end{defi}

Let us now abstract away from this definition.  For this, it would
make sense to start from a partial, parallel composition map.
However, in the model, the corresponding map will involve a pushout of
positions which is only determined up to isomorphism. We thus
generalise from partial maps to relations, but we need to impose some
conditions.  What matters is the use we will make of parallel
composition for testing. Intuitively, we will check that the parallel
composition $C @ C'$ belongs to some given \emph{pole}, which is
closed under weak bisimilarity over $\Sierp$.  This choice is perhaps
slightly arbitrary, but it encompasses all known testing
equivalences. Finally, to gain just a little more generality, we will
use the fact that weak bisimilarity over $\Sierp$ coincides with
strong bisimilarity over $\freecat{\Sierp}$, and work with the latter
(Definition~\ref{def:pole}). For \ltss{} over $\fcSierp$ obtained by
applying $\freecatfun$ to ones over $\Sierp$, this is equivalent to
the previous setting. The extra generality is useful for describing
\ltss{} which are not free categories. E.g., in
Section~\ref{subsec:SSS}, we introduce the graph with testing $\CCC$,
whose transitions are \traces{}, which compose in a non-free way.
We thus start from the following notion, and then define fair testing
equivalence.
\begin{defi}
  A \emph{graph with testing} is a graph $G$ together with a morphism
  $p \colon G \to \fcSierp$ and a relation $R \colon (\ob (G))^2 \modto
  \ob (G)$ whose domain is an equivalence relation and which is
  partially functional up to strong bisimilarity over $\fcSierp$.

  The domain being an equivalence relation more precisely means that
  the set $\ens{ (x,y) \aalt \exists z. (x,y) \relR z }$ forms an
  equivalence relation over $\ob(G)$.

  Partial functionality up to strong bisimilarity means that if
  $(x,y) \relR z$ and $(x,y) \relR z'$, then $z \bisimfcsierp z'$.
\end{defi}

\begin{notation}
  The relation is called the \emph{testing} relation, and we denote it
  by $\para_G$, i.e., $(x,y) \relR z$ is denoted by $z \in (x \para_G
  y)$. Furthermore, its domain is denoted by $\coh_G$.  We use $\para$
  and $\coh$ when there is no ambiguity.  Since $\para$
  is partially functional up to strong bisimilarity, for any $(x,y) \relR z$, as
  long as what we say about $z$ is invariant under strong bisimilarity,
  then it also holds for any other $z'$ such that $(x,y) \relR z'$. In
  such cases, we implicitly make some global choice of $z$ and
  consider $\para$ as partially functional.
\end{notation}

\begin{example}
Figure~\ref{fig:redpi} defines a morphism
$\pp^\Conf\colon \Conf \to \Sierp$.  Because, $@$ is a partial map, it
induces a partially functional relation
$\ob(\Conf)^2 \modto \ob(\Conf)$, whose domain is an equivalence
relation. Because partially functional implies partially functional up
to strong bisimilarity, we have:
\begin{prop}
  The morphism $\freecat{\pp^\Conf} \colon \freecat{\Conf} \to
  \fcSierp$, with $@$ as testing relation, forms a graph with testing.
\end{prop}
\end{example}

We may now mimick the standard definition of fair testing equivalence
in the abstract setting:
\begin{defi}\label{def:faireq}
  For any graph with testing $p\colon G \to \freecat{\Sierp}$, let
  $\bot^G$ denote the set of objects $x$ such that for all
  $x \xot{} y$, there exists $y \xot{\tick} z$.

  Any two objects $x$ and $y$ are \emph{fair testing equivalent} iff
  $x \coh_G y$ and for all $z \coh_G x$, $(x \para_G z) \in \bot^G$ iff
  $(y \para_G z) \in \bot^G$.
\end{defi}
\begin{rem}
  Because we are working over $\fcSierp$, if $G$ has the shape
  $\freecat{G'}$, then single transitions like $x \xot{} y$ denote
  arbitrary paths of silent transitions. And analogously
  $x \xot{\tick} y$ denotes any path with all edges silent except
  exactly one.
\end{rem}

\begin{notation}
  We denote fair testing equivalence in $G$ by $\faireqof{G}$.  Given
  $x$, any $z$ such that $x \coh z$ is called a \emph{test} for $x$,
  and $x$ \emph{passes} the test iff $(x \para_G z) \in \bot^G$.
\end{notation}

\begin{example}
  Definition~\ref{def:faireq} instantiates to
  Definition~\ref{def:faireq:pi}.
\end{example}

In fact, the construction of the graph with testing
$\freecat{\pp^\Conf}$ from $\pp^\Conf$ is easily generalised:
\begin{lem}\label{lem:freetesting}
  For any morphism $p^G \colon G \to \Sierp$ and relation $R \colon
  (\ob (G))^2 \modto \ob (G)$ whose domain is an equivalence relation,
  $R$ equips $\freecat{p^G}$ with testing structure iff it is partially
  functional up to weak bisimilarity over $\Sierp$.
\end{lem}
\begin{proof}
  $R$ is a strong bisimulation for $\freecat{G}$ iff it is
  a weak bisimulation for $G$.
\end{proof}
\begin{defi}
  A graph with testing is \emph{free} iff it is of the form
  $\freecat{p^G}$.
\end{defi}
Our $\freecat{\pp^\Conf}$ is thus a free graph with testing.

Let us conclude this section with a sufficient condition for a
relation between the vertices of two graphs with testing to be
adequate for fair testing equivalence, and a natural specialisation to
free graphs with testing.
\begin{defi}\label{def:fair}
  A relation $R \colon \ob(G) \modto \ob(H)$ between the vertex
  sets of two graphs with testing $p^G \colon G \to \fcSierp$ and $p^H
  \colon H \to \fcSierp$ is \emph{fair} iff
  \begin{itemize}
  \item $x \relR y$ and $x' \relR y'$ implies $(x \coh_G x') \Leftrightarrow (y \coh_H y')$;
  \item $R$ is total and surjective, i.e.,
    \begin{itemize}
    \item for all $x \in G$, there exists $y \in H$ such that $x \relR y$, and
    \item for all $y \in H$, there exists $x \in G$ such that $x \relR y$;
    \end{itemize}
  \item $x \relR y$ implies $x \bisimfcsierp y$;
  \item if $x \relR y$, $x' \relR y'$, and $x \coh_G x'$, then there
    exist $u \in (x \para_G x')$ and $v \in (y \para_H y')$ such that
    $u \relR v$.
  \end{itemize}
\end{defi}

\begin{lem}\label{lem:symmetry}
  A relation $R \colon \ob(G) \modto \ob(H)$ between the vertex sets
  of two graphs with testing $p^G \colon G \to \fcSierp$ and
  $p^H \colon H \to \fcSierp$ is fair iff its converse
  $\converse{R}$ is, where $\converse{R}$ is defined by
  $(y \mathrel{\converse{R}} x) \Leftrightarrow (x \relR y)$.
\end{lem}
\begin{proof}
  Easy.
\end{proof}

\begin{lem}\label{lem:fairness}
  For any fair relation $R \colon \ob(G) \modto \ob(H)$, 
  if $x \relR y$ and $x' \relR y'$, then
  $(x \faireqof{G} x') \Leftrightarrow (y \faireqof{H} y')$.
\end{lem}

For proving this lemma, we need:
\begin{lem}
  For all graphs with testing $p^G \colon G \to \fcSierp$ and $p^H
  \colon H \to \fcSierp$, $x \in G$ and $y \in H$, if $x \bisimfcsierp
  y$, then $(x \in \bot^G) \Leftrightarrow (y \in \bot^H)$.
\end{lem}
\begin{proof}
  Assume $x \in \bot^G$ and consider any transition $y \ot y'$ in $H$.
  By bisimilarity, $x \ot x' \bisimfcsierp y'$. By hypothesis, we find
  $x' \xot{\tick} x''$, so by bisimilarity, $y' \xot{\tick} y''$.
  Thus, $y \in \bot^H$, which entails the result by symmetry.
\end{proof}
\begin{proof}[Proof of Lemma~\ref{lem:fairness}]
  Consider any such $x,y,x',$ and $y'$. By Lemma~\ref{lem:symmetry},
  it suffices to show one direction of the desired equivalence. So let
  us assume that $x \faireqof{G} x'$.  Then $x \coh_G x'$, hence also
  $y \coh_H y'$ by fairness of $R$.  By symmetry, it again suffices to
  check one direction of the desired implication.  Consider thus any
  $t \coh_H y$ such that $(y \para t) \in \bot^H$. By surjectivity of
  $R$, we find $s \in G$ such that $s \relR t$. By fairness again, we
  find $u \in (x \para s)$ and $v \in (y \para t)$ such that $u \relR
  v$, so $(x \para s) \bisimfcsierp u \bisimfcsierp v \bisimfcsierp
  (y \para t)$, and hence by the previous lemma $(x \para s) \in
  \bot^G$. Because $x \faireqof{G} x'$, this entails $(x' \para s) \in
  \bot^G$, hence by a similar argument $(y' \para t) \in \bot^H$, as
  desired.
\end{proof}

\begin{defi}
  A relation $R \colon \ob(G) \modto \ob(H)$ between the vertex sets
  of two free graphs with testing respectively generated by $p^G
  \colon G \to \Sierp$ and $p^H \colon H \to \Sierp$ is \emph{weakly
    fair} iff it satisfies the conditions of
  Definition~\ref{def:fair}, except for the third one, which is replaced
  by: $x \relR y$ implies $x \wbisimsierp y$.
\end{defi}
\begin{cor}\label{cor:fairness:weak}
  For any weakly fair relation $R \colon \ob(G) \modto \ob(H)$, 
  if $x \relR y$ and $x' \relR y'$, then
  $(x \faireqof{\freecat{G}} x') \Leftrightarrow (y \faireqof{\freecat{H}} y')$.
\end{cor}
\begin{proof}
  Because weak bisimilarity over $\Sierp$ is the same as strong
  bisimilarity over $\fcSierp$, being weakly fair is the same as being
  fair for the generated graphs with testing.
\end{proof}

\subsection{Playgrounds}\label{subsec:playgrounds}
To conclude this preliminary section, let us recall the axioms for
playgrounds~\cite{HirschoDoubleCats}.  Some constructions and results
are developped from these axioms in op.\ cit.  Some of the main ideas
are reviewed and reworked in Sections~\ref{sec:model}
and~\ref{sec:translation}.

Let us start with a brief recap of pseudo double categories.  A pseudo
double category $\D$ consists of a set $\ob (\D)$ of \emph{objects},
shared by a `horizontal' category $\Dh$ and a `vertical' bicategory
$\Dv$. Since we won't consider strict double categories, we'll often
omit the word `pseudo'. Following Paré~\cite{PareYoneda}, $\Dh$, being
a mere category, has standard notation (normal arrows, $\rond$ for
composition, $\id$ for identities), while the bicategory $\Dv$ earns
fancier notation ($\proto$ for arrows, $\vrond$ for composition,
$\idv$ for identities). $\D$ is furthermore equipped with a set of
\emph{double cells} $\alpha$, which have vertical, resp.\ horizontal,
domain and codomain, denoted by $\domv (\alpha)$, $\codv (\alpha)$,
$\domh (\alpha)$, and $\codh (\alpha)$.  The horizontal domain and
codomain of a double cell are vertical morphisms, while the vertical
domain and codomain are horizontal morphisms. E.g., for $\alpha$ in
the diagram delow, we have $u = \domh (\alpha)$, $u' = \codh
(\alpha)$, $h = \domv (\alpha)$, and $h' = \codv (\alpha)$.  $\D$ is
furthermore equipped with operations for composing double cells:
$\rond$ composes them along a common vertical morphism, $\vrond$
composes along horizontal morphisms. Both vertical compositions (of
morphisms and double cells) may only be associative up to coherent
isomorphism. The full axiomatisation is given by
Garner~\cite{GarnerPhD}, and we here only mention the
\emph{interchange law}, which says that the two ways of parsing the
above diagram coincide: $(\beta' \rond \beta) \vrond (\alpha' \rond
\alpha) = (\beta' \vrond \alpha') \rond (\beta \vrond \alpha)$.
\begin{center}
      \Diag{%
      \twocellbr{m-2-1}{m-1-1}{m-1-2}{\alpha} %
      \twocellbr{m-2-2}{m-1-2}{m-1-3}{\alpha'} %
      \twocellbr{m-3-1}{m-2-1}{m-2-2}{\beta} %
      \twocellbr{m-3-2}{m-2-2}{m-2-3}{\beta'} %
    }{%
      X \& X' \& X'' \\
      Y \& Y' \& Y'' \\
      Z \& Z' \& Z'' %
    }{%
      (m-1-1) edge[labelu={h}] (m-1-2) %
      edge[pro,labell={u}] (m-2-1) %
      (m-2-1) edge[labelu={h'}] (m-2-2) %
      (m-1-2) edge[pro,labell={u'}] (m-2-2) %
      (m-1-2) edge[labelu={k}] (m-1-3) %
      (m-2-2) edge[labelu={k'}] (m-2-3) %
      (m-1-3) edge[pro,labelr={u''}] (m-2-3) %
      (m-2-1) edge[pro,labell={v}] (m-3-1) %
      (m-3-1) edge[labeld={h''}] (m-3-2) %
      (m-2-2) edge[pro,labell={v'}] (m-3-2) %
      (m-3-2) edge[labeld={k''}] (m-3-3) %
      (m-2-3) edge[pro,labelr={v''}] (m-3-3) %
    }

\end{center}
For any (pseudo) double category $\D$, we denote by $\DH$ the category
with vertical morphisms as objects and double cells as morphisms, and
by $\DV$ the bicategory with horizontal morphisms as objects and
double cells as morphisms.  Domain and codomain maps form
functors $\dom_v,\cod_v \colon \DH \to \D_h$ and pseudofunctors
$\dom_h,\cod_h \colon \DV \to \D_v$. We will refer to $\domv$ and
$\codv$ simply as $\dom$ and $\cod$, reserving subscripts for $\domh$
and $\codh$.

We then need to recall the notion of fibration (see~\cite{Jacobs}).
Consider any functor $p \colon \E \to \B$. A morphism $r \colon E' \to
E$ in $\E$ is \emph{cartesian} when, as in
\begin{mathpar}
\Diag(.25,.8){%
    }{%
      |(U'')| E'' \\ \\ 
      \& |(U')| E' \& \& |(U)| E \\ 
      |(Y'')| p(E'') \\ \\ 
      \& |(Y')| p(E') \& \& |(Y)| p(E)\rlap{,}  %
    }{%
      (U') edge[labelb={r}] (U) %
      (Y') edge[labelb={p(r)}] (Y) %
      (U'') edge[bend left=10,labelar={t}] (U) %
      (Y'') edge[bend left=10,labelar={p(t)}] (Y) %
      (U'') edge[dashed,labelbl={s}] (U') %
      (Y'') edge[labelbl={k}] (Y') %
    } %
\end{mathpar}
 for all $t
\colon E'' \to E$ and $k \colon p(E'') \to p(E')$, if $p(r) \rond k =
p(t)$ then there exists a unique $s \colon E'' \to E'$ such that $p(s)
= k$ and $r \rond s = t$.
\begin{defi} 
  A functor $p \colon \E \to \B$ is a \emph{fibration} iff for all
  $E \in \E$, any $h \colon B' \to p(E)$ has a \emph{cartesian
    lifting}, i.e., a cartesian antecedent by $p$.
\end{defi}
\begin{notation}
  We denote by $\restr{E}{h}$ the domain of the (chosen) cartesian lifting,
  and call it the \emph{restriction} of $E$ along $h$.
\end{notation}

We may now state the definition of playgrounds. 
\begin{rem}\label{rem:basic}
  The following differs slightly from the original definition, mostly
  in presentation and terminology, but more significantly because the
  class $\B$ of basic moves was mistakenly required to be replete
  in~\cite{HirschoDoubleCats}.
\end{rem}
We provide some intuition right after the definition.

\begin{defi}
  In a double category, a cell $\alpha$ is \emph{special} when
  its vertical domain and codomain $\cod(\alpha)$ and $\dom(\alpha)$
  are identities.
\end{defi}

\begin{defi}\label{def:playground}
  A \emph{playground} is a double category $\D$ such that
  $\cod$ is a fibration, equipped with
  \begin{itemize}
  \item a full subcategory $\DI \into \Dh$ of objects called
    \emph{individuals},
  \item a full, replete\footnote{Replete means closed under
      isomorphism.} subcategory $\DM \into \D_H$, whose objects are
    called \emph{\actions}, with full subcategories $\DB$ and
    $\DF$ of \emph{basic} and \emph{full} \actions, with $\F$ replete,
  \item a map $\length{-} \colon \ob(\DH) \to \Nat$ called the \emph{length},
  \end{itemize}
  satisfying the following conditions:
  \begin{myaxioms}
  \item $\DI$ is discrete. Basic \actions have no non-trivial
    automorphisms in $\DH$.  Vertical identities on individuals have
    no non-trivial endomorphisms.
    \label{discreteness}
  \item (Individuality) Basic \actions have individuals as both domain
    and codomain. \label{individuality}
  \item
      (Atomicity) For any cell $\alpha \colon v \to u$ in $\DH$, if
      $\length{u} = 0$ then also $\length{v} = 0$.  Up to a special
      isomorphism in $\DH$, all plays $u$ of length $n > 0$ admit
      decompositions into $n$ \actions.  For any $u \colon X \proto Y$
      of length 0, there is an isomorphism $\idv_X \to u$ in $\DH$, as in
      \begin{center}
        \diag{%
            |(X0)| X \&|(X)| X \\
            |(Xi)| X \&|(Y)| Y\rlap{.} %
          }{%
            (X0) edge[identity,pro,twol={}] (Xi) %
            edge[identity] (X) %
            (X) edge[pro,twor={u}] (Y) %
            (Xi) edge[labelb={\bar{u}}] (Y) %
            (l) edge[cell=0.3,labelb={\alpha^u}] (r) %
          }
      \end{center}
\label{atomicity}
\item (Fibration, continued) Restrictions of \actions (resp.\ full \actions)
  to individuals either are \actions (resp.\ full \actions), or have length
  0. \label{fibration:continued}
\item (Views) For any \action $M \colon Y \proto X$, and $y
  \colon d \to Y$ with $d\in \DI$, there exists a cell
  \begin{center}
    \diag{%
      |(d)| d \& |(Y)| Y \\
      |(dMy)| d^{y,M} \& |(X)| X, %
    }{%
      (d) edge[labelu={y}] (Y) %
      edge[pro,dashed,twol={v^{y,M}}] (dMy) %
      (Y) edge[pro,twor={M}] (X) %
      (dMy) edge[dashed,labeld={y^M}] (X) %
      (l) edge[cell=.3,dashed,labelu={\scriptstyle \alpha^{y,M}}] (r) %
    }
  \end{center}
  where $v^{y,M}$ either is a basic \action or has length $0$, which
  is unique up to canonical isomorphism, i.e., for any
  $y' \colon d' \to X$, $v' \colon d \proto d'$, and
  $\alpha' \colon v' \to M$ with $\dom(\alpha') = y$ and
  $\cod(\alpha') = y'$, we have $y' = y^M$ and there exists a unique
  isomorphism $\beta \colon v \to v'$ making the diagram
\begin{center}
  \Diag (.3,.5) {%
    \twocellr{dMy}{d}{Y}{\alpha^{y,M}} %
    \twocellr[.5]{d''}{d}{Y}{\alpha'} %
    \twocellro[.5]{dMy}{d}{d''}{\beta} %
     \path[->,draw] %
    (dMy) edge node[pos=.7,above] {$\scriptstyle y^M$} (X) %
    edge[identity] (d'')
    (d'') edge[labelbr={y'}] (X) %
    (d) edge[pro,fore,labeloat={v'}{.352}] (d'') %
    ; %
  }{%
    \& |(d)| d \& \& \& |(Y)| Y \\
    \& {} \\
    |(dMy)| d^{y,M} \& \& \& \& |(X)| X \\
    \& \& |(d'')| d' \& %
  }{%
    (d) edge[labelu={y}] (Y) %
    edge[pro,twol={v^{y,M}}] (dMy) %
    (Y) edge[pro,twor={M}] (X) %
  }    
\end{center}
commute. \label{views}
  \item (Left decomposition) 
    Any double cell
    \begin{center}
      \Diag{%
        \twocellbr{B}{A}{X}{\alpha} %
      }{%
        |(A)| A \& |(X)| X \\
        \& |(Y)| Y \\
        |(B)| B \& |(Z)| Z %
      }{%
        (A) edge[labelu={h}] (X) %
        edge[pro,labell={u}] (B) %
        (X) edge[pro,labelr={w_1}] (Y) %
        (Y) edge[pro,labelr={w_2}] (Z) %
        (B) edge[labeld={k}] (Z) %
      }
      \hfil decomposes as \hfil
      \Diag(1,2){%
        \twocellbr[.3]{C}{A}{X}{\alpha_1} %
        \twocellbr[.3]{B}{C}{Y}{\alpha_2} %
        \twocell[.3]{L}{A}{C}{}{celllr={0}{0},bend
          right,fore,labelbl={\scriptscriptstyle \alpha_3}} %
      }{%
        \& |(A)| A \& |(X)| X \\
        |(L)| \& |(C)| C \& |(Y)| Y \\
        \& |(B)| B \& |(Z)| Z %
      }{%
        (A) edge[labelu={h}] (X) %
        edge[pro] node[pos=.5,anchor=north west] {$\scriptscriptstyle {u_1}$} (C) %
        edge[pro,bend right=70,labell={u}] (B) %
        (C) edge[pro] node[pos=.5,anchor=north west] {$\scriptscriptstyle {u_2}$} (B) %
        edge[labelu={l}] (Y) %
        (X) edge[pro,labelr={w_1}] (Y) %
        (Y) edge[pro,labelr={w_2}] (Z) %
        (B) edge[labeld={k}] (Z) %
      }
    \end{center}
    with $\alpha_3$ an isomorphism, in an essentially unique way.
    \label{leftdecomposition}
  \item (Right decomposition) Any double cell as in the center below,
    where $b$ is a basic \action and $M$ is \anaction, decomposes in exactly
    one of the forms on the left and right:
  \begin{mathpar}
    \begin{minipage}[t]{0.18\linewidth}
      \centering \Diag {%
        \twocell[.4][.3]{B}{A}{X}{}{celllr={0.0}{0.0},bend
          right=30,labelbr={\alpha_1}} %
        \twocell[.4][.3]{C}{B}{Y}{}{celllr={0.0}{0.0},bend
          right=20,labelbr={\alpha_2}} %
      }{%
    |(A)| A \& |(X)| X \\
     |(B)| B \& |(Y)| Y \\
    |(C)| C \& |(Z)| Z %
      }{%
        (A) edge[pro] (B) %
        edge (X) %
        (B) edge (Y) %
        edge[pro] (C) %
        (C) edge (Z) %
        (X) edge[pro] (Y) %
        (Y) edge[pro] 
        (Z) %
      } 
    \end{minipage}
    \and \leftlsquigarrow \and
    \begin{minipage}[t]{0.18\linewidth}
      \centering      \Diag{%
    \twocellbr{B}{A}{X}{\alpha} %
  }{%
    |(A)| A \& |(X)| X \\
     |(B)| B \& |(Y)| Y \\
    |(C)| C \& |(Z)| Z %
  }{%
    (A) edge[labelu={h}] (X) %
    edge[pro,labell={w}] (B) %
    (B) edge[pro,labell={b}] (C) %
    (X) edge[pro,labelr={u}] (Y) %
    (Y) edge[pro,labelr={M}] (Z) %
    (C) edge[labeld={k}] (Z) %
  }
    \end{minipage}
\and \rightrsquigarrow \and
     \begin{minipage}[t]{0.18\linewidth}
      \centering \Diag {%
        \twocell[.4][.3]{B}{A}{X}{}{celllr={0.0}{0.0},bend
          right=30,labelbr={\alpha_1}} %
        \twocell[.3][.4]{C}{Y}{Z}{}{celllr={0.0}{0.0},bend
          right=20,labeld={\alpha_2}} %
      }{%
    |(A)| A \& |(X)| X \\
     |(B)| B \& |(Y)| Y \\
    |(C)| C \& |(Z)| Z. %
      }{%
        (A) edge[pro] (B) %
        edge (X) %
        (B) 
        edge[pro] (C) %
        (C) edge (Z) %
        edge (Y) %
        (X) edge[pro] (Y) %
        (Y) 
        edge[pro] (Z) %
      } 
    \end{minipage}
  \end{mathpar}
    \label{views:decomp}
  \item (Finiteness) For any object $X$, the comma category $\DI / X$
    (taken in $\Dh$) is finite. \label{finiteness}
  \item (Basic vs.\ full) For all $d \in \DI$ and \actions
    $M \colon X \proto d$, $M' \colon X' \proto d$, and
    $b \colon d' \proto d$ with $M$ and $M'$ full and $b$ basic, if
    there exist cells $M \ot b \to M'$ then $M \iso M'$.
    \label{basic:full}
\end{myaxioms}
\end{defi}
Intuitively, the objects of $\D$ are configurations, or \positions, in
the game. The considered games are multi-party, so it makes sense to
consider embeddings of \positions: this is intended to be described by
the horizontal category $\Dh$.  The vertical category $\Dv$ is that of
\traces, or plays: morphisms $u\colon Y \proto X$ model plays from the
initial \position $X$, to the final one, $Y$. Finally, cells model
embeddings of plays, preserving initial and final position. E.g., this
could model embedding the part of a play involving a particular
\agent.

Individuals are intended to model \agents in a \position, with a
role similar to that of representable presheaves among general ones.
Typically, the hom-set $\Dh(d, X)$, with $d \in \DI$, models the set
of \agent slots of type $d$ in the \position $X$.  For example,
in the playground we construct below for $\pi$, individuals bear the
number of channels that \anagent is connected to, so that a morphism
$[n] \to X$ amounts to an $n$-ary \agent in $X$, i.e., \anagent
connected to $n$ channels.  Accordingly, the object $[n]$ models a
position with just one $n$-ary \agent.

\Actions model moves in the game, and atomicity~\axref{atomicity}
notably says that any play decomposes into moves.  The distinction
between \emph{basic} and \emph{full} \actions has to do with
innocence. The two notions are more or less dual: basic \actions are
as thin as possible, while full ones are as wide as possible.
Intuitively, \anaction is full when it cannot embed into a larger one
(unless possibly some \agents are added), while it is basic when it
cannot embed any smaller one. As we will see, in our playground for
$\pi$, we have \anaction for forking, which describes how \anagent{}
$x$ may fork into two, say $x_1$ and $x_2$. This \action is full, and
it has two basic sub-\actions, which respectively model the passage
from $x$ to $x_1$ and to $x_2$.  For another example, we also have
\anaction for inputting on some given channel: it is obviously basic,
but in fact also full, because the only way to embed it into a wider
\action is to add \anagent and do a synchronisation.  In view of this,
it should be natural that basic \actions have individuals as their
domain and codomain~\axref{individuality}.  As alluded to
in~\axref{views}, views will be defined as composites of basic
\actions.  Axiom \axref{views} intuitively enforces existence of one
sub-\action for each \agent in the final position of any
\action. Extending this to general plays will yield an operation
analogous to taking the view of a play in standard game
semantics. Axiom~\axref{basic:full} requires that these basic
sub-\actions may not be shared among different full \actions.

Both~\axref{leftdecomposition} and~\axref{views:decomp} are
decomposition axioms.  The former says that a decomposition of a play
reflects essentially uniquely onto any subplay.  The latter
essentially says that basic \actions are strictly sequential: if any
play of the form $b \vrond w$ with $b$ basic embeds into some other
play, then the image of $w$ should occur after that of $b$.
This is expressed in a slightly convoluted way by saying that
if the latter play decomposes as $M \vrond u$, then
\begin{itemize}
\item either $b$ maps to $M$, in which case $w$ should map to $u$,
\item or $b$ maps to $U$, in which case $w$ should also map to $u$.
\end{itemize}

This should make most of the axioms rather intuitive: the others are
technical, which means that they emerged from our attempts to make
things work out, but that we are not yet able to explain them
satisfactorily.

\section{A pseudo double category of \traces}\label{sec:double}
In this section, we introduce our notion of \trace, which is based on
certain combinatorial objects, close in spirit to string diagrams.  We
first define these string diagrams, and then use them to define
\traces. \Positions are special, hypergraph-like string diagrams
whose vertices represent \agents and whose hyperedges represent
channels.  A perhaps surprising point is that \actions are not just a
binary relation between \positions, because we not only want to
say \emph{when} there is \anaction from one \position to another,
but also \emph{how} this \action is performed. This will be
implemented by viewing \actions from $X$ to $Y$ as \emph{cospans} $Y
\to M \ot X$ in a certain category $\Chatf$, whose objects we call
higher-dimensional string diagrams for lack of a better term. The idea
is that $X$ and $Y$ respectively are the initial and final
\positions, and that $M$ describes how one goes from $X$ to $Y$.
By combining such \actions (by pushout), we get a bicategory $\Dv$ of
\positions and \traces. Finally, we recast $\Dv$ as the vertical
bicategory of a pseudo double category $\D$.

\subsection{String diagrams}\label{subsec:string}
The category $\Chatf$ will be a category of finite presheaves over a base
category, $\C$.  Let us motivate the definition of $\C$ by recalling
that (directed, multi) graphs may be seen as presheaves over the
category with two objects $\star$ and $[1]$, and two non-identity
morphisms $s,t \colon \star \to [1]$. Any such presheaf $G$ represents the
graph with vertices in $G(\star)$ and edges in $G[1]$, the source and
target of any $e \in G[1]$ being respectively $G (s) (e)$ and $G (t)
(e)$, or $e \cdot s$ and $e \cdot t$ for short. A way to visualise how
such presheaves represent graphs is to compute their \emph{categories
  of elements}~\cite{MM}. Recall that the category of elements $\el
G$ for a presheaf $G$ over $\C$ has as objects pairs $(c,x)$ with $c
\in \C$ and $x \in G(c)$, and as morphisms $(c,x) \to (d,y)$ all
morphisms $f \colon c \to d$ in $\C$ such that $y \cdot f = x$. This
category admits a canonical functor $\pi_G$ to $\C$, and $G$ is the
colimit of the composite $\el G \xto{\pi_G} \C \xto{\yoneda} \Chat$
with the Yoneda embedding. E.g., the category of elements for
$\yoneda_{[1]}$ is the poset $(\star, s) \xto{s} ([1],\id_{[1]}) \xot{t}
(\star, t)$, which could be pictured as
\diagramme[stringdiag={0.1}{0.6}]{baseline=(A.south)}{%
  \path[-,draw] %
  (A) edge (E) %
  (B) edge (E) %
  ; %
  \node at ($(B.south east) + (.1,0)$) {,} ;%
}{%
  \joueur{A} \& \node[regular polygon,anchor=center,regular polygon
  sides=3,fill,minimum size=3pt,draw,rotate=-90] (E) {}; \&
  \joueur{B} %
}{%
} \hspace*{-.7em} where dots represent vertices, the triangle
represents the edge, and links materialise the graph of $G(s)$ and
$G(t)$, the convention being that $t$ connects to the apex of the
triangle.  We thus recover some graphical intuition.


Let us give the formal definition of $\C$ for reference.  We advise to
skip it on a first reading, as we then attempt to provide some
intuition.
\begin{definition}
  Let $G_{\C}$ be the graph with, for all $n$, $m$, with $a,b \in n$ and $c,d \in m$:
  \begin{itemize}
  \item vertices $\star$, $[n]$, $\forkln$, $\forkrn$, $\forkn$,
    $\nun$, $\tickn$, $\taun$, $\inna$, $\outnab$, and $\taunamcd$;
  \item edges $s_1,...,s_n : \star \to [n]$, plus, $\forall v \in
    \ens{\forkln,\forkrn,\tickn,\taun,\outnab}$, edges $s,t : [n] \to
    v$;
  \item edges $[n] \xto{t} \nun \xot{s} [n+1]$ and $[n] \xto{t} \inna
    \xot{s} [n+1]$;
  \item edges $\forkln \xto{l} \forkn \xot{r} \forkrn$ and $\inna \xto{\rho} \taunamcd \xot{\epsilon} \outmcd$.
  \end{itemize}

  Let $\C$ be the free category on $G_{\C}$, modulo the equations
  \begin{mathpar}
    {s \rond s_i = t \rond s_i}
    \and 
    {l \rond t = r \rond t}  \and
    {\rho \rond t \rond s_a = \sender \rond t \rond s_c}  \and 
    {\rho \rond s \rond s_{n+1} = \sender \rond s \rond s_d}.
  \end{mathpar}
  The first equation should be understood in $\C(\star,v)$ for all $n
  \in \Nat$, $i \in n$, and $v \in {\cup_{a,b \in n}} \{\forkln,
  \forkrn, \linebreak\tickn, \taun, \inna, \outnab, \nun\}$.  (This is rather
  elliptic: if $v$ has the shape $\inna$ or $\nun$, $s \rond s_i$ is
  really $\star \xto{s_i} [n+1] \xto{s} v$.)  The second equation should
  be understood in $\C([n],\paran)$ for all $n$, and the last two in 
  $\C(\star, \taunamcd)$, for all $n,m$, $a \in n$, and $c,d \in m$.
\end{definition}

Our category of string diagrams is the category of finite presheaves
$\Chatf$.  
A presheaf $X$ over $\C$ is a kind of higher-dimensional graph whose
components are typed by objects of $\C$:
\begin{itemize}
\item $X(\star)$ is the set of
vertices, or \emph{channels}; 
\item $X[n]$ is the set of \emph{\agents}
connected to $n$ channels (which are given by $X(s_i)$);
\item $X(\iotana)$ is the set of \emph{input \actions} by some $n$-ary \agent
  on its $a$th channel; 
\item $X(\outmcd)$ is the set of \emph{output \actions} by
  some $m$-ary \agent on its $c$th channel of its $d$th channel;
\item $X (\taunamcd)$ is the set of \emph{synchronisations}
  between some input and some output on a common channel;
\item $X(\forkn)$ is the set of \emph{forking \actions} by some $n$-ary \agent;
\item and similarly for $X(\forkln)$, $X(\forkrn)$, $X(\nun)$,
  $X(\tickn)$, and $X (\taun)$.
\end{itemize}
To see these intuitions at work, let us compute a few categories of
elements. Let us start with an easy one, that of $[3] \in \C$
(recalling that we implicitly identify any $c \in \C$ with
$\yoneda_c$). An easy computation shows that it is the poset pictured
in the top left part of Figure~\ref{fig:tau}. We think of it as
\aposition with one \agent $([3],\id_{[3]})$ connected to three
channels, and draw it as in the top right part, where the bullet
represents the \agent, and circles represent channels.  
\begin{defi}
  \emph{\Positions} are finite presheaves empty except perhaps on
  $\star$ and $[n]$s.
\end{defi}
Let us organise \positions{} into a category, by designing a notion of
morphism.  We may equip the objects of $\C$ with a \emph{dimension}:
$\star$ has dimension $0$, any $[n]$ has dimension $1$, all of
$\tau_n,\forkln,\forkrn,\tickn,\iotana,\outncd,\nun$ have dimension
$2$, $\forkn$ has dimension $3$, $\taunimjk$ has dimension~$4$.

\begin{defi} We accordingly define the \emph{dimension} of a
presheaf $X$ on $\C$ to be the lowest $n \in \Nat$ such that for any
$m \in \C$ of dimension strictly greater than $n$, $X(m) = \emptyset$.

A \position is thus equivalently a finite presheaf in
$[\op{\C},\set]$ of dimension at most $1$. An \emph{interface} is one
of dimension~$0$.
\end{defi}

  \begin{defi}
    A map in $\Chat$ is \emph{1-injective} iff it is injective in all
    strictly positive dimensions.
  \end{defi}

  A \emph{morphism of \positions} is a 1-injective morphism in
  $\Chat$.  The intuition for a morphism $X \to Y$ between
  \positions is thus that $X$ embeds into $Y$, possibly
  identifying some channels.
\begin{definition}\label{def:Dh}
  \Positions and morphisms between them form a category $\Dh$.
\end{definition}

  \newcommand{\longueurfigun}{.6}
  \newcommand{\separation}{} 
  \begin{figure*}[t]
    \centering
      \diagramme[stringdiag={.8}{1.3}]{}{%
}{%
  \node (s_1) {$(\star, s_1)$}; \& \node (s_2) {$(\star, s_2)$}; \& \node (s_3) {$(\star, s_3)$}; \\
    \& \node (id) {$([3], \id_{[3]})$}; 
    }{%
      (s_1) edge (id) %
      (s_2) edge (id) %
      (s_3) edge (id) %
    }
    \hfil
      \diagramme[stringdiag={.8}{1.3}]{}{%
    \path[-,draw] %
    (a) edge (j1) %
    (c) edge (j1) %
    (b) edge (j1) %
    ; %
}{%
    \canal{a}     \& \canal{b} \&  \canal{c} \\
    \& \joueur{j1}
    }{%
    }
    \\
    \diag (.4,.3) {%
      \& \& \& \& |(lt)| l s \& \& |(rt)| r s \& \& \\ 
      \& \&       |(lt1)| l s s_1 \& \& |(l)| l \& |(para)| \id_{\paraof{2}} \& |(r)| r \& \& |(lt2)| l s s_2 \\ 
      \& \&       \& \& \& |(ls)| l t = r t \&  \& \& 
    }{%
      (lt1) 
      edge (ls) %
      (lt2) 
      edge (ls) %
      (ls) edge (l) edge (r) %
      (lt) edge (l) %
      (rt) edge (r) %
      (l) edge (para) %
      (r) edge (para) %
      (lt1) edge[identity] (lt1) 
      edge (lt) %
      edge[fore,bend left=10] (rt) %
      (lt2) edge[identity] (lt2) 
      edge (rt) %
      edge[bend right=10,fore] (lt) %
    }
    \hfil
          \diagramme[stringdiag={.3}{.6}]{}{
  }{%
     \& \& \joueur{t_1} \&  \& \joueur{t_2} \\
    \& \&   \&  \\
    \& \ \& \\
    \canal{t0} \& \& \& \couppara{para} \& \& \& \canal{t1} \\ 
    \& \ \& \\
    \& \&  \\
    \& \& \& \joueur{s} \& \& 
  }{%
    (para) edge[-] (t_2) %
    (t1) edge[-,bend right=10] (t_2) %
    (t0) 
    edge[-] (s) %
    (t1) 
    edge[-] (s) %
    (s) edge[-] (para) %
    (para) edge[fore={.3}{.3},-] (t_1)
    (t0) edge[fore={.5}{.5},-,bend left=10] (t_2) %
    (t0) edge[-,bend left=15] (t_1) %
    (t1) edge[-,fore={1}{.5},bend right=10] (t_1) %
  }  \\[-2.5em]
      \diagramme[diag={.2}{.5}]{}{%
        \path (iota) -- node[pos=.6,inner sep=0pt,outer sep=0pt] (tau) {$\scriptscriptstyle \id_{\taunamcd}$}  (iota') ; %
        \path[->] %
        (t1) edge[bend left=10] (s) %
        (t2) edge (s) %
        (t0) edge (t) %
        (t0) edge (s) %
        (t2) edge (t) %
        (t1) edge (s') %
        (t') edge (t2) %
        (t2) edge[->] (iota') %
        ; %
        \path[->,draw] %
        (t1) edge[fore={.3}{.5},bend left=10] (iota) %
        (t1) edge[fore={0.3}{.5}] (iota') %
        ; %
        \path[draw,->] (t2) edge[->,bend right=10] (iota) %
        (iota) edge[fore={.4}{0}] (tau) %
        (iota') edge[fore={.4}{0}] (tau) %
        ; %
        \path[->] (t1) edge[fore={.3}{.3}] (t') %
        ; %
        \path[->] (t1) edge[fore={.1}{.5},bend left=15] (t) ; %
        \foreach \x/\y in {s/iota,t/iota,s'/iota',t'/iota'} \path[->]
        (\x) edge (\y) ; %
      }{%
        \&  |(t)| \epsilon s \& \&   \& \&  |(t')| \rho s   \\
        \&   \& \& |(t2)| \epsilon t s_3 \\
        |(t0)[anchor=base west]| \epsilon t s_1 \&  |(iota)| \epsilon \& \& \& 
        \&  |(iota')| \rho \\
        \&  \& \& \& |(t1)| \epsilon t s_2 \\
        \& |(s)| \epsilon t \& \& \& \& |(s')| \rho t 
      }{%
      }%
      \hfil
      \diagramme[stringdiag={.5}{.8}]{}{%
        \path[-] %
        (s) edge[bend right=5] (t1) %
        (t2) edge (s) %
        (t0) edge (t) %
        (t0) edge (s) %
        (t2) edge (t) %
        (s') edge (t1) %
        (t') edge (t2) %
        (iota') edge[gray,very thin] (t2) %
        ; %
        \path[-] %
        (iota) edge[fore={.3}{.5},very thick,bend right=10] (t1) %
        (iota') edge[fore={0.3}{.5},very thick] (t1) %
        ; %
        \twocell[.3][.5]{iota}{t1}{iota'}{}{
          decorate,decoration={snake,amplitude=.3mm,segment length=1mm},bend left=30,
        shorten <=-.1cm}
        \path[] (t2) edge[color=gray,very thin,bend right=10]
        node[coordinate,pos=.6] (iotatip) {} (iota) ; %
        \path (iota') -- (t2) node[coordinate,pos=0.55] (iotatip') {}
        ; %
        \path[-] (iotatip) edge[bend right=7,-latex] (iota) ; %
        \path[-] (iota') edge[-latex] (iotatip') ; %
        \path[-] (t1) edge[fore={.3}{.3}] (t') %
        ; %
        \path[-] (t) edge[fore={.1}{.5},bend right=10] (t1) ; %
        \foreach \x/\y in {s/t,s'/t'} \path[-] (\x) edge (\y) ; %
        \node[anchor=south] at (t2.north) {$\scriptstyle \beta$} ; %
        \node[anchor=north] at (t1.south) {$\scriptstyle \alpha$} ; %
        \node[anchor=south] at (t.north) {$\scriptstyle x'$} ; %
        \node[anchor=north] at (s.south) {$\scriptstyle x$} ; %
        \node[anchor=south] at (t'.north) {$\scriptstyle y'$} ; %
        \node[anchor=north] at (s'.south) {$\scriptstyle y$} ; %
      }{%
        \& \joueur{t} \& \& \& \& %
        \joueur{t'}   \\
        \&   \& \& \canal{t2} \\
        \canal{t0} \&  \coupout{iota}{0} \& \&   \& \&  \coupin{iota'}{0} \\
        \&  \& \& \& \canal{t1} \\
        \& \joueur{s} \& \& \& \& \joueur{s'} 
      }{%
      }%
    \caption{Categories of elements for $[3]$, $\paraof{2}$, and $\tau_{1,1,3,2,3}$, with graphical representation}
\label{fig:tau}
\end{figure*} %
Returning to our explanation of $\C$ through categories of elements,
let us consider that of $\paraof{2}$. It is the poset generated by the
left-hand graph in the second row of Figure~\ref{fig:tau} (omitting
base objects for conciseness).  We think of it as a binary \agent ($l
t$) forking into two \agents ($l s$ and $r s$), and draw it as on the
right.  The equation $lt=rt$ ensures that $\paralof{2}$ and
$\pararof{2}$ are performed by the same \agent. The graphical
convention is that a black triangle stands for the presence of
$\id_{\forkof{2}}$, $l$, and $r$. Below, we represent just $l$ as a
white triangle with only a left-hand branch, and symmetrically for
$r$.  Furthermore, in all our pictures, time flows `upwards'.

Another category of elements, characteristic of the $\pi$-calculus, is
the one for synchronisation $\taunamcd$. The case $(n,a,m,c,d) =
(1,1,3,2,3)$ is the poset generated by the graph at the bottom left of
Figure~\ref{fig:tau}, which we will draw as on the right. %
The left-hand ternary \agent $x$ outputs its $3$rd channel, here
$\beta$, on its $2$nd channel, here $\alpha$. The right-hand unary
\agent $y$ receives the sent channel on its unique channel, here
$\alpha$. Both \agents have two occurrences, one before and one after
the \action, respectively marked as $x / x'$ and $y / y'$.  Both $x$
and $x'$ are ternary here, while $y$ is unary and $y'$, having gained
knowledge of $\beta$, is binary. There are actually three \actions
here, in the sense that there are three higher-dimensional elements.
The first is the output \action $\epsilon$ from $x$ to $x'$,
graphically represented as the middle point of %
\raisebox{.25em}{
\diagramme[ampersand replacement=\&,column sep=.5cm,inner sep=0.1pt]{inner sep=0pt}{
          \path[-] %
          (b) edge[very thick] (c) %
          (a) edge[-latex] (b) %
          ; %
        }{%
          \node[coordinate] (a){}; \& \node[coordinate] (b){}; \& \node[coordinate] (c){}; %
        }{%
        }%
}
(intended to evoke the point where $\beta$ enters channel $\alpha$).
The second is the input \action $\rho$ from $y$ to $y'$, graphically represented
as the middle point of 
\raisebox{.25em}{
\diagramme[ampersand replacement=\&,column sep=.5cm,inner sep=1pt]{inner sep=0pt}{
          \path[-] %
          (a) edge[very thick] (b) %
          (b) edge[-latex] (c) %
          ; %
        }{%
          \node[coordinate] (a){}; \& \node[coordinate] (b){}; \&
          \node[coordinate] (c){}; %
        }{%
        }%
      } %
      (where $\beta$ exits channel $\alpha$).  The third \action is
      the synchronisation itself, which `glues' the other two
      together, as represented by the squiggly line.

We leave the computation of other categories of elements as an
exercise to the reader. 
The remaining string diagrams 
are depicted 
in the top row of Figure~\ref{fig:stringmoves}, for $p = 2$ and
$(n,a,m,c,d) = (1,1,3,2,3)$. 
\begin{figure*}[t]
  \centering
  \begin{tabular}{c@{\,}c@{\,}c@{\,}c@{\,}c@{\,}c@{\,}c@{\,}c}
  \diagramme[stringdiag={.2}{.33}]{}{ }{%
    \& \joueur{t_1} \& \& \& \\ 
    \& \&   \&  \\
    \& \ \& \\
    \canal{t0} \& \& \coupparacreux{para} \& \& \canal{t1}
    \\ 
    \& \ \& \\
    \& \&  \\
    \& \& \joueur{s} \& \&
  }{%
    (t0) edge[-] (t_1) %
    (t1) edge[-,bend right=20] (t_1) %
    (t0) edge[-] (s) %
    (t1) edge[-] (s) %
    (s) edge[-] (para) %
    (para) edge[-] (t_1) %
  }
  &
%
  \diagramme[stringdiag={.2}{.33}]{}{ }{%
    \& \& \& \joueur{t_2} \& \\ 
    \& \&   \\
    \& \ \& \\
    \canal{t0} \& \& \coupparacreux{para} \& \& \canal{t1}
    \\ 
    \& \ \& \\
    \& \&  \\
    \& \& \joueur{s} \& \&
  }{%
    (t0) edge[-,bend left=20] (t_2) %
    (t1) edge[-] (t_2) %
    (t0) edge[-] (s) %
    (t1) edge[-] (s) %
    (s) edge[-] (para) %
    (para) edge[-] (t_2) %
  }
  &
%
  \diagramme[stringdiag={.3}{.5}]{baseline=($(iota.south)$)}{%
    \path[-] %
    (t2) edge (s) %
    (t1) edge (t) %
    (t0) edge (t) %
    (t2) edge (t) %
    (t) edge (iota.west) %
    (s) edge (iota.west) %
    (t0) edge[bend right=20] (s) %
    (t1) edge (s) %
    ; %
    \movepiout[1]{t1}{iota}{t2}{.4} %
    \foreach \x/\y in {s/t} \path[-] (\x) edge (\y) ; %
  }{%
    \& \& \joueur{t}  \\
    \&  \\
    \canal{t0} \& \& \coupout{iota}{0}  \& \canal{t2} \\
    \& \canal{t1} \\
    \& \& \joueur{s} }{%
  }%
  &
  \diagramme[stringdiag={.6}{\longueurfigun}]{baseline=($(in.south)$)}{
    \path[-] (a) edge (p) %
    (in) edge (p) edge (p') %
    edge[gray,very thin] (b) %
    (p') edge (a) edge (b) %
    ; %
    \movepiin[.5]{a}{in}{b}{.6} %
    \foreach \x/\y in {p/p',a/a} \path[-] (\x) edge (\y) ; %
  }{ %
    \& \joueur{p'} \& \\ %
    \canal{a} \& \coupout{in}{0} \& \canal{b} \\ %
    \& \joueur{p} 
  }{%
  } %
  &
%
  \diagramme[stringdiag={.6}{\longueurfigun}]{}{ \path[-] (a) edge
    (a) %
    edge (p) %
    (tick) edge[shorten <=-1pt] (p) edge[shorten <=-1pt] (p') %
    (p') edge (a) edge (b) %
    (b) edge (p) edge (b) %
    ; %
  }{ %
    \& \joueur{p'} \& \\ %
    \canal{a} \& \couptick{tick} \& \canal{b} \\ %
    \& \joueur{p} \& %
  }{%
  } &
%
  \diagramme[stringdiag={.6}{\longueurfigun}]{}{ \path[-] (a) edge
    (a) %
    edge (p) %
    (tick) edge[shorten <=-1pt] (p) edge[shorten <=-1pt] (p') %
    (p') edge (a) edge (b) %
    (b) edge (p) edge (b) %
    ; %
  }{ %
    \& \joueur{p'} \& \\ %
    \canal{a} \& \couptau{tick} \& \canal{b} \\ %
    \& \joueur{p} \& %
  }{%
  } &
    %
  \diagramme[stringdiag={.3}{.5}]{baseline=($(nu.center)$)}{%
    \path[-,draw] %
    (t1) edge (s) %
    (t1) edge (t) %
    (t0) edge (t) %
    (t2) edge (t) %
    (t) edge (nu) %
    (s) edge (nu) %
    (nu) edge[gray,very thin] (t2) %
    (t0) edge[bend right=20] (s) %
    ; %
  }{%
      
    \& \& \joueur{t} \& \&  \& \\
    \&  \&    \\
    \canal{t0} \& \& \coupnu{nu} \& \& \canal{t2}  \\
    \& \canal{t1} \\
    \& \& \joueur{s} \& }{%
  }%
  \\
  \diag(.2,.2){%
    {[p]} \\ {\forklp} \\ {[p]} %
  }{%
    (m-1-1) edge (m-2-1) %
    (m-3-1) edge (m-2-1) %
  }%
  &
  \diag(.2,.2){%
    {[p]} \\ {\forkrp} \\ {[p]} %
  }{%
    (m-1-1) edge (m-2-1) %
    (m-3-1) edge (m-2-1) %
  }%
  &
  \diag(.2,.2){%
    {[m]} \\ {\outmcd} \\ {[m]} %
  }{%
    (m-1-1) edge (m-2-1) %
    (m-3-1) edge (m-2-1) %
  }%
  &
  \diag(.2,.2){%
    {[n+1]} \\ {\inna} \\ {[n]} %
  }{%
    (m-1-1) edge (m-2-1) %
    (m-3-1) edge (m-2-1) %
  }%
  &
  \diag(.2,.2){%
    {[p]} \\ {\tickp} \\ {[p]} %
  }{%
    (m-1-1) edge (m-2-1) %
    (m-3-1) edge (m-2-1) %
  }%
  &
  \diag(.2,.2){%
    {[p]} \\ {\taup} \\ {[p]} %
  }{%
    (m-1-1) edge (m-2-1) %
    (m-3-1) edge (m-2-1) %
  }%
  &
  \diag(.2,.2){%
    {[p+1]} \\ {\nup} \\ {[p]} %
  }{%
    (m-1-1) edge (m-2-1) %
    (m-3-1) edge (m-2-1) %
  }%
  &
\end{tabular}
  \caption{Pictures and corresponding cospans 
    for $\paralp$, $\pararp$, $\outmcd$, $\inna$, $\tickp$, $\taup$, and $\nup$}
\label{fig:stringmoves}
\end{figure*}
%
The first two are \emph{\threads}, in the game semantical sense, of
the fork \action $\forkof{2}$ explained above. The next two, $\outmcd$
(for `output') and $\inna$ (for `input'), respectively are \threads
for the sender and receiver in a synchronisation \action. The $\taup$
\action is a silent, dummy action standard in $\pi$-calculus.  The
$\tickp$ \action is the special `tick' \action used for defining fair
testing equivalence. The last one, $\nup$, is a channel creation \action.

\subsection{From string diagrams to \actions}
In the previous section, we have defined our category of string
diagrams as $\Chatf$, and provided some intuition on its objects.  The
next step is to construct a bicategory whose objects are \positions,
and whose morphisms represent \traces. We start in this section by
defining in which sense higher-dimensional objects of $\C$ represent
\actions, and continue in the next one by explaining how to compose
\actions to form \traces. \Actions are defined in two stages:
\emph{seeds}, first, give their local form, their global form being
given by embedding into bigger \positions{}.

To start with, until now, our string diagrams contain no information
about the `flow of time', although we mentioned it informally in the
previous section. To add this information, for each string diagram $M$
representing \anaction, we define its initial and final
\positions, say $X$ and $Y$, and view the whole \action as a
cospan $Y \xto{s} M \xot{t} X$. We have taken care, in drawing our
pictures before, of placing initial \positions at the bottom, and
final \positions at the top.  So, e.g., the initial and final
\positions for the example synchronisation of
Figure~\ref{fig:tau} are as follows.
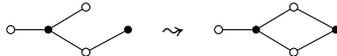
\begin{center}
  \diagramme[stringdiag={.3}{.5}]{
  }{%
    \path[-,draw] %
    (t) edge (t0) %
    edge (a) %
    edge (b) %
    (t') edge (a) %
    ; %
  }{%
    \& \& \canal{b} \\
    \canal{t0} \& \joueur{t} \& \& \joueur{t'} \\
    \& \& \canal{a} \\
  }{%
  }%
  \ {$\leadsto$} \
  \diagramme[stringdiag={.3}{.5}]{baseline=($(nu.center)$)}{%
    \path[-,draw] %
    (t) edge (t0) %
    edge (a) %
    edge (b) %
    (t') edge (a) %
    edge (b) %
    ; %
  }{%
    \& \& \canal{b} \\
    \canal{t0} \& \joueur{t} \& \& \joueur{t'} \\
    \& \& \canal{a} \\
  }{%
  }%
\end{center}
They map into (the representable presheaf over) $\tau_{1,1,3,2,3}$,
yielding the cospan $$Y \xto{s} \tau_{1,1,3,2,3} \xot{t} X.$$
We leave it to the reader to define, based on the above pictures, the
expected cospans for forking and synchronisation 
\begin{center}
      \diag(.2,.2){%
      {[p] \para [p]} \\ {\forkp} \\ {[p]} %
    }{%
      (m-1-1) edge (m-2-1) %
      (m-3-1) edge (m-2-1) %
    }%
    \hfil
    \diag(.2,.2){%
      {[m] \paraofij{c,d}{a,n+1} [n+1] } \\ {\taunamcd } \\ {
        [m] \paraofij{c}{a} [n] } %
    }{%
      (m-1-1) edge (m-2-1) %
      (m-3-1) edge (m-2-1) %
    }
\end{center}
plus the remaining ones specified in the bottom row of
\figurename~\ref{fig:stringmoves}.  Initial \positions are at the
bottom, and we use:
\begin{notation}
  We denote by $[m] \paraofij{a_1,\ldots,a_p}{c_1,\ldots,c_p} [n]$ the
  \position consisting of an $m$-ary \agent $x$ and an $n$-ary
  \agent $y$, quotiented by the equations $x \cdot s_{a_k} = y \cdot
  s_{c_k}$ for all $k \in p$. When both lists are empty, by
  convention, $m=n$ and the \agents share all channels in order.
\end{notation}
\begin{definition}
  These cospans are called \emph{seeds}. 
\end{definition} 

We now define \actions from seeds by embedding the latter into bigger
\positions. E.g., we allow a fork \action to occur in
\aposition with more than one \agent. 
\begin{definition}\label{def:interface}
  The \emph{interface} $I_F$ of a presheaf $F \in \Chat$ is $F(\star)
  \cdot \star$, the $F(\star)$-fold coproduct of $\star$ with itself,
  or in other words the \position consisting solely of $F$'s
  channels.  The interface of a seed $Y \xto{s} M \xot{t} X$ is $I_X$.
\end{definition}

Since channels occurring in the initial \position remain in the
final one, we have for each seed a cone from $I_X$ to the seed.  For
any morphism of \positions $I_X \to Z$,
pushing the cone along $I_X \to Z$ using the
universal property of pushout as in
    \begin{equation}
    \Diag(.02,.3){%
      \pbk{X}{X'}{Z} %
      \pullback[1.2em]{Z}{M'}{M}{draw,-} %
      \pullback{Z}{Y'}{Y}{draw,-} %
    }{%
      \&\& |(Y)| Y \& \&\& |(Y')| {Y'} \\
      \&\& \ \&\& \\
      \&\& |(M)| M \& \&\& |(M')| M'  \\
      |(I)| I_X \&\&\& |(Z)| Z \\
      \&\& |(X)| X \& \&\& |(X')| X' %
    }{%
      (Z) edge[] (X') %
      edge (M') %
      edge (Y') %
      (I) edge[] (X) %
      edge (Z) %
      edge (M) %
      edge (Y) %
      (Y) edge[fore] (Y') %
      (M) edge[fore] (M') %
      (X) edge (X') %
      (X') edge[dashed] (M') %
      (Y') edge[dashed] (M') %
      (X) edge[fore] (M) %
      (Y) edge[fore] (M) %
    }\label{morphismofinterfacedseeds}
  \end{equation}
yields a new cospan, say
$Y' \to M' \ot X'$. 

  \begin{definition}
    Let \emph{\actions} be all such pushouts of seeds.
  \end{definition}
  Intuitively, taking pushouts glues string diagrams together. Let us do a
  few examples.
  \begin{example}\label{ex:forkmove}
    The seed $[2]\para[2] \xto{[ls,rs]} \forkof{2} \xot{lt} [2]$ has
    as interface the presheaf $I_{[2]} = \star + \star$, consisting of
    two channels, say $a$ and $b$.  Consider the \position $[2] +
    \star$ consisting of \anagent $y$ connected to two channels $b'$ and $c$,
    plus an additional channel $a'$. Further consider the map $h
    \colon I_{[2]} \to [2]+\star$ defined by $a \mapsto a'$ and $b
    \mapsto b'$. The pushout
    \begin{center}
      \Diag{%
        \pbk{pi}{M'}{star} %
        }{%
         |(I2)| {I_{[2]}} \& |(star)| {[2]+\star} \\
         |(pi)| {\forkof{2}} \&|(M')| {M'} %
        }{%
          (I2) edge (star) edge (pi) %
          (pi) edge (M') %
          (star) edge (M') %
        } %
        \hfil \mbox{is} \hfil
        \begin{minipage}[c]{0.55\linewidth}
          \vspace*{-.5em} 
          \diagramme[stringdiag={.2}{.6}]{}{
            \node[diagnode,at= (c.south east)] {\ \ \ .} ; %
            \node[anchor=south] at (t_1.north) {$\scriptstyle x_1$}
            ; %
            \node[anchor=south] at (t_2.north) {$\scriptstyle x_2$}
            ; %
            \node[anchor=north] at (s.south) {$\scriptstyle x$} ; %
            \node[anchor=north] at (y.south) {$\scriptstyle y$} ; %
            \node[anchor=north] at (c.south) {$\scriptstyle c$} ; %
            \node[anchor=north] at (t0.south) {$\scriptstyle a=a'$}
            ; %
            \node[anchor=north] at (t1.south) {$\scriptstyle b=b'$}
            ; %
          }{%
            \& \& \joueur{t_1} \&  \& \joueur{t_2} \\
            \& \&   \&  \\
            \& \ \& \\
            \canal{t0} \& \& \& \couppara{para} \& \& \& \canal{t1} \&
            \& \joueur{y} \& \& \canal{c}
            \\ 
            \& \ \& \\
            \& \&  \\
            \& \& \& \joueur{s} \& \& \&
            \& 
          }{%
            (para) edge[-] (t_2) %
            (t1) edge[-,bend right=10] (t_2) %
            (t0) 
            edge[-] (s) %
            (t1) 
            edge[-] (s) %
            (s) edge[-] (para) %
            (y) edge[-] (t1) %
            edge[-] (c) %
            (para) edge[fore={.3}{.3},-] (t_1) (t0)
            edge[fore={.5}{.5},-,bend left=10] (t_2) %
            (t0) edge[-,bend left=15] (t_1) %
            (t1) edge[-,fore={1}{.5},bend right=10] (t_1) %
          }
        \end{minipage}
      \end{center}
      The meaning of such \anaction is that $x$ forks while $y$ is passive.
  \end{example}
  \begin{example}\label{ex:interface}
    Because we push along \emph{initial} channels, the interface of a
    seed may not contain all involved channels. E.g., in an input
    \action (not part of any synchronisation), the received channel
    cannot be part of the initial \position.
  \end{example}

\subsection{From \actions to \traces}\label{subsec:plays}
Having defined \actions, we now define their composition to yield our
bicategory $\Dv$ of \positions and \traces. Consider
$\Cospan{\Chatf}$, the bicategory which has as objects all finite
presheaves on $\C$, as morphisms $X \to Y$ all cospans $X \to U \ot
Y$, and obvious 2-cells.  Composition is given by pushout, and hence
is not strictly associative.
\begin{notation} By convention, the initial \position is the
  \emph{target} of the morphism in $\Cospan{\Chatf}$. We denote
  morphisms in $\Cospan{\Chatf}$ with special arrows $Y \proto X$;
  composition and identities are denoted with $\vrond$ and $\idv$,
  which matches the notation of pseudo double categories
  (Section~\ref{subsec:playgrounds}).
\end{notation}

\begin{definition}
  A \emph{trace} is any cospan in $\Chatf$ which is isomorphic to some
  finite, possibly empty composite of \actions in $\Cospan{\Chatf}$.
  Let $\Dv$ denote the subbicategory of $\Cospan{\Chatf}$ obtained by
  restricting to \positions, \traces, and 1-injective 2-cells.
\end{definition}

Thus, arrows $X \to Y$ in $\Dh$ denote embeddings of $X$ into $Y$ (up
to identification of channels), whereas arrows $Y \proto X$ in $\Dv$
denote \traces with $X$ initial and $Y$ final.  Intuitively,
composition in $\Dv$ glues string diagrams on top of each other, which yields
a truly concurrent notion of \trace: the only information retained in
\atrace about the order of occurrence of \actions is their causal
dependencies.

\begin{figure}[t]
  \centering
                  \diagramme[stringdiag={.2}{.4}]{baseline=(t0)}{
        \node[anchor=south] at (t_1.north) {$\scriptstyle x_1$} ; %
        \node[anchor=south] at (t_2.north) {$\scriptstyle x_2$} ; %
        \node[anchor=south] at (y_1.north) {$\scriptstyle y_1$} ; %
        \node[anchor=south] at (y_2.north) {$\scriptstyle y_2$} ; %
        \node[anchor=north] at (s.south) {$\scriptstyle x$} ; %
        \node[anchor=north] at (y.south) {$\scriptstyle y$} ; %
        \node[anchor=north] at (c.south) {$\scriptstyle c$} ; %
        \node[anchor=north] at (t0.south) {$\scriptstyle a=a'$} ; %
        \node[anchor=north] at (t1.south) {$\scriptstyle b=b'$} ; %
  }{%
     \& \& \joueur{t_1} \&  \& \joueur{t_2} \& \& \& \& \joueur{y_1} \& \& \joueur{y_2} \\
    \& \&   \&  \\
    \& \ \& \\
    \canal{t0} \& \& \& \couppara{para} \& \& \& \canal{t1} \& \&  \& \couppara{para'} \&  \& \& \canal{c} \\
    \& \ \& \\
    \& \&  \\
    \& \& \& \joueur{s} \& \& \& \& \& \& \joueur{y} 
  }{%
    (t1) edge[-,bend right=10] (t_2) %
    (t0) 
    edge[-] (s) %
    (t1) 
    edge[-] (s) %
    (para) edge[-] (t_2) %
    (s) edge[-] (para) %
    (y) edge[-] (t1) %
     edge[-] (c) %
    (para) edge[fore={.3}{.3},-] (t_1)
    (t0) edge[fore={.5}{.5},-,bend left=10] (t_2) %
    (t0) edge[-,bend left=15] (t_1) %
    (t1) edge[-,fore={1}{.5},bend right=10] (t_1) %
    (c) edge[-,bend right=10] (y_2) %
    (para') edge[-] (y_2) %
    (y) edge[-] (para') %
    (y) edge[-] (t1) %
     edge[-] (c) %
    (para') edge[fore={.3}{.3},-] (y_1) %
    (t1) edge[fore={.5}{.5},-,bend left=10] (y_2) %
    (t1) edge[-,bend left=15] (y_1) %
    (c) edge[-,fore={1}{.5},bend right=10] (y_1) %
  }  
\hfill
      \diagramme[stringdiag={.6}{.8}]{baseline=(a)}{%
        \framenode{a} %
        \framenode{a'} %
        \circlenode{c} %
        \circlenode{c'} %
        \receive{b}{i}{a'}{.4}{} %
        \envoie{a}{o}{b}{.4}{} %
        \envoie{c}{o'}{a'}{.4}{} %
        \receive[fore={.1}{.1}]{a'}{i'}{c'}{.4}{} %
        \node[below=1pt] at (a.south) {$\scriptstyle a$} ; %
        \node[below=1pt] at (b.south) {$\scriptstyle b$} ; %
        \node[below=1pt] at (c.south) {$\scriptstyle c$} ; %
        \node[below=1pt] at (a'.south) {$\scriptstyle a$} ; %
        \node[right=1pt] at (c') {$\scriptstyle c$} ; %
        \node[below=1pt] at (x.south) {$\scriptstyle x$} ; %
        \node[below=1pt] at (y.south) {$\scriptstyle y$} ; %
        \node[below=1pt] at (z.south) {$\scriptstyle z$} ; %
        \node[above left,inner sep=1pt] at (y') {$\scriptstyle y'$} ; %
        \node[left] at (y'') {$\scriptstyle y''$} ; %
        \twocell{o}{b}{i}{}{
          decorate,decoration={snake,amplitude=.3mm,segment length=1mm},bend left=50}
        \twocell[.4][.25]{o'}{a'}{i'}{}{
          decorate,decoration={snake,amplitude=.3mm,segment length=1mm},bend right=30}
      }{%
        \& \& \& \joueur{y''} \\
        \& \& \& \coupin{i'}{0} \& \canal{c'} \\
        \& \joueur{x'} \& \& \joueur{y'} \& \& \joueur{z'} \\
    \canal{a} \& \coupout{o}{0} \& \canal{b} \& \coupin{i}{0} \& \canal{a'} \& \coupout{o'}{0} \& \canal{c} \\
    \& \joueur{x} \& \& \joueur{y} \& \& \joueur{z} %
  }{%
    (a) edge[-] (x) %
    edge[-] (x') %
    (b) edge[-] (x) %
    edge[-] (x') %
    edge[-] (y) %
    edge[-] (y') %
    edge[-] (y'') %
    (a') edge[-] (y') %
    edge[-,bend right=10] (y'') %
    edge[-] (z) %
    edge[-] (z') %
    (c) edge[-] (z) %
    edge[-] (z') %
    (c') edge[-] (y'') %
    (i) edge[-] (y) %
    edge[-] (y') %
    (i') edge[-] (y') %
    edge[-] (y'') %
    (o) edge[-] (x) %
    edge[-] (x') %
    (o') edge[-] (z) %
    edge[-] (z') %
  }
  \caption{Example \traces}
  \label{fig:plays}
\end{figure}
\begin{example}\label{exa:causal}
  Composing the \action of Example~\ref{ex:forkmove} with a forking
  \action by $y$ yields the first string diagram of
  Figure~\ref{fig:plays}, which shows that the ordering between remote
  \actions is irrelevant.  To illustrate how composition retains
  causal dependencies between \actions, consider the second string
  diagram.  It is unfolded for readability: one should identify both
  framed nodes, resp.\ both circled ones.  In the initial
  \position, there are channels $a,b$, and $c$, and three \agents
  $x(a,b)$, $y(b)$, and $z(a,c)$ (channels known to each \agent are in
  parentheses).  In a first \action, $x$ sends $a$ on $b$, and $y$
  receives it. In a second \action, $z$ sends $c$ on $a$, and the
  avatar $y'$ of $y$ receives it. The second \action is enabled by the
  first, by which $y$ gains knowledge of $a$.
\end{example}

Before going on to construct the base double category for our
playground, let us observe the following two basic facts about
\traces.
\begin{lem}\label{lem:s:t:monos}
  For any \trace $Y \xto{s} U \xot{t} X$,
  $s$ and $t$ are monos.
\end{lem}
\begin{proof}
  We proceed by induction on the number of \actions involved in any
  decomposition of $U$. The base case is trivial.  For the induction
  step, because composition of cospans is by pushout, the result
  follows from the induction hypothesis, stability of monos under
  pushout and composition, and the fact that the result holds for
  \actions. The latter in turn follows from the fact that monos are
  stable under pushout (again!), the pushout lemma, and the
  fact that the result holds for seeds, which holds by case inspection.
\end{proof}

\begin{lem}\label{lem:s:star:iso}
  For any \trace $Y \xto{s} U \xot{t} X$, $s_\star$ is an isomorphism.
\end{lem}
\begin{proof}
  Similar to the previous proof.
\end{proof}
The intuition behind the last lemma is that no channel is forgotten
during the play.

\subsection{The main double category}
At last, we define the base double category $\D$ of our
playground for the $\pi$-calculus.  It is a sub-double category of a
double category of cospans in~$\Chat$.

Consider the double category $\D^0$ with 
\begin{itemize}

\item \positions as objects, 

\item horizontal morphisms $X \to Y$ given by all natural
transformations $h \colon X \to Y$,

\item vertical morphisms $X \proto Y$ given by cospans
  $X \xto{s} U \xot{t} Y$ in $\Chat$,

\item and double cells $U \to V$ given by commuting diagrams

  \begin{equation}
    \label{generaldoublecell}
      \diag{%
       |(X')| X' \&|(Y')| Y' \\
       |(U)| U \&|(V)| V \\
       |(X)| X \&|(Y)| Y\rlap{.} %
      }{%
        (X') edge[labell={s_U}] (U) %
        edge[labela={k}] (Y') %
        (X) edge[labell={t_U}] (U) %
        edge[labelb={h}] (Y) %
        (U) edge[labela={l}] (V) %
        (Y') edge[labelr={s_V}] (V) %
        (Y) edge[labelr={t_V}] (V) %
      }
  \end{equation}
\end{itemize} 

\begin{defi} Let $\D$ denote the sub-double
  category of $\D^0$ obtained by restricting
  \begin{itemize}
  \item vertical morphisms to traces, 
  \item horizontal morphisms to 1-injective maps,
  \item double cells to diagrams~\eqref{generaldoublecell} in which $k,l$, and $h$ are 1-injective.
  \end{itemize}
\end{defi}
\begin{prop}\label{prop:D:D0}
  $\D$ indeed forms a sub-double category of $\D^0$, i.e., is closed
  under all composition operations.
\end{prop}
\begin{proof}
  The only non-obvious point is that double cells in $\D$ are closed
  under vertical composition, and in particular that the middle
  component of the composite is 1-injective. This follows from
  Lemma~\ref{upper squares all pullbacks} and Corollary~\ref{cor:cube}
  below.
\end{proof}

\begin{defi}\label{def:V}
Let $\vertical_0$ denote the set of `$t$'-legs (i.e., lower legs) of seeds.
\end{defi}

\begin{lem}\label{upper squares all pullbacks}
 For any morphism~\eqref{generaldoublecell}
in $\D^0_H$, if $U$ and $V$ are \traces, then the upper square is a
pullback.  \end{lem}

\begin{proof} 
  For any \trace $Y \xto{s} P \xot{t} X$ and $n \in \Nat$, $Y[n]$
  consists of all elements of $P[n]$ which are not in the image of
  (the action of) any map in $\vertical_0$.

  Now, consider any double cell as in~\eqref{generaldoublecell}.
  Because $s_V$ is monic, $U \times_V Y'$ may be chosen to be just
  $l^{-1}(Y') \subseteq U$.  By Lemma~\ref{lem:s:t:monos} and standard
  cancellation properties of monos, the mediating arrow $X' \to U
  \times_V Y'$ is mono. To show that it is epi, we proceed
  pointwise. Over $\star$, the result follows from $s_U$ and $s_V$
  being isomorphisms (Lemma~\ref{lem:s:star:iso}). Over $[n]$, if $x
  \in (U \times_V Y') [n]$ then $x \in U[n]$, and $l(x) \in V[n]$ is
  not in the image of (the action of) any $t \in \vertical_0$. But if
  there existed $y$ such that $y \cdot t = x$, then by naturality we
  would have $l(y) \cdot t = l(x)$, contradicting the latter. Any
  natural transformation being both epi and mono is an isomorphism,
  hence the result. \end{proof}

\begin{lem}
  In $\Set$, consider any cube 
  \begin{center}
    \Diag(.4,.6){%
      \pbk{A}{C}{B} %
      \pbk{I'}{I}{B} %
      \pbk{A'}{C'}{B'} %
      \path[->,draw] %
      (I) edge[fore] (A) %
      ; %
    }{%
      |(I)| I \& \& \&|(B)| B \\
      \&|(A)| A \& \& \&|(C)| C \\
      \& \& \& \& \& \\
      |(I')| I' \& \& \&|(B')| B' \\
      \&|(A')| A' \& \& \&|(C')| C'\rlap{,} \\
        \& \& \& \& \& %
    }{%
      (I) edge (I') edge (B) %
      (I') edge (A') edge (B') %
      (B') edge (C') %
      (B) edge (C) %
      edge (B') (A') edge[fore] (C') %
      (A) edge[fore] (C) %
      edge[fore] (A') %
      (C) edge[labelr={f}] (C') %
    }
  \end{center}
  with the marked pushouts and pullback, and with all arrows mono
  except perhaps $f$.  Then, $f$ is also mono and the front square is
  also a pullback.
\end{lem}
\begin{proof}
  Any such cube is naturally isomorphic to some cube of the shape
  \begin{center}
    \Diag(.4,.6){%
      \pbk[3em]{A}{C}{B} %
      \pbk{I'}{I}{B} %
      \pbk[3em]{A'}{C'}{B'} %
      \path[->,draw] %
      (I) edge[fore] (A) %
      ; %
    }{%
      |(I)| I \& \& \&|(B)| I+R \\
      \&|(A)| I+Y \& \& \&|(C)| I+Y+R \\
      \& \& \& \& \& \\
      |(I')| I+X \& \& \&|(B')| I+X+S \\
      \&|(A')| I+X+Z \& \& \&|(C')| I+X+Z+S\rlap{,} \\
        \& \& \& \& \& %
    }{%
      (I) edge[labell={\injl}] (I') edge (B) %
      (I') edge[labelbl={\injl}] (A') edge (B') %
      (B') edge[labelar={}] (C') %
      (B) edge (C) %
      edge[labellat={\injl+k}{.7}] (B') (A') edge[fore] (C') %
      (A) edge[fore] (C) %
      edge[fore,labellat={I+h}{.3}] (A') %
      (C) edge[labelr={f=I+h+k}] (C') %
    }
  \end{center}
  the only non-trivial point being that the map $I+R \to I+X+S$ has
  the given shape. But this is because we know that its pullback along
  $I+X \to I+X+S$ is $\injl$, so the image of $R$ has to lie in $S$.
\end{proof}
\begin{cor}\label{cor:cube}
  Consider any cube   
  \begin{center}
    \Diag(.4,.6){%
      \pbk{A}{C}{B} %
      \pbk{I'}{I}{B} %
      \pbk{A'}{C'}{B'} %
      \path[->,draw] %
      (I) edge[fore] (A) %
      ; %
    }{%
      |(I)| X \& \& \&|(B)| B \\
      \&|(A)| A \& \& \&|(C)| C \\
      \& \& \& \& \& \\
      |(I')| X' \& \& \&|(B')| B' \\
      \&|(A')| A' \& \& \&|(C')| C'\rlap{,} \\
        \& \& \& \& \& %
    }{%
      (I) edge (I') edge (B) %
      (I') edge (A') edge (B') %
      (B') edge (C') %
      (B) edge (C) %
      edge (B') (A') edge[fore] (C') %
      (A) edge[fore] (C) %
      edge[fore] (A') %
      (C) edge[labelr={f}] (C') %
    }
  \end{center}
  in $\Chat$ in which all arrows except perhaps $f$ are 1-injective,
  and the marked squares are pushouts, resp.\ pullbacks. Then $f$ is
  also 1-injective.
\end{cor}
\begin{proof}
  We proceed pointwise. On any object $C$ of dimension $>0$, we obtain
  a diagram in sets for which the lemma applies.
\end{proof}

\section{Codomain is a fibration}\label{sec:fib}
In this section, we prove that the double category $\D$ of
\traces{} constructed in the previous section satisfies the primary
axiom for playgrounds, namely that the codomain functor $\DH \to \Dh$
is a fibration.  We proceed as follows.  We first define a (strong)
factorisation system on $\Chat$ (Section~\ref{subsec:facto}), from
which we derive in Section~\ref{subsec:D1} an intermediate
sub-double category $\D \into \D^{1} \into \D^{0}$. We further show
that by the properties of factorisation, the codomain functor
$\D^{1}_H \to \D^{1}_h$ is a fibration. Finally, we want to show that
$\DH \to \Dh$ is a fibration by proving that \traces{} are stable
under relevant cartesian lifting in $\D^1$, i.e., cartesian liftings
of any \trace{} along any morphism in $\Dh$ are again in $\D_H$. We
first check this for seeds, by case analysis, in
Section~\ref{subsec:seeds}.  In order to generalise to \actions
$M\colon Y \proto X$, the basic idea is to
\begin{enumerate}
\item decompose $X$ as a pushout
of $X_0$, where the generating seed $M_0$ takes place, and $Z$, which is
passive;
\item decompose the morphism along which we want to restrict, say
  $h\colon X' \to X$, accordingly, say as $h_0\colon X'_0 \to X_0$ and
  $h_Z\colon Z' \to Z$;
\item restrict $M_0$ along $h_0$ to obtain $P'_0$;
\item \label{lab:recompose} recompose \atrace{} $P'$ from $P'_0$ and $Z'$;
\item check that $P'$ admits a cartesian morphism to $M$.
\end{enumerate}
Step~\eqref{lab:recompose} is non-trivial, so we devote
Section~\ref{subsec:opliftings} to it.  It works as a kind of formal
opposite to restriction, as we essentially lift $P'_0$ along
$h_0\colon X'_0 \to X'$. We call this an \emph{oplifting} of $P'_0$
along $h_0$, by analogy with lifting in opfibrations.  However,
opliftings do not enjoy the relevant universal property
(opcartesianness, which is dual to cartesianness). Instead, we find
that opliftings are in fact cartesian!  In
Section~\ref{subsec:actions}, we use opliftings to show that
\actions{} are stable under restriction, following the above plan.
Finally, we extend the result to arbitrary \traces in
Section~\ref{subsec:traces}.

\subsection{A factorisation system}\label{subsec:facto}
Let us start by defining the (strong) factorisation system
$(\vertical,\horizontal)$ on $\Chat$, on which the intermediate
sub-double category $\D^1$ will be based. The idea is that all three
components of cartesian morphisms in $\DH$ are in $\horizontal$, while
$t$-legs of vertical morphisms are in $\vertical$. The cartesian
lifting of any $V\colon Y' \proto Y$ as in~\eqref{generaldoublecell}
along any $h \colon X \to Y$ is then given by factoring $t_V \circ h$
as $l \circ t_U$ with $t_U \in \vertical$ and $l \in \horizontal$ to
obtain
\begin{equation}
      \Diag(.5,.7){%
      \pbk{m-2-1}{m-1-1}{m-1-2} %
    }{%
      X' \& Y' \\
      U \& V \\
      X \& Y, %
    }{%
      (m-1-1) edge[dashed,labela={k}] (m-1-2) %
      (m-2-1) edge[dashed,labelo={l}] (m-2-2) %
      (m-3-1) edge[labelb={h}] (m-3-2) %
      (m-1-1) edge[dashed,labell={s_U}] (m-2-1) %
      (m-1-2) edge[labelr={s_V}] (m-2-2) %
      (m-3-1) edge[dashed,labell={t_U}] (m-2-1) 
      (m-3-2) edge[labelr={t_V}] (m-2-2) %
    } \label{eq:facto}%
\end{equation}
where the upper square is a pullback.

We recall from Definition~\ref{def:V} that $\vertical_0$ denotes the
set of `$t$'-legs (i.e., lower legs) of seeds.  Following
Bousfield's~\cite{Bousfield} construction of `cofibrantly generated'
factorisation systems, we define $\horizontal = \vertical_0^{\bot}$ to
be the class of maps $f$ such that for any $t \in \vertical_0$ and
commuting square $(u,v) \colon t \to f$ in $\Chat^{\mathrel{\to}}$,
there exists a unique filler $h$ making the following diagram commute:
\begin{center}
    \diag{%
    X \& X' \\
    Y \& Y'\rlap{.} %
  }{%
    (m-1-1) edge[labelu={u}] (m-1-2) %
    edge[labell={t}] (m-2-1) %
    (m-2-1) edge[labeld={v}] (m-2-2) %
    (m-1-2) edge[labelr={f}] (m-2-2) %
    (m-2-1) edge[dashed,labelal={h}] (m-1-2) %
  }
\end{center}
In this situation, one says that $f$ is \emph{right-orthogonal} to $t$, and $t$ is \emph{left-orthogonal}
to $f$, which is denoted by $t \ortho f$.

We finally define $\vertical = {{}^{\bot}}\horizontal$ to consist of
all maps which are left-orthogonal to any map in $\horizontal$.  Of
course, we have $\vertical_0 \subseteq \vertical$.  The following is
an application of~\cite[Theorem 4.1]{Bousfield}:
\begin{prop}\label{prop:facto}
  The pair $(\vertical,\horizontal)$ forms a factorisation system.
\end{prop}
What does that mean? Here is a modern definition~\cite{FK}:
\begin{defi}\label{def:facto}
  The classes of maps $\vertical$ and $\horizontal$ form a
  factorisation system iff $\vertical = \botleft{\horizontal}$,
  $\botright{\vertical} = \horizontal$, and any arrow factors as $h
  \rond v$ with $h \in \horizontal$ and $v \in \vertical$.
\end{defi}
In the case where $\horizontal = \botright{\vertical_0}$ and
$\vertical = \botleft{\horizontal}$, Bousfield proves that any map in
$\Chat$ admits a factorisation using a transfinite construction (a
so-called \emph{small object} argument). But here we will only need
factorisations of particular morphisms, which we will actually be able
to calculate by hand.  Bousfield's results include:

\begin{lem}\label{vertical stable under pushout}
  $\vertical$ is stable under pushout and composition, contains all
  isomorphisms, and enjoys the right cancellation property, i.e., if
  $v \in \vertical$ and $f v \in \vertical$, then $f \in \vertical$.

  $\horizontal$ is stable under pullback and composition, contains all
  isomorphisms, and enjoys the left cancellation property, i.e., if
  $h \in \horizontal$ and $h f \in \horizontal$, then $f \in \horizontal$.
\end{lem}
\begin{rem}
  Stability under pushout is ambiguous here: we mean 
  that for any pushout
  \begin{center}
    \Diag{%
      \pbk{m-2-1}{m-2-2}{m-1-2} %
    }{%
      X \& Y \\
      Z \& T\rlap{,} %
    }{%
      (m-1-1) edge[labela={}] (m-1-2) %
      edge[labell={v}] (m-2-1) %
      (m-2-1) edge[labelb={}] (m-2-2) %
      (m-1-2) edge[labelr={v'}] (m-2-2) %
    }
  \end{center}
  if $v \in \vertical$, then $v' \in \vertical$. Stability under
  pullback is defined dually.
\end{rem}

\subsection{A first `fibred' double category}\label{subsec:D1}
We now make concrete the idea evoked in the previous section, of using
our factorisation system to obtain a codomain fibration.  Consider the
sub-double category $\D^1$ of $\D^{0}$ obtained by restricting
vertical morphisms to cospans $X \xto{s} U \xot{t} Y$ with $t \in
\vertical$. Its vertical morphisms are stable under composition and
contain identities by Lemma~\ref{vertical stable under pushout}, i.e.:
\begin{lem}
  $\D^1$ forms a sub-double category of $\D^0$.
\end{lem}

\begin{lem}
\Traces are in $\D^1_v$, i.e., we have $\D \subseteq \D^{1}$.
\end{lem}
\begin{proof}
 By Lemma~\ref{vertical stable under pushout}.
\end{proof}

The main interest of introducing $\D^{1}$ is:
\begin{lem}\label{cartesiandoublecells} The codomain functor $\cod \colon \D^1_H \to \D^1_h$ is a
  fibration in which a double cell~\eqref{generaldoublecell} is
  cartesian iff $l \in \horizontal$ and the upper square is a
  pullback.  \end{lem}

\begin{proof} Let us show that the lifting candidate computed
  in~\eqref{eq:facto} is cartesian. Indeed, consider any double
  cell~\eqref{eq:facto}, and any morphism from some vertical morphism
  $Z' \to W \ot Z$ to $V$ whose bottom component factors through
  $h \colon X \to Y$. By unique lifting in $(\vertical, \horizontal)$,
  we obtain a unique dashed arrow making
\begin{center}
  \diag{%
    Z \& X \& U \\
    \& Y \\
    W \& \& V %
  }{%
    (m-1-1) edge (m-1-2) %
    edge (m-2-2) %
    edge (m-3-1) %
    (m-2-2) edge (m-3-3) %
    (m-1-2) edge (m-1-3) %
    edge (m-2-2) %
    (m-1-3) edge (m-3-3) %
    (m-3-1) edge (m-3-3) %
    (m-3-1) edge[dashed,bend right=20] (m-1-3) %
  }
\end{center}
commute. We finally obtain the desired arrow $Z' \to X'$ by universal
property of pullback.  Conversely, any cartesian double cell, being
isomorphic to such a lifting, satisfies the conditions.
\end{proof}

As a final observation, let us record:
\begin{lem}\label{lem:Dh:H}
  Any morphism in $\Dh$ is automatically in $\horizontal$.
\end{lem}
\begin{proof}
  Indeed, consider any $h\colon X \to Y$ in $\Dh$. There cannot be any
  commuting square
\begin{center}
  \diag{%
    X_0 \& X \\
    U \& Y %
  }{%
    (m-1-1) edge[labela={u}] (m-1-2) %
    edge[labell={t}] (m-2-1) %
    (m-2-1) edge[labelb={v}] (m-2-2) %
    (m-1-2) edge[labelr={h}] (m-2-2) %
  }
\end{center}
with $t \in \vertical_0$, because $U$ is a representable presheaf of
dimension $> 1$ and $Y$ has dimension $\leq 1$, so there cannot be any
$v\colon U \to Y$.
\end{proof}

\subsection{Restriction of seeds}\label{subsec:seeds}
We now show that restrictions of seeds (in the sense of $\D^1$) are
\traces.

\begin{lem}\label{lem:restriction:seeds} 
  Consider any diagram $X' \xto{h} X \xto{t} M$, where $t \in
  \vertical_0$ and $h \in \Dh(X',X)$. Its factorisation $X' \xto{t'}
  M' \xto{h'} M$ with $t' \in \vertical$ and $h' \in \horizontal$ is
  such that $h'$ is 1-injective and the obtained restriction is a
  \trace of length at most 2.  If $X'$ is an \emph{individual}, i.e.,
  a \position of shape $[n]$, then it is isomorphic to some seed.
  If $X'$ is an interface, then the restriction is an
  equivalence (in $\Dv$).
\end{lem}
\begin{proof} We proceed by case analysis. In each case, one has to
  check that $h'$ is 1-injective, that $X'$ individual implies $t'$
  seed, that $X'$ interface implies that $t'$ is an isomorphism, and
  that the upper leg of the obtained cospan is as expected: this is
  routine so we mention it here once and for all. 

  Let us first treat the case where $M = \yoneda_c$, for $c$ not of
  the shape $\taunimjk$. Then, we have $X = [n]$ for some $n$. If
  $id_c \in \im(th)$, then $X' \iso [n] + I$ for some interface $I$
  (since $h$ is 1-injective). Consider the diagram
\begin{center}
  \diag{%
  M + I\& M \\
  {[n]+I} \& {[n].} %
  }{%
(m-1-1) edge[labelu={[id,tk]}] (m-1-2) %
edge[<-,labell={t + id_I}] (m-2-1) %
(m-2-1) edge[labeld={h=[id,k]}] (m-2-2) %
(m-1-2) edge[<-,labelr={t}] (m-2-2) %
  }
\end{center}
The map $t+id_I$ is in $\vertical$ by Lemma~\ref{vertical stable under pushout}, so we just 
have to prove that $[id,tk]$ is in $\vertical_0^{\bot}$, which is a simple verification.

If now $id_c \notin \im(h)$, then $X'$ is an interface, and the relevant factorisation is
\begin{equation}
  \diag{%
    X' \& M \\
    X' \& {[n],} %
  }{%
(m-1-1) edge[labelu={t \rond h}] (m-1-2) %
edge[<-,labell={\id}] (m-2-1) %
(m-2-1) edge[labeld={h}] (m-2-2) %
(m-1-2) edge[<-,labelr={t}] (m-2-2) %
  }
\label{X'isaninterface}
\end{equation}
because $t \circ h$ is easily checked to be in $\vertical_{0}^{\bot}$.

The case of $\taunimjk$ is a bit more complicated.  Here, $t$ is
actually $t_0 = [\rho t , \sender t]$. First of all, if $X'$ is an
interface, then we obtain a factorisation analogous
to~\eqref{X'isaninterface}.  Consider now the case where $\im (h)$
contains both agents of $[n] \paraij [m]$. Let $x$ denote the $n$-ary
one and $y$ denote the $m$-ary one (in $X'$).  If $x \cdot s_i = y
\cdot s_j$, then $X' = ([n] \paraij [m]) + I$ for some interface $I$ and
the required factorisation is easily seen to be
  \begin{center}
    \diag{%
      \taunimjk + I \& \taunimjk \\
      ([n] \paraij [m]) + I \& {[n] \paraij [m].} %
    }{%
(m-1-1) edge[labelu={[\id,t_0 k]}] (m-1-2) %
edge[<-,labell={t_0 + \id_I}] (m-2-1) %
(m-2-1) edge[labeld={h=[\id,k]}] (m-2-2) %
(m-1-2) edge[<-,labelr={t_0}] (m-2-2) %
    }
  \end{center}

  Consider now the case where $X'$ still contains both agents but $x
  \cdot s_i \neq y \cdot s_j$.  Then $X' = [n] + [m] + I$ for some
  interface $I$, and the required factorisation is
\begin{center}
  \diag{%
  \iotani + \outmjk + I \& \taunimjk \\
  {[n] + [m] + I} \& {[n] \paraij [m]\rlap{.}}
  }{%
    (m-1-1) edge[labelu={[\rho,\sender,t_0 k]}] (m-1-2) %
    edge[<-,labell={t + t + \id_I}] (m-2-1) %
    (m-2-1) edge[labeld={[x,y,k]}] (m-2-2) %
    (m-1-2) edge[<-,labelr={t_0}] (m-2-2) %
  }
\end{center}

The only non-trivial point here is to show that $[\rho,\sender,t_0 k]$
is in $\vertical_{0}^{\bot}$, which easily reduces to showing that there is
no commuting square
  \begin{center}
    \diag{%
      [n] \paraij [m] \&  \iotani + \outmjk + I \\
	\taunimjk \&  \taunimjk,
    }{%
(m-1-1) edge[labelu={u}] (m-1-2) %
edge[labell={t_0}] (m-2-1) %
(m-2-1) edge[labeld={v}] (m-2-2) %
(m-1-2) edge[labelr={[\rho,\sender,t_0 k]}] (m-2-2) %
    }
  \end{center}
  which is true because there is no such $u$. 

The cases where $X'$ only contains one agent of $[n] \paraij [m]$
are similar to the latter case: if it contains $x$ then the
factorisation is through $\iotani$, and otherwise it is through
$\outmjk$.  \end{proof}

\subsection{Opliftings}\label{subsec:opliftings}
We now aim at extending Lemma~\ref{lem:restriction:seeds} from seeds
to \actions.  Consider an arbitrary \action $Y \to M \ot X$, obtained
by pushing some seed $Y_0 \to M_0 \ot X_0$ along $I_{X_0} \to
Z$. Consider now a morphism $h\colon X' \to X$ in $\Dh$, along which
we wish to restrict $M$. As explained at the beginning of
Section~\ref{sec:fib}, our strategy is to consider the pullback
\begin{center}
  \Diag{%
    \pbk{m-2-1}{m-1-1}{m-1-2} %
  }{%
    X'_0 \&  X_0 \\
    X' \&  X\rlap{.} %
  }{%
    (m-1-1) edge[labelu={}] (m-1-2) %
    edge[labell={}] (m-2-1) %
    (m-2-1) edge[labeld={}] (m-2-2) %
    (m-1-2) edge[labelr={}] (m-2-2) %
  }
\end{center}
Indeed, the part of $X$ really concerned by $M$ is the image of $X_0$,
so we would like to first restrict $M_0$ to $X'_0$ using
Lemma~\ref{lem:restriction:seeds}, and then extend it to $X'$. The
present section is devoted to constructing the second step, the
extension of $\restr{(M_0)}{X'_0}$ to $X'$, which we call an
\emph{oplifting} of $\restr{(M_0)}{X'_0}$ along $X'_0 \to X'$.

In the general case, we want to compute the oplifting of $Y \to U \ot
X$ along $h \colon X \to X'$. Because $h$ is 1-injective, we
can complete the solid part of
\begin{center}
  \Diag{%
    \pbk{m-2-1}{m-2-2}{m-1-2} %
  }{%
    I_X \& Z \\
    X \& X' %
  }{%
(m-1-1) edge[dashed,labelu={}] (m-1-2) %
edge[into,labell={\subseteq}] (m-2-1) %
(m-2-1) edge[labeld={h}] (m-2-2) %
(m-1-2) edge[dashed,into,labelr={\subseteq}] (m-2-2) %
  }
\end{center}
into a pushout. Indeed, we take $Z(\star) = X'(\star)$, $Z[n] =
X'[n] \setminus \im(h_{[n]})$ for all $n$, and $I_X \to Z$ is uniquely
determined by $h_{\star}$.  In passing, we have:
\begin{lem}\label{lem:unique:pushout}
  This pushout is uniquely determined up to canonical isomorphism by
  $h$ alone.
\end{lem}
Then, we observe that, because $U$ is a \trace, $I_X \to X \to U$
factors through $Y \to U$. 
\begin{defi}
  Let the \emph{oplifting of $U\colon Y \proto X$ along
    $h\colon X \to X'$} be the cospan obtained as
  in~\eqref{morphismofinterfacedseeds} by pushing $U$ along
  $I_X \to Z$.
\end{defi}
All horizontal maps are 1-injective by construction, and we have:
  \begin{lem}
 The obtained cospan $Y' \to U' \ot X'$ is a trace.
  \end{lem}
\begin{proof}
  We start by showing that $U'$ is an action if $U$ is. So assume $U$
  is obtained by pushing a seed $Y_0 \to M_0 \ot X_0$ along some
  $I_{X_0} \to Z_0$. Then, $Z_0 \to X$ is surjective on $\star$
  because $I_{X_0} \to X_0$ is and epis are stable under pushout in
  presheaf categories.  Thus, $I_X \to X$ factors through $Z_0 \to
  X$. Let $Z''$ denote the pushout
  \begin{center}
    \Diag{%
      \pbk{m-2-1}{m-2-2}{m-1-2} %
    }{%
      I_X \& Z_0 \\
      Z \& Z'' \rlap{.} %
    }{%
      (m-1-1) edge[labela={}] (m-1-2) %
      edge[labell={}] (m-2-1) %
      (m-2-1) edge[labelb={}] (m-2-2) %
      (m-1-2) edge[labelr={}] (m-2-2) %
    }
  \end{center}
  By the pushout lemma, $Y' \to U' \ot X'$ is isomorphic to the cospan
  obtained by pushing $Y \to U \ot X$ along $Z_0 \to Z''$. By the
  pushout lemma again, it is isomorphic to the cospan obtained by
  pushing $Y_0 \to M_0 \ot X_0$ along $I_{X_0} \to Z_0 \to Z''$.
  Thus, it is indeed an action.

  We now prove the general case by induction on $\card{U}$.  This is
  trivial if $U$ is isomorphic to an identity.  If now $U$ is a
  composite $Z \xproto{V} Y \xproto{M} X$, then we compute the
  oplifting of $M$ along $h$ to obtain a double cell, say
\begin{center}
  \diag{%
    Y \& Y' \\
    X \& X'\rlap{,} %
  }{%
    (m-1-1) edge[labela={h'}] (m-1-2) %
    edge[pro,twol={M}] (m-2-1) %
    (m-2-1) edge[labelb={h}] (m-2-2) %
    (m-1-2) edge[pro,twor={M'}] (m-2-2) %
    (l) edge[cell=.3, labela={\alpha_M}] (r) %
  }%
\end{center}
by pushing $M$ along some morphism $I_X \to Z_1$ making
  \begin{center}
  \Diag{%
    \pbk{m-2-1}{m-2-2}{m-1-2} %
  }{%
    I_X \& Z_1 \\
    X \& X' %
  }{%
    (m-1-1) edge[labela={}] (m-1-2) %
    edge[labell={}] (m-2-1) %
    (m-2-1) edge[labelb={h}] (m-2-2) %
    (m-1-2) edge[labelr={}] (m-2-2) %
  }%
\end{center}
into a pushout.
  
The crucial insight is then that by computing the following pushout
$Z_2$ and applying its universal property, we obtain a diagram
\begin{equation}
  \Diag{%
    \pbk{m-2-1}{m-2-2}{m-1-2} %
  }{%
    I_X \& I_Y \& Y \\
    Z_1 \& Z_2 \& Y' %
  }{%
    (m-1-1) edge[labela={}] (m-1-2) %
    edge[labell={}] (m-2-1) %
    (m-2-1) edge[labelb={}] (m-2-2) %
    (m-1-2) edge[labelr={}] (m-2-2) %
    (m-1-2) edge (m-1-3) %
    edge (m-2-2) %
    (m-1-3) edge (m-2-3) %
    (m-2-2) edge (m-2-3) %
  }%
\label{eq:twopushouts}
\end{equation}
whose exterior is a pushout by construction. So by the pushout lemma
the right-hand square is again a pushout, which by
Lemma~\ref{lem:unique:pushout} is the unique pushout along which to
push $V$ to compute the oplifting of $V$ along $h'$, say
\begin{center}
  \diag{%
    Z \& Z' \\
    Y \& Y'\rlap{.} %
  }{%
    (m-1-1) edge[labela={h''}] (m-1-2) %
    edge[pro,twol={V}] (m-2-1) %
    (m-2-1) edge[labelb={h'}] (m-2-2) %
    (m-1-2) edge[pro,twor={V'}] (m-2-2) %
    (l) edge[cell=.3, labela={\alpha_V}] (r) %
  }%
\end{center}
By induction hypothesis, $V'$ is again a \trace.  But by the pushout
lemma again, $\alpha_V$ is also what we obtain by pushing $V$ along
the exterior rectangle of~\eqref{eq:twopushouts}.  A bunch of
applications of the pushout lemma finally yields that
$\alpha_M \vrond \alpha_V$ is the oplifting of $M \vrond V$ along $h$,
so $M' \vrond V'$ is a \trace by construction.
\end{proof}

Although opliftings have an opcartesian flavour, they are in fact not
opcartesian in general, and moreover opcartesian liftings do not exist
in general.

\begin{example}
  Consider the oplifting of the seed $\iota_{1,1}$ along $[1] \into
  [1] \paraofij{1}{1} [1]$, say $\iota_{1,1} \paraofij{1}{1} [1]$,
  whose final position is $[2] \paraofij{1}{1} [1]$.
  To see that it is not opcartesian, consider the diagram
  \begin{mathpar}
\Diag(.25,.8){%
    }{%
      \& \& \& |(X'')| [2]\paraofij{1,2}{1,1}[1] \\ \\ 
      |(X)| [2] \& \& |(X')| [2]\paraofij{1}{1} [1] \\
      \& \& \& |(U'')| \tauuof{1,1,1,1,1} \\ \\ 
      |(U)| \iota_{1,1} \& \& |(U')| \iota_{1,1} \paraofij{1}{1} [1] \\ 
      \& \& \& |(Y'')| [1]\paraofij{1}{1}[1] \\ \\ 
      |(Y)| [1] \& \& |(Y')| [1]\paraofij{1}{1} [1]\rlap{.}  %
    }{%
      (U) edge[labelb={}] (U') %
      (Y) edge[labelb={}] (Y') %
      (X) edge[labelb={}] (X') %
      (U) edge[bend left=10] (U'') %
      (Y) edge[bend left=10] (Y'') %
      (X) edge[bend left=10] (X'') %
      (Y') edge[labelbl={}] (Y'') %
      (U) edge[<-] (Y) edge[<-] (X) %
      (U') edge[fore,<-] (Y') edge[fore,<-] (X') %
      (U'') edge[<-] (Y'') edge[<-] (X'') %
      (U') edge[dashed] (U'') %
      (X') edge[dashed] (X'') %
    } %
\end{mathpar}
In this case, opcartesianness would mean finding dashed arrows
rendering the diagram commutative.  There is indeed a unique arrow
$\iota_{1,1} \paraofij{1}{1} [1] \to \tauuof{1,1,1,1,1}$ making the
corresponding square commute, but then no arrow $[2]\paraofij{1}{1}
[1] \to [2]\paraofij{1,2}{1,1}[1]$ fits. Indeed, there is only one
arrow of the given type, and it does not make the relevant square
\begin{center}
  \diag{%
{[2]\paraofij{1}{1} [1]} \& {[2]\paraofij{1,2}{1,1}[1]} \\
{\iota_{1,1} \paraofij{1}{1} [1]} \& {\tauuof{1,1,1,1,1}} %
  }{%
    (m-1-1) edge[labela={}] (m-1-2) %
    edge[labell={}] (m-2-1) %
    (m-2-1) edge[labelb={}] (m-2-2) %
    (m-1-2) edge[labelr={}] (m-2-2) %
  }
\end{center}
commute, because one side (down; right) of the square maps the unary
player to $\rho \rond t$, while the other (right; down) maps it to
$\rho \rond s$.  This shows that $\iota_{1,1} \paraofij{1}{1} [1]$ is
not an opcartesian lifting of $\iota_{1,1}$. But in fact,
$\iota_{1,1} \paraofij{1}{1} [1]$ and $\tauuof{1,1,1,1,1}$ are two
liftings of $\iota_{1,1}$ along $[1] \into [1] \paraofij{1}{1} [1]$.
Thus, any opcartesian lifting should have a final position mapping
both to $[2]\paraofij{1}{1}[1]$ and $[2]\paraofij{1,2}{1,1}[1]$, hence
containing just one, binary player: no \trace can meet this
requirement.
\end{example}

However, even though they are not opcartesian, opliftings are in fact
cartesian. Let us now show this, starting with a few preliminary
results.

\begin{defi}
  Let $\Pl(X) = \sum_{n \in \Nat} X[n]$ denote the set of agents of
  any \position~$X$.
\end{defi}

\begin{lem}\label{lem:connected}
  For any seed  $Y \to C \ot X$, the morphism 
  $$\sum_{(n,x) \in \Pl (X)} [n] \xto{[x]_{(n,x)}} X$$
  is epi.
\end{lem}
\begin{proof}
  This is trivial except in dimension 0, where it holds by case
  inspection.
\end{proof}

\begin{cor}\label{cor:connected}
  For all arrows as in
  \begin{center}
    \diag{%
      X \& U \\
      C \& U' %
    }{%
      (m-1-1) edge[bend left,labela={h}] (m-1-2) %
      edge[bend right,labela={h'}] (m-1-2) %
      edge[labell={t}] (m-2-1) %
      (m-2-1) edge[labelb={g}] (m-2-2) %
      (m-1-2) edge[labelr={f}] (m-2-2) %
    }
  \end{center}
  in $\Chat$ such that $f$ is 1-injective, $t \in \vertical_0$, and
  $fh = gt = fh'$, we have $h=h'$.
\end{cor}
\begin{proof}
  We construct the diagram
  \begin{center}
    \diag(1,1){%
      \sum_{(n,x) \in \Pl(X)} [n] \& \& U \\
      X \& \& U \\
      C \& \& U'\rlap{,} %
    }{%
      (m-1-1) edge[bend left=15,labela={h e}] (m-1-3) %
      edge[bend right=15,labelb={h' e}] (m-1-3) %
      edge[onto,labell={e}] (m-2-1) %
      (m-1-3) edge[identity] (m-2-3) %
      (m-2-1) edge[bend left=15,labela={h}] (m-2-3) %
      edge[bend right=15,labelb={h'}] (m-2-3) %
      edge[labell={t}] (m-3-1) %
      (m-2-3) edge[labelr={f}] (m-3-3) %
      (m-3-1) edge[labelb={g}] (m-3-3) %
    }
  \end{center}
  and observe that $h' e = h e$ by 1-injectivity of $f$, hence $h =
  h'$ because $e$ is epi by Lemma~\ref{lem:connected}.
\end{proof}

\begin{lem}\label{opcartesians are cartesian}
Opliftings of \traces are cartesian. 
\end{lem}
\begin{proof}
  Consider any oplifting $U \xto{f} U'$ of $Y \to U \ot X$ along, say
  $X \to X'$.  By Lemma~\ref{cartesiandoublecells}, it is enough to
  show that the middle arrow $U \to U'$ is in $\horizontal$ and that
  its upper square is a pullback.  The latter follows from
  Lemma~\ref{upper squares all pullbacks}.  So we just have to show
  that any square

\begin{center}
  \diag{%
  Z \&  U \\
  C \&  U' %
  }{%
(m-1-1) edge[labelu={u}] (m-1-2) %
edge[labell={t}] (m-2-1) %
(m-2-1) edge[labeld={v}] (m-2-2) %
(m-1-2) edge[labelr={}] (m-2-2) %
  }
\end{center}
with $t \in \vertical_{0}$ admits a unique lifting $C \to U$. By the
Yoneda lemma, $v$ amounts to an element of $U'(C)$ of dimension $> 1$,
but by construction $U$ and $U'$ have exactly the same such
elements. This yields a candidate lifting, say $k$, which is unique by
1-injectivity and makes the bottom triangle commute by
construction. The top one finally commutes by
Corollary~\ref{cor:connected} with $h=u$ and $h'=kt$.
\end{proof}

\subsection{Restriction of \actions}\label{subsec:actions}
Let us now extend Lemma~\ref{lem:restriction:seeds} from seeds to
\actions, following the strategy sketched at the beginning of
Section~\ref{sec:fib}.

\begin{lem}\label{fibration on actions} \
  For any \action $Y \xto{s} M \xot{t} X$ and $h \in \Dh(X',X)$, the
  factorisation $X' \xto{t'} P' \xto{h'} M$ of $t h$ with $t' \in
  \vertical$ and $h' \in \horizontal$ is such that $h'$ is 1-injective
  and the obtained restriction is a \trace of length at most $2$.  If
  $X'$ is an individual then it is either a seed or an equivalence; if
  $X'$ is an interface then it is an equivalence.  \end{lem}

\begin{proof}
  Consider any \action $Y \to M \ot X$ obtained by pushing the
  following seed-with-interface along $I \to Z$:
  \begin{center}
        \diag(.3,.6){%
    \&|(I)| I \\
   |(X)| Y_0 \&|(M)| M_0 \&|(Y)| X_0\rlap{.} %
  }{%
    (I) edge (X) edge (M) edge (Y) %
    (X) edge (M) %
    (Y) edge (M) %
  }
  \end{center}
  By Lemma~\ref{opcartesians are cartesian}, the morphism $M_0 \to M$
  is cartesian.

Consider the pullback of the bottom square below
along $h \colon X' \to X$ to obtain the top square
\begin{center}
    \Diag(.4,.6){%
      \pbk{A}{C}{B} %
      \pbk{A'}{C'}{B'} %
      \path[->,draw] %
      (I) edge[fore] (A) %
      ; %
    }{%
      |(I)| I' \& \& \&|(B)| Z' \\
      \&|(A)| X'_0 \& \& \&|(C)| X' \\
      \& \& \& \& \& \\
      |(I')| I \& \& \&|(B')| Z \\
      \&|(A')| X_0 \& \& \&|(C')| X\rlap{,} \\
        \& \& \& \& \& %
    }{%
      (I) edge (I') edge (B) %
      (I') edge (A') edge (B') %
      (B') edge (C') %
      (B) edge (C) %
      edge (B') (A') edge[fore] (C') %
      (A) edge[fore] (C) %
      edge[fore] (A') %
      (C) edge[labelr={f}] (C') %
    }
  \end{center}
which, because presheaf categories are \emph{adhesive}~\cite{DBLP:conf/fossacs/LackS04}
and $I \to X_0$ is mono, is again a pushout.

Furthermore, consider the front face, which is a pullback in $\Dh$.
By Lemma~\ref{lem:restriction:seeds}, restricting $M_0$ along $X'_0
\to X_0$ yields a \trace,
say $Y'_0
\to P'_0 \ot X'_0$ with a morphism to $Y_0 \to M_0 \ot X_0$ in $\DH$.
Since it is a \trace, $I' \into X'_0 \into M'_0$ factors through
$Y'_0 \into M'_0$.
Pushing $Y'_0 \to P'_0 \ot X'_0$ along $I' \to Z'$, we obtain a \trace
$Y' \to P' \ot X'$ (we may choose $X'$ as initial position because the
top square above is a pushout) with a morphism from $Y'_0 \to P'_0 \ot
X'_0$ in $\DH$, which is an oplifting, hence cartesian (by
Lemma~\ref{opcartesians are cartesian}).  Then, by universal property
of pushout we obtain a unique morphism $f \colon P' \to M$ making the
following cube commute:
\begin{center}
    \Diag(.4,.6){%
      \pbk{A}{C}{B} %
      \pbk{I'}{I}{B} %
      \pbk{A'}{C'}{B'} %
      \path[->,draw] %
      (I) edge[fore] (A) %
      ; %
    }{%
      |(I)| I' \& \& \&|(B)| Z' \\
      \&|(A)| P'_0 \& \& \&|(C)| P' \\
      \& \& \& \& \& \\
      |(I')| I \& \& \&|(B')| Z \\
      \&|(A')| M_0 \& \& \&|(C')| M\rlap{.} \\
        \& \& \& \& \& %
    }{%
      (I) edge (I') edge (B) %
      (I') edge (A') edge (B') %
      (B') edge (C') %
      (B) edge (C) %
      edge (B') (A') edge[fore] (C') %
      (A) edge[fore] (C) %
      edge[fore] (A') %
      (C) edge[labelr={f}] (C') %
    }
  \end{center}
By Corollary~\ref{cor:cube}, $f$ is 1-injective, which entails that the induced 
morphism of \traces is in~$\DH$.

We now need to show that $P' \to M$ is cartesian, which by
Lemmas~\ref{upper squares all pullbacks}
and~\ref{cartesiandoublecells} amounts to showing that its middle
arrow $f \colon P' \to M$ is in $\horizontal$.  To this end, consider
any morphism $t \colon Z'' \to C$ in $\vertical_0$ and morphism $(u,v)
\colon t \to f$ in $\Chat^{\mathrel{\to}}$. First of all, because $M_0
\to M$ is identity in dimensions $> 1$, the morphism $v$ uniquely
factors through $M_0 \to M$. Furthermore, in all cases where $f_0
\colon P'_0 \to M_0$ is identity in dimensions $>1$, the Yoneda lemma
entails that $C \to M_0$ uniquely factors through $P'_0$. This yields
a diagram
\begin{center}
  \Diag{%
    \path (a) -- node {?} (m-1-2) ; %
  }{%
    Z'' \& P'_0 \& P' \\
    C \& M_0 \& M %
  }{%
    (m-1-1) edge[dashed,labela={l}] (m-1-2) %
    edge[bend left=40,twoa={u}] (m-1-3) %
    edge[labell={}] (m-2-1) %
    (m-2-1) edge[labelb={}] (m-2-2) %
    edge[bend right=40,labelb={v}] (m-2-3) %
    edge[dashed,labelal={k}] (m-1-2) %
    (m-1-2) edge[labelr={f_0}] (m-2-2) %
    (m-1-2) edge[labela={g}] (m-1-3) %
    (m-2-2) edge[labelb={}] (m-2-3) %
    (m-1-3) edge[labelr={f}] (m-2-3) %
  }
\end{center}
which commutes except perhaps for the upper part marked `?'.  But the
latter also commutes by Corollary~\ref{cor:connected} with $h = u$ and
$h' = gl$. We thus obtain a lifting, which is unique by 1-injectivity
of $f$.

So what happens when is $f_0$ non-identity in dimensions $>1$? By
inspection of the proof of Lemma~\ref{lem:restriction:seeds}, except
for the easy case where $P'_0 \iso X'_0$, this is when
$M_0 = \taunimjk$ for some $n,m,i,j,k$, and $P'_0$ has one of the
shapes $\iotani+J$, $\outmjk+J$, or $\iotani + \outmjk + J$, for some
interface $J$. In the first two cases, $C \neq \taunimjk$ because
there can be no $u \colon [n] \paraij [m] \to P'$ (one agent is
missing in $P'$), so the previous argument applies.  In the third
case, letting $x$ and $y$ respectively denote the $n$- and $m$-ary
initial agents in $P'_0$ and $a = x \cdot s_i$ and $a' = y \cdot s_j$
the corresponding channels, one easily shows that $g(a) \neq g(a')$,
so again there can be no $u \colon [n] \paraij [m] \to P'$. Thus,
$C \neq \taunimjk$ and the previous argument again applies.
\end{proof}

\subsection{Restriction of \traces}\label{subsec:traces}
So far, we have shown that seeds and \actions admit cartesian liftings
in $\D_H$. We now show that it is also the case for arbitrary \traces.
We proceed by induction on the length of the considered \trace, which
requires a few preliminary results.

\begin{lem} \label{cube pullback}
  In sets, for any commuting diagram
  \begin{center}
  \Diag{%
    \pbk{J}{B}{A}
  }{%
   |(I)| I \& |(A)| A \& |(C)| C \\
   |(J)| J \& |(B)| B \& |(D)| D \\
  }{%
    (I) edge (J) edge (A) %
    (A) edge[into] (B) edge (C) %
    (J) edge (B) %
    (B) edge (D) %
    (C) edge[into] (D)
  }%
\end{center}
whose exterior rectangle is a pullback, with the marked pushout and
monos, the right-hand square is also a pullback.
\end{lem}
\begin{proof}
 We check the universal property of pullback for $A$, relative
  to $1$, which is enough in sets. So consider any commuting square
  \begin{center}
    \diag{%
      1 \& C \\
      B \& D\rlap{.} %
    }{%
      (m-1-1) edge[labela={c}] (m-1-2) %
      edge[labell={b}] (m-2-1) %
      (m-2-1) edge[labelb={}] (m-2-2) %
      (m-1-2) edge[labelr={}] (m-2-2) %
    }
  \end{center}
  First, we observe that there is at most one mediating arrow $1 \to
  A$, because $A \to B$ is mono.

  If $b$ has an antecedent in $A$, say $a$, then because $C \to D$ is
  mono, $a$ makes both required triangles commute and we are done.

  Otherwise, by surjectivity of $A+J \to B$, $b$ admits an antecedent
  in $J$, i.e., there exists $j \colon 1 \to B$ such that $b$ is $1
  \xto{j} J \to B$, then we have a cone to $C \to D \ot J$, so we
  apply the universal property of pullback to obtain $i$ as in
  \begin{center}
  \Diag{%
    \pbk{J}{B}{A}
  }{%
    |(one)| 1 \\
    \& |(I)| I \& |(A)| A \& |(C)| C \\
    \& |(J)| J \& |(B)| B \& |(D)| D \\
  }{%
    (I) edge (J) edge (A) %
    (A) edge[into] (B) edge (C) %
    (J) edge (B) %
    (B) edge (D) %
    (C) edge[into] (D)
    (one) edge[fore,bend right=20,labelo={b}] (B) %
    edge[labelo={c}] (C) %
    edge[bend right=20,labelo={j}] (J) %
    edge[labelo={i}] (I)
  }%
\end{center}
  making everything commute and so $1 \xto{i} I \to A$ suits our needs.
\end{proof}

\begin{cor} \label{cor:cube pullback} In any presheaf category, in any
  commuting cube
  \begin{center}
    \Diag(.4,.6){%
      \pbk{A}{C}{B} %
      \pbk{I'}{I}{B} %
      \pbk{A'}{C'}{B'} %
      \path[->,draw] %
      (I) edge[fore] (A) %
      ; %
    }{%
      |(I)| I \& \& \&|(B)| B \\
      \&|(A)| A \& \& \&|(C)| C \\
      \& \& \& \& \& \\
      |(I')| I' \& \& \&|(B')| B' \\
      \&|(A')| A' \& \& \&|(C')| C'\rlap{,} \\
    }{%
      (I) edge (I') edge[] (B) %
      (I') edge (A') edge[into] (B') %
      (B') edge (C') %
      (B) edge (C) %
      edge (B') (A') edge[fore] (C') %
      (A) edge[fore] (C) %
      edge[fore] (A') %
      (C) edge[labelr={f}] (C') %
    }
  \end{center}
  with the marked pushouts, pullback, and mono,
  the front square is also a pullback.
\end{cor}
\begin{proof}
  It suffices to show the result in sets, as all involved properties
  are pointwise in presheaf categories.  First, as monos are stable
  under pullback and pushout, $I \to B$, $A \to C$, and $A' \to C'$
  are also monos. Furthermore, pushouts along monos are also
  pullbacks, so the top and bottom faces are also pullbacks. By the
  pullback lemma, the rectangle
\begin{center}
  \Diag{%
    \pbk{m-2-1}{m-1-1}{m-1-2} %
    \pbk{m-2-2}{m-1-2}{m-1-3} %
  }{%
    I \& I' \& A' \\
    B \& B' \& C' %
  }{%
    (m-1-1) edge[labela={}] (m-1-2) %
    edge[labell={}] (m-2-1) %
    (m-2-1) edge[labelb={}] (m-2-2) %
    (m-1-2) edge[labelr={}] (m-2-2) %
    (m-1-2) edge (m-1-3) %
    (m-2-2) edge (m-2-3) %
    (m-1-3) edge (m-2-3) %
  }%
\end{center}
is a pullback.  The previous lemma thus entails the result.
\end{proof}

\begin{lem}\label{main point fibration on traces} 
For any commuting diagram as the solid part of
\begin{center}
  \diag{%
   |(T)| T \& |(V')| V' \& |(U')| U' \\
   |(C)| C \& |(V)| V \& |(U)| U %
  }{%
    (T) edge[labell={t}] (C) edge (V') %
    (V') edge[labellat={h}{.2}] (V) edge (U') %
    (C) edge (V) %
    (V) edge (U) %
    (U') edge[labelr={f}] (U)
    (C) edge[dashed,fore,labelalat={k}{.3}] (U')
  }
\end{center}
with $t \in \vertical_0$, $h \in \horizontal$, and
$f$ 1-injective, there is a unique lifting $k$ as shown.
\end{lem}

\begin{proof}
  Because $h \in \horizontal$, there is a unique map from $k' \colon C
  \to V'$ making both triangles commute. By composing $k'$ with $V'
  \rightarrow U'$, we obtain a lifting $k$ for the desired
  square. Uniqueness follows from 1-injectivity of $f$.
\end{proof}

\begin{lem}\label{fibration} \ 
 The codomain functor $\cod \colon \D_H \to \D_h$ is a subfibration
    of $\cod \colon \D^{1}_H \to \D^{1}_h$.
\end{lem}

\begin{proof}
  First of all, it is sufficient to prove that given any \trace
  $Y \xto{s} U \xot{t} X$ and $h \in \Dh(X',X)$, any cartesian lifting
  in $\D^1$ lies in $\D$, i.e., the obtained vertical morphism is
  again \atrace and the double cell to $U$ lies in $\D$ (i.e., all its
  components are 1-injective). Indeed, mediating morphisms computed in
  $\D^1$ are automatically in $\D$ by cancellation.

  We proceed by induction on $U$.  If $U$ is an equivalence, then the
  result is obviously true.  Otherwise, $U = M \vrond V$ for some
  \action $M$ and \trace $V$.  Let us call $Z$ the final \position
  of $M$. By Lemma~\ref{fibration on actions}, $M$ admits a lifting
  $P'$ along $h$, with a final \position $Z'$, and $Z' \rightarrow
  Z$ and $f_M \colon P' \to M$ are 1-injective. By induction
  hypothesis, $V$ admits a lifting $V'$ along $Z' \rightarrow Z$ with
  a double cell to $V$ in $\DH$. Therefore, considering the composite
  $P' \vrond V'$, we have a commuting diagram as in Figure~\ref{fig:subfib},
where $f$ is obtained by universal property of pushout.
\begin{figure}[ht]
\diagramme[diagorigins={1}{1.5}]{}{%
      \pbk[2em]{V'}{C'}{P'} %
      \pbk[2em]{V}{C}{M} %
      \pbk{P'}{Z'}{Z} %
      \pbk{C'}{V'}{V} %
      \pbk{V'}{Y'}{Y} %
      \path[->,draw] %
      (Z) edge (V) %
      (C') edge[fore,labelaat={f}{0.2}] (C) %
      ; %
    }{%
      |(Y')| Y' \& \& \& \& |(Y)| Y \\
      \& |(V')| V' \& \& \& \& |(V)| V \\
      \& \& |(C')| P' \vrond V' \& \& \& \& |(C)| M \vrond V \\
      |(Z')| Z' \& \& \& \& |(Z)| Z \\
      \& |(P')| P' \& \& \& \& |(M)| M \\
      |(blank)| \\
      |(X')| X' \& \& \& \& |(X)| X\rlap{,} %
    }{%
      (X) edge[labelbr={t_0}] (M) %
      (Z) edge[labelbl={s_0}] (M) %
          edge[labelalat={t_1}{.7}] (V) %
      (Y) edge[labelar={s_1}] (V) %
      (M) edge[labelbr={t_2}] (C) %
      (V) edge[labelar={s_2}] (C) %
      (X') edge (X) %
           edge (P') %
      (Z') edge (Z) %
           edge (P') %
           edge (V') %
      (Y') edge (Y) %
           edge (V') %
      (P') edge[labela={f_M}] (M) %
           edge[fore,labeloat={t'_2}{.3}] (C') %
      (V') edge[labela={f_V}] (V) %
           edge[labeloat={s'_2}{.6}] (C') %
    }
\label{fig:subfib}
\caption{Diagram for induction step in Lemma~\ref{fibration}}
\end{figure}

Because pushouts along monos are also pullbacks in presheaf
categories, both marked pushouts are also pullbacks.  Furthermore, by
Lemma~\ref{upper squares all pullbacks}, $Z' = P' \times_M Z$ and
$Y' = (V' \times_V Y)$, as shown.  Also, by Corollary~\ref{cor:cube
  pullback}, $V' = (P' \vrond V') \times_{M \vrond V} V$, as shown.
Furthermore, by Proposition~\ref{prop:D:D0}, $f$ is 1-injective.

By Lemmas~\ref{cartesiandoublecells} and~\ref{upper squares all
  pullbacks}, it suffices to show that $f$ is in $\horizontal$, i.e.,
that it is right-orthogonal to any $T \xto{t} C$ in
$\vertical_0$. Consider any commuting square
\begin{center}
    \diag{%
      |(X)| T \& |(Co')| P' \vrond V' \\
      |(C)| C \& |(Co)| M \vrond V. %
    }{%
    	  (X)   edge[labell={t}] (C) %
    	        edge[labela={u}] (Co') %
    	  (C)   edge[labelb={v}] (Co) %
    	  (Co') edge[labelr={f}] (Co) %
    }
\end{center}
Since $M \vrond V$ is the coproduct of $M$ and $V$ in dimensions greater than
1 and $C$ is a representable of dimension greater than 1, we have that
$v$ factors either through $t_2$ or $s_2$.

If $v$ factors through $s_2$, then by universal property of pullback we find a map $u' : T \rightarrow V'$ 
making
\begin{center}
  \Diag{%
    \pbk{m-2-2}{m-1-2}{m-1-3} %
  }{%
    T \& V' \& P' \vrond V' \\
    C \& V \& M \vrond V %
  }{%
    (m-1-1) edge[labela={u'}] (m-1-2) %
    edge[labell={t}] (m-2-1) %
    edge[bend left,labela={u}] (m-1-3) %
    (m-2-1) edge[labelb={}] (m-2-2) %
    edge[bend right,labelb={v}] (m-2-3) %
    (m-1-2) edge[labelr={}] (m-2-2) %
    (m-1-2) edge (m-1-3) %
    (m-2-2) edge[labelb={s_2}] (m-2-3) %
    (m-1-3) edge (m-2-3) %
  }
\end{center}
commute.  Then, by Lemma~\ref{main point fibration on traces}, we find
a unique lifting as desired.

If $v$ factors as $t_2 v'$, then by Lemma~\ref{main point fibration on
  traces}, it is sufficient to show that there is a map $u' : T
\rightarrow P'$ making
\begin{center}
  \diag{%
    T \& P' \& P' \vrond V' \\
    C \& M \& M \vrond V %
  }{%
    (m-1-1) edge[labela={u'}] (m-1-2) %
    edge[labell={t}] (m-2-1) %
    edge[bend left,labela={u}] (m-1-3) %
    (m-2-1) edge[labelb={v'}] (m-2-2) %
    edge[bend right,labelb={v}] (m-2-3) %
    (m-1-2) edge[labelr={}] (m-2-2) %
    (m-1-2) edge[labela={t'_2}] (m-1-3) %
    (m-2-2) edge[labelb={t_2}] (m-2-3) %
    (m-1-3) edge[labelr={f}] (m-2-3) %
  }
\end{center}
commute.
To that end, it is sufficient to show that for every $[n] \xto{x} T$,
there is a map $[n] \xto{f_x} P'$ such that
\begin{equation}
    \diag{%
      |(ltl)| [n] \& |(ltr)| T \\
      |(lbl)| P' \& |(lbr)| P' \vrond V' %
    }{%
    	  (ltl) edge[labelu={x}] (ltr) %
    	        edge[labell={f_x}] (lbl) %
    	  (ltr) edge[labelr={u}] (lbr) %
    	  (lbl) edge[labelb={t'_2}] (lbr) %
    }\label{eq:square:fx}
\end{equation}
commutes. Indeed, if that is the case, then the square
\begin{center}
    \diag(.6,2){%
      |(ltl)| \sum_{(n,x) \in \Pl(T)}[n] \& |(ltr)| T \\
      |(lbl)| P' \& |(lbr)| P' \vrond V' %
    }{%
    	  (ltl) edge[labelu={[x]_{(n,x) \in \Pl(T)}}] (ltr) %
    	        edge[labell={[f_x]_{(n,x) \in \Pl(T)}}] (lbl) %
    	  (ltr) edge[labelr={u}] (lbr) %
    	  (lbl) edge[labelb={t'_2}] (lbr) %
          (ltr) edge[dashed,labelal={u'}] (lbl) %
    }
\end{center}
also commutes, and since its bottom map is mono and its top map is epi
by Lemma~\ref{lem:connected}, there is a unique lifting $u' \colon T
\rightarrow P'$ making both triangles commute. The square
\begin{center}
    \diag{%
      |(ltl)| T \& |(lbl)| P' \\
      |(ltr)| C \& |(lbr)| M %
    }{%
    	  (ltl) edge[labell={t}] (ltr) %
    	        edge[labela={u'}] (lbl) %
    	  (ltr) edge[labelb={v'}] (lbr) %
    	  (lbl) edge (lbr) %
    }
\end{center}
also commutes because it commutes when composed with $t_2$, which is mono.

So we now need to show that for every $[n] \xto{x} T$, there is a map
$[n] \xto{f_x} P'$ making the square~\eqref{eq:square:fx} commute.
Because $P' + V' \xto{[t'_2,s'_2]} P' \vrond V'$ is epi and $[n]$ is a
representable presheaf, $[n] \xto{ux} P' \vrond V'$ factors through either
$t'_2$ or $s'_2$. 

If it factors through $s'_2$, say as $s'_2 x'$, then we have 
$$t_2 (v' t x) = v t x = f u x = f s'_2 x' = s_2 f_V x',$$
so by universal property of the pullback $Z$ there exists a unique
$x'' \colon [n] \to Z$ such that $s_0 x'' = v' t x$ and $t_1 x'' = f_V
x'$. We thus obtain a commuting diagram
\begin{center}
  \diag{%
        |(n)| {[n]} \& \&|(Z)| Z \\
        |(T)| T \&|(C)| C \&|(M)| M, %
  }{%
    (n) edge[labell={x}] (T) %
    edge[labelal={x''}] (Z) %
    (T) edge[labela={t}] (C) %
    (C) edge[labela={v'}] (M) %
    (Z) edge[labelr={s_0}] (M) %
  }
\end{center}
which is impossible because $v' t x$, as one of the \agents{}
performing \action{} $C$ in $M$, cannot remain in the final \position
of $M$.

Thus, $ux$ factors through $t'_2$ and we are done.
\end{proof}

\section{A playground for \texorpdfstring{$\pi$}{π}}\label{sec:playground:pi}
In the previous section, we have proved the main playground axiom,
asserting that the codomain functor $\DH \to \Dh$ of the double
category of \traces{} constructed in Section~\ref{sec:double} is a
fibration. We now prove the remaining axioms. In
Section~\ref{subsec:candidate}, we equip $\D$ with playground
structure and prove that it satisfies all the needed axioms, except
both decomposition axioms (\axref{leftdecomposition}
and~\axref{views:decomp}) which require a bit more work. In
Section~\ref{subsec:correctness}, we establish a correctness criterion
detecting when a given cospan in $\Chat$ is \atrace. We then use this
criterion in Section~\ref{subsec:playground:pi} to prove both
remaining axioms.
\subsection{A candidate playground}\label{subsec:candidate}
So we start in this section by defining the needed additional
structure on $\D$.

\begin{defi}\label{def:playground:data}
  We recall from Lemma~\ref{lem:restriction:seeds} that $\DI$, the set
  of \emph{individuals}, consists of representable \positions $[n]$.
  Let $\B$, the full subcategory of \emph{basic} \actions, span
  all seeds of shape $\tau_n$, $\forkrn$, $\forkln$, $\nun$, $\tickn$,
  $\iotani$, or $\outmjk$.  \emph{Full} \actions (notation $\F$) are
  all \actions obtained from seeds of shape $\tau_n$, $\forkn$,
  $\nun$, $\tickn$, $\iotani$, $\outmjk$, or $\taunimjk$.
  \emph{Closed-world} \actions are all \actions obtained from seeds of
  shape $\tau_n$, $\forkn$, $\nun$, $\tickn$, or $\taunimjk$. Let $\W$
  denote the graph with \positions{} as vertices and closed-world
  \actions between them as edges (the initial \position being the
  target). Finally, all decompositions of any \trace $U$ into \actions
  have the same length which we denote by $\length{U}$.
\end{defi}
Here is a summary of which \actions{} are basic, full and closed-world:
\begin{center}
  $\begin{array}[b]{|r|c|c|c|c|c|c|c|c|c|} \hline
     & \taun & \forkln & \forkrn & \forkn & \iotani & \outmjk & \taunimjk & \nun & \tickn \\
     \hline
     \text{Basic} & \checkmark & \checkmark & \checkmark & & \checkmark & \checkmark & & \checkmark & \checkmark  \\
     \text{Full} & \checkmark &   &  & \checkmark & \checkmark & \checkmark & \checkmark & \checkmark & \checkmark  \\
     \text{Closed-world} & \checkmark & && \checkmark & && \checkmark & \checkmark &
                                                                     \checkmark \\ \hline
  \end{array}$~.
\end{center}
\begin{rem}
  The definition matches the explanations following
  Definition~\ref{def:playground}. Basic \actions are as small as
  possible, which here means they start from one \agent and only
  retain one \agent in the final position. Full \actions are those
  that retain all possible \agents in the final position.  So the only
  kind of \actions which are not basic are $\forkn$ and $\taunimjk$.
  They each have one
  sub-\action for each \agent in their final \positions, namely $\forkrn$
  and $\forkln$ for $\forkn$, and $\inni$ and $\outmjk$ for $\taunimjk$.
  All of these sub-\actions are basic, but only $\inni$ and $\outmjk$ are full.
  Finally, a new class of \actions appear here, that of closed-world
  \actions. Intuitively, it consists of those \actions that do not
  involve any interaction with the environment. Or, in other words,
  those that cannot be extended, even by adding new \agents.  E.g.,
  $\iotani$ may be completed to $\taunimjk$ by adding \anagent, hence
  is not closed-world. But $\taunimjk$ is. Closed-world \actions will
  be at the basis of our semantic notion of testing equivalence
  (Definition~\ref{def:fair:semantic}).
\end{rem}
\begin{lem} \ 
  \begin{itemize}
  \item \axref{discreteness} $\DI$, viewed as a subcategory of $\Dh$,
    is discrete. Basic \actions have no non-trivial automorphisms in
    $\DH$.  Vertical identities on individuals have no non-trivial
    endomorphisms.
  \item \axref{individuality} (Individuality) Basic \actions have individuals as both domain
    and codomain. 
  \item \axref{atomicity}
      (Atomicity) For any cell $\alpha \colon U \to U'$, if $\length{U'}
      = 0$ then also $\length{U} = 0$.  Up to a special isomorphism in
      $\DH$, all \traces of length $n > 0$ admit decompositions into
      $n$ \actions.  For any $U \colon X \proto Y$ of length 0, there is
      an isomorphism $\idv_X \to U$ in $\DH$ as in
      \begin{center}
        \diag{%
            |(X0)| X \&|(X)| X \\
            |(Xi)| X \&|(Y)| Y\rlap{.} %
          }{%
            (X0) edge[identity,pro,twol={}] (Xi) %
            edge[identity] (X) %
            (X) edge[pro,twor={U}] (Y) %
            (Xi) edge[labelb={\bar{U}}] (Y) %
            (l) edge[cell=0.3,labelb={\alpha^U}] (r) %
          }
      \end{center}
  \item \axref{fibration:continued} Restrictions of \actions (resp.\
    full \actions) to individuals either are \actions (resp.\ full
    \actions), or have length 0. 
  \end{itemize}
\end{lem}
\begin{proof}
  \axref{discreteness} and~\axref{individuality} are direct by
  Yoneda. \axref{fibration:continued} is also easy in view of
  Lemma~\ref{fibration} and its proof.  For~\axref{atomicity}, any
  vertical $X \to U \ot Y$ of length $0$, being a \trace, is
  isomorphic to an identity cospan, say $Z \to Z \ot Z$. To construct
  $\alpha^U$, just take the composite $\idv_X \xto{\iso} \idv_Z
  \xto{\iso} U$.  \end{proof}

Let us now treat the axiom for views, which is really easy.  It
actually becomes stronger because of Remark~\ref{rem:basic}, though
this does not affect the rest of the construction:
\begin{defi}\label{def:DB0}
  Let $\DB_{0}$ be the full subcategory of $\DH$ having as objects
  basic \actions and vertical identities between individuals.
\end{defi}
\begin{lem}\label{lem:views} \ 
 \axref{views} For any \action $M \colon Y \proto X$ and $y
    \colon d \to Y$ in $\Dh$ with $d\in \DI$, there exists a unique
    cell
\begin{center}
  \diag{%
    |(d)| d \& |(Y)| Y \\
    |(dMy)| d^{y,M} \& |(X)| X, %
  }{%
    (d) edge[labelu={y}] (Y) %
    edge[pro,dashed,twol={v^{y,M}}] (dMy) %
    (Y) edge[pro,twor={M}] (X) %
    (dMy) edge[dashed,labeld={y^M}] (X) %
    (l) edge[cell=.3,dashed,labelu={\scriptstyle \alpha^{y,M}}] (r) %
  }
\end{center}
with $v^{y,M} \in \DB_0$.
\end{lem}
\begin{proof}
  The result holds for seeds, by case analysis.  E.g.,
  $v^{ls, \forkn} = \forkln$, $v^{rs, \forkn} = \forkrn$, and so on.
  Now, any \action $M$ comes with a cell from its generating seed, say
  $\beta\colon M_0 \to M$.  If $y$ is in the image of $M_0$, then the
  required cell $\alpha^{y,M}$ is $\beta \rond
  \alpha^{y,M_0}$. Otherwise, $v^{y,M} = \idv_d$ admits a cell to $M$,
  which suits our needs. Uniqueness follows by~\axref{discreteness}.
\end{proof}
\begin{rem}
  The important result that Axiom~\axref{views}
  entails~\cite[Proposition 4.24]{HirschoDoubleCats} says that when we
  replace $M$ with any \trace $u$, we get a double cell $\alpha^{y,u}$
  which is only unique up to isomorphism.  Below, we still define
  views up to isomorphism, so our modified Axiom~\axref{views} does
  not make this any stronger.
\end{rem}
We conclude this section with the (straightforward) verification
of~\axref{finiteness} and~\axref{basic:full}.
\begin{lem}
\axref{finiteness} For any $X$, $\DI / X$ is finite.
\end{lem}
\begin{proof}
All \positions are finitely presentable presheaves.
\end{proof}

\begin{lem}
\axref{basic:full} For all $d \in \DI$ and \actions $M \colon X \proto d$, $M' \colon X' \proto d$, 
    and $b \colon d' \proto d$ with $M$ and $M'$ full and $b$ basic,
    if there exist cells $M \ot b \to M'$ then $M \iso M'$.
\end{lem}
\begin{proof}
By case analysis. E.g., if $b = \paraln$, then $M \iso \paran \iso M'$.
\end{proof}

\begin{rem}
  While the verification of~\axref{basic:full} is straightforward in
  our case, this axiom does impose strong constraints on playgrounds.
  Morally, it demands that the basic sub\actions of a given full
  \action should be disjoint from those of a different full \action.
  To see why this is restrictive, let, for all $j \in n$,
  $\iota_{n,i,j}$ denote the quotient of $\iota_{n,i}$ by the
  equation $s \circ s_{n+1} = s \circ s_j$. The equation says that
  the received channel was already known as channel number
  $j$. Further let $[n]/\ens{i=j}$ denote $[n]$ quotiented by
  $s_i = s_j$. We could be tempted to decree that the cospan
$${[n+1] / \ens{n+1 = j}} \xto{s} \iota_{n,i,j} \xot{t} [n]$$
is \anaction. An example consequence would be that, e.g., the
synchronisation on $[n] \paraofij{i,l}{j,k} [m]$ where $[m]$ sends $k$
on $j$, when restricted to the receiver, would give $\iota_{n,i,l}$
instead of $\iotani$. But then $\iotani$ and $\iota_{n,i,l}$ would be
two non-isomorphic full \actions sharing a common basic sub\action,
$\iotani$.
\end{rem}

We now have proved all playground axioms for $\D$, except right and
left decomposition. These require the development of more machinery,
which we undertake in the next section.

\subsection{Correctness criterion}\label{subsec:correctness}
In order to prove the remaining playground axioms for $\D$, we set up
a combinatorial characterisation of \traces among cospans.  Before
delving into technicalities, let us briefly map out our correctness
criterion.  Given a \trace $Y \into U \otni X$, we start by forgetting
the cospan structure and exploring the properties of $U$ alone.

The main idea is to construct a binary relation over the elements of
$U$, modeling causality. So, e.g., if \anagent $x \in U[n]$ forks into
$x_1$ and $x_2$, then we will have causal relations
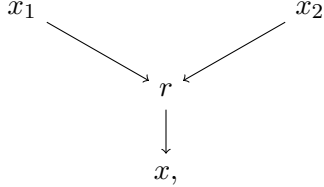
\begin{center}
  \diag{%
    x_1 \& \& x_2 \\
    \& r \\
    \& x, %
  }{%
    (m-2-2) edge[<-] (m-1-1)
    edge[<-] (m-1-3)
    edge (m-3-2) %
  }
\end{center}
where $r$ denotes the corresponding element in $U(\forkn)$.  In order
for $U$ to admit a sequential decomposition into \actions, the main
criterion is that the causality relation should be acyclic.  

In addition to this, a few sanity checks are necessary.  First of all,
because \actions are merely seeds pushed along 1-injective maps from
their interfaces, the neighbourhood of each \action $x \in U(\mu)$
should not be too degenerate. For instance, the corresponding map
$\name{x}\colon \mu \to U$ should be 1-injective.  Moreover, for
inputs and channel creations, the new channel should really be
new. This property, which is a bit tedious to define properly, is
called \emph{local 1-injectivity}.

Furthermore, when we add a new \action to some
\trace, it is played by \anagent in the final \position. This
entails that no two \actions in $U$ may be performed by the same
\agent. We call this \emph{target-linearity} below. Symmetrically, no
two \actions may share their `created' \agents, which we call
\emph{source-linearity}. \emph{Linearity} is then the conjunction of
source- and target-linearity.

These conditions are sufficient, in the sense that if any
$U \in \FPshf{\C}$ has an acyclic causal relation, and is furthermore
locally 1-injective and linear, then it is the middle object of
\atrace.  But in fact, it is then easy to determine the corresponding
initial and final \positions. We design notions of \emph{initial}
and \emph{final} morphisms, so that $Y \xto{s} U \xot{t} X$ is \atrace
iff $U$ satisfies the above conditions, $t$ is initial, and $s$ is
final.

Let us first define the causal relation.  A first step is to restrict
attention to the \emph{cores} of $U$, in the following sense, which
are intuitively the main elements. E.g., for a forking \action
$x \in U(\forkn)$, keeping track of $x$ is enough, and handling
$x \cdot l$ and $x \cdot r$ tends to get in the way.  Technically, an
input or output is a core iff it is not part of a synchronisation; and
a left or right fork \action is a core iff it is not part of a full
fork \action. Here is a concise definition:
 \begin{defi}
   A \emph{core} of a presheaf $U \in \Chat$ is an element of
   dimension $> 1$ which is not the image of any element of higher
   dimension.
 \end{defi}

 Our definition of the causal relation will rely on the preliminary
 notions of sources and targets of a core, and that of channels
 created by a core.  These notions will fix the direction of the
 causal relation.
 \begin{defi}
   For any $U$ and core $\mu \in U(c)$, 
   \begin{itemize}
   \item the \emph{sources} of $\mu$ are the \agents{} $x$ such that
     $x = \mu \cdot f \cdot s$ for some $f$;
   \item the \emph{targets} of $\mu$ are the \agents{} $y$ such that
     $y = \mu \cdot f \cdot t$ for some $f$;
   \item a channel $a \in U(\star)$ is \emph{created} by $\mu$ iff $\mu$
     has the shape $\nun$ or $\iotani$, and $a = \mu \cdot s \cdot
     s_{n+1}$.
   \end{itemize}
 \end{defi}
 \begin{example}
   In the representable $\forkn$, there is one target, $l \circ t$ (or
   equivalently $r \circ t$), and two sources, $s_1 = l \circ s$ and
   $s_2 = r \circ s$.  Another example is $\taunimjk$, which has two
   targets, $\sender \circ t$ and $\rho \circ t$, and two
   sources. However, $\rho \rond s \rond s_{n+1}$ is not created by
   the input element $\rho$, because it is not a core.
 \end{example}

 \begin{defi}
   For any $U$, let its \emph{causal graph} $G_U$ have:
   \begin{itemize}
   \item as vertices, all channels, agents, and cores in $U$,
   \item  for all $x \in U[n]$ and $i \in n$, an edge $x \to x
     \cdot s_i$,
   \item and, for each core $\mu$, an edge from each of its sources
     and created channels, and one into each of its targets, as in
     \begin{center}
     \diag{%
       \textrm{source}_1 \& \textrm{source}_2 \& \textrm{created} \\
       \& \textrm{core} \& \\
       \textrm{target}_1 \& \& \textrm{target}_2 \rlap{.} %
     }{%
       (m-1-1) edge (m-2-2) %
       (m-1-2) edge (m-2-2) %
       (m-1-3) edge (m-2-2) %
       (m-2-2) edge (m-3-1) %
       edge (m-3-3) %
     }
   \end{center}
 \end{itemize}
\end{defi}

Please note that edges $a \to \mu$ from a channel to an input \action
exist only if the involved \action is not part of a synchronisation;
for otherwise the synchronisation is a core, not the input.

 The obtained graph is actually a binary relation, since there is at most one
 edge between any two vertices. It is also a colored graph, in the
 sense that it comes equipped with a morphism to the graph $L$:

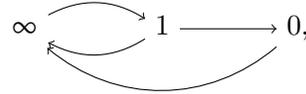
  \begin{center} 
    \diag{%
      \infty \& 1 \& 0, %
    }{%
      (m-1-1) edge[bend left] (m-1-2) %
      (m-1-2) edge[bend left] (m-1-1) %
      (m-1-3) edge[bend left=40] (m-1-1) %
      (m-1-2) edge (m-1-3) %
    } 
  \end{center} 
  mapping cores to $\infty$, agents to $1$, and channels to
  $0$. (In particular, there are no edges from channels to
  agents or from cores to channels.)  For any graph $G$, equipped with a
  morphism $l \colon G \to L$, we call vertices of $G$ channels,
  agents, or cores, according to their label.



 As expected, we have:
 \begin{prop}
   For any \trace $U$, $G_U$ is acyclic (in the directed sense).
 \end{prop}
 \begin{proof}
   By induction on any decomposition of $U$.
 \end{proof}

 Let us now consider local 1-injectivity, linearity, initiality and
 finality. First, let us emphasise that for all seeds $Y \into M \otni
 X$, $M$ is a representable presheaf, so, e.g., it makes sense to
 consider $U(M)$.
 \begin{defi} A presheaf $U$ is \emph{locally 1-injective} iff for any
   seed $Y \into M \otni X$ with interface $I$ and core $\mu \in
   U(M)$, if two distinct elements of $M$ are identified by the Yoneda
   morphism $\mu \colon M \to U$, then they are in (the image of)
   $I(\star)$.
  \end{defi}
  This is equivalent to requiring that all morphisms $\yoneda_c \to
  U$, for all $c \in \C$, are 1-injective, and that
  for all core inputs and channel creations $x$ of arity $n$ in $U$, 
  for all $i \in n$, we have
  $$x \cdot s \cdot s_{n+1} \neq x \cdot s \cdot s_i.$$

\begin{prop}
  Any \trace $U$ is locally 1-injective.
\end{prop}
\begin{proof}
  Choose a decomposition of $U$ into \actions; $\mu$ corresponds to
  precisely one such \action, say $M'$, obtained, by definition, from
  some seed $M$ as a pushout~\eqref{morphismofinterfacedseeds}.  By
  construction of pushouts in presheaf categories, $M'$ is obtained
  from $M$ by identifying some channels according to $I \to Z$.
\end{proof}

Let us observe that, because local 1-injectivity is only about cores,
an input which is part of a synchronisation may receive an already
known channel, even if its $n+1$th channel is not part of its
interface --- because it is not a core.

After local 1-injectivity, let us consider linearity.
\begin{defi}
  Any $G \in \Gph / L$ is \emph{source-linear} iff
  for any cores $\mu, \mu'$, and other vertex (necessarily an agent or a channel) $x$,
  $\mu \ot x \to \mu'$ in $G$, then $\mu = \mu'$;
  $G$ is \emph{target-linear} iff for any cores $\mu,\mu'$ and agent $x$, if
  $\mu \to x \ot \mu'$ in $G$, then $\mu = \mu'$;
  $G$ is \emph{linear} iff it is both source-linear and target-linear.
\end{defi}

 \begin{prop}
   For any \trace $Y \xto{s} U \xot{t} X$, $G_U$ is linear.
 \end{prop}

 \begin{proof} Straightforward, by induction on any decomposition of
 $U$ into \actions, observing that we glue along agents and channels
 which are initial on one side and final on the other.  \end{proof}

The last of our necessary sanity checks is about initiality and finality. 
The idea here is that one may read in $U$ alone what both legs of the 
cospan $Y \into U \otni X$ should be. 
 \begin{defi}
   An agent is \emph{initial} in $U$ when it is not the source of any
   \action, i.e., for no \action $\mu \in U$, $x = \mu \cdot s$. A channel
   is initial when it is not created by any core.

   An agent $x$ in $U$ is \emph{final} iff it is not the target of any
   \action, i.e., for no \action $\mu \in U$, $x = \mu \cdot t$.  All
   channels are final.
 \end{defi}
 \begin{lem}
   An agent is initial in $U$ iff it has no edge to any core in $G_U$.
 \end{lem}
 \begin{lem}\label{lem:final}
   An agent is final in $U$ iff it has no edge from any core in $G_U$.
 \end{lem}

Now, here is the expected characterisation:
\begin{thm}\label{thm:completeness}
  A monic cospan $Y \xinto{s} U \xotni{t} X$ of finite presheaves is a \trace iff
  \begin{corrcond}
  \item \label{cond:locinj} $U$ is locally 1-injective, 
  \item \label{cond:X} $X$ contains exactly the initial agents and channels in
    $U$, 
  \item \label{cond:Y} $Y$ contains exactly the final agents and channels in
    $U$, 
  \item \label{cond:causality} and $G_U$ is linear and acyclic.
  \end{corrcond}
\end{thm}
Of course, we have almost proved the `only if' direction, and the rest
is easy, so only the `if' direction remains to prove. The rest of this
section is devoted to this.  So given a cospan satisfying the above
conditions, we intend to sequentialise it, i.e., decompose it into
\actions. We will proceed by induction on the number of cores in $U$,
by picking a core $\mu$ which is maximal according to $G_U$, removing
it from $U$ and applying the induction hypothesis to the rest.
However, it may not be obvious how we should remove $\mu$ from $U$.
E.g., the topos-theoretic difference $U \setminus \mu$ does not yield
the expected result, as it removes all sources of $\mu$.  Instead, we
consider the following operation: for any morphism of presheaves
$f \colon U \to V$ and set $W$, let
$U - W = \sum_{c \in \C} \im(U (c)) \setminus W \subseteq \sum_{c \in
  \C} V (c)$. This is a slight abuse of notation, as $f$ is implicit,
but it should be easily inferred from context.
\begin{rem}
  We observe that $U - W$ is generally just a set, not a presheaf;
  i.e., its elements are not necessarily stable under the action of
  morphisms in $\C$. Consider for example $U = [1] \para [1]$ and let
  $W$ consist of the first agent and the unique channel. Then $U - W$
  does not contain the unique channel of $U$, so the action of $s_1$
  on the second agent steps outside $U-W$.
\end{rem}
But there is one useful case where $U - W$ is indeed a subpresheaf of
$U$, as we show below in Lemma~\ref{maxcore}.
\begin{defi}
  For any seed $Y \into M \otni X$, let the \emph{past}
  $\past{M} = M - Y$ of $M$ be the set of its elements not in the
  image of $Y$.  For any such $M$, presheaf $U$, and core $\mu \in U
  (M)$, let $\past{\mu} = \im(\past{M})$ consist of all images of
  $\past{M}$.
\end{defi}

To explain the statement a bit more, by Yoneda, we see $\mu$ as a map
$M \to U$, so we have a set-function $$\past{M} \into \ob(\el(M)) \to
\ob(\el(U))$$ (recalling that $\el$ denotes the category of
elements). We observe that $\past{\mu}$ is always a set of agents and
\actions only, since channels present in $X$ always are in $Y$ too.

Given a core $\mu \in U$, the relevant way of removing $\mu$ from $U$
will be:
$$U \fatbslash \mu = \bigcup \ens{V \into U \aalt \ob(\el(V)) \cap
  \past{\mu} = \emptyset}.$$
$U \fatbslash \mu$ is thus the largest
subpresheaf of $U$ not containing any element of the past of
$\mu$. The good property of this operation is:
\begin{lem}\label{maxcore}
  If $\mu$ is a \emph{maximal} core in $G_U$ (i.e., there is no path to any
  further core) and $G_U$ is target-linear, then $(U \fatbslash \mu)
  (c) = U (c) \setminus \past{\mu}$ for all $c$.
\end{lem}

\begin{proof} The direction $(U \fatbslash \mu) (c) \subseteq U (c)
  \setminus \past{\mu}$ is by definition of $\fatbslash$. Conversely,
  it is enough to show that $c \mapsto U(c) \setminus \past{\mu}$
  forms a subpresheaf of $U$, i.e., that for any $f \colon c \to c'$
  in $\C$, and $x \in U(c') \setminus \past{\mu}$, $x \cdot f \notin
  \past{\mu}$. Assume on the contrary that $x' = x \cdot f \in
  \past{\mu}$. Then, of course $f$ cannot be the identity, and
  w.l.o.g.\ we may assume that $x'$ is \anagent{} and $x$ is a core.
  But then, because $x' \in \past{\mu}$, there is an edge $\mu \to x'$
  in $G_U$, and because $x' = x \cdot f$, there is an edge $x \to x'$
  or $x' \to x$ in $G_U$.  The former case is impossible by
  target-linearity, and the latter case would imply the existence of a
  path $\mu \to x$ in $G_U$, which contradicts the maximality of
  $\mu$. So $x' \in \past{\mu}$ is impossible altogether.
%
\end{proof}

\begin{proof}[Proof of Theorem~\ref{thm:completeness}]
  We proceed by induction on the number of \actions in $U$. If it is
  zero, then $U$ is a \position; by~\ref{cond:X}, $t$ is an iso,
  and by~\ref{cond:Y} so is $s$, hence the cospan is a \trace. For
  the induction step, we first decompose $U$ into
  $$Y \xto{s_2} U' \xot{t_2} Z \xto{s_1} M' \xot{t_1} X,$$
  and then show that $M'$ is \anaction and $U'$ satisfies the
  conditions of the theorem. 

  First, by acyclicity, pick a maximal core $\mu$ in $G_U$, i.e., one
  with no path to any other core. Let
  \begin{center}
    \diag{%
      \& I_0 \\
      Y_0 \& M_0 \& X_0 %
    }{%
      (m-1-2) edge (m-2-1) edge (m-2-2) edge (m-2-3) %
      (m-2-1) edge (m-2-2)
      (m-2-3) edge (m-2-2)
    }
  \end{center}
  be the seed with interface corresponding to $\mu$, so
  we have the Yoneda morphism $\mu\colon M_0 \to U$.

  Let $U' = (U \fatbslash \mu)$, and $X_1 = X - \Pl(X_0)$. $X_1$ is a
  subpresheaf of $X$, since it contains all channels. The square
  \begin{center}
    \Diag{%
      \pbk{m-2-1}{m-2-2}{m-1-2}  %
    }{%
         I_0 \& X_1 \\
         X_0  \& X  %
    }{%
      (m-1-1) edge[labelu={}] (m-1-2) %
      edge[labell={}] (m-2-1) %
      (m-2-1) edge[labeld={}] (m-2-2) %
      (m-1-2) edge[labelr={}] (m-2-2) %
    }
  \end{center}
  is a pushout, since it just adds the missing agents to $X_1$. 
  Define now $Z$, $M'$, $s_1$, and $t_1$ by the pushouts
\begin{center}
    \Diag(.3,.8){%
      \pbk{X0}{X}{X1} %
      \pullback{X1}{M'}{M0}{draw,-} %
      \pullback{X1}{Z}{Y}{draw,-} %
    }{%
      \& |(Y)| Y_0 \& \& |(Z)| {Z} \\
      \& \ \& \\
      \& |(M0)| M_0 \& \& |(M')| M' \& \& |(U)| U  \\
      |(I)| I_0 \&\& |(X1)| X_1 \\
      \& |(X0)| X_0 \& \& |(X)| X %
    }{%
      (U) edge[<-] (X) %
      (U) edge[<-] (X1) %
      (U) edge[<-] (M') %
      (U) edge[<-] (Z) %
      (X1) edge[] (X) %
      edge (M') %
      edge (Z) %
      (I) edge[] (X0) %
      edge (X1) %
      edge (M0) %
      edge (Y) %
      (Y) edge[fore] (Z) %
      (M0) edge[fore] (M') %
      (X0) edge (X) %
      (X) edge[dashed,fore,labelr={t_1}] (M') %
      (Z) edge[dashed,labelr={s_1}] (M') %
      (X0) edge[fore] (M0) %
      (Y) edge[fore] (M0) %
    }
\end{center}
and the induced arrows.  We further obtain arrows to $U$ by universal
property of pushout. 

Let us show that the arrow $f \colon M' \to U$ is mono.  First, it is
obviously mono in dimensions $> 1$.  It is also mono in dimension $1$,
because $M'[n] = X[n] + Y_0[n]$ for all $n$ and $X \to U$ is mono with
image consisting only of initial agents, which are thus dijoint from
the image of $Y_0$. Finally, for dimension $0$, i.e., at $\star$, the
pushout defining $M'$ is isomorphic to
  \begin{center}
    \Diag(1,.6){%
      \pullbackk{m-2-1}{m-2-2}{m-1-2}{draw,-,shorten <=.2cm} %
    }{%
        I_0(\star)  = X_0(\star) \& X_1(\star) = X(\star) \\
  M_0(\star) = X_0(\star) + I \&  M'(\star) = X(\star) + I  %
    }{%
(m-1-1) edge[labelu={}] (m-1-2) %
edge[labell={\injl}] (m-2-1) %
(m-2-1) edge[labeld={}] (m-2-2) %
(m-1-2) edge[labelr={\injl}] (m-2-2) %
    }
  \end{center}
  where $I = M_0(\star) \setminus X_0(\star)$ is the set of channels
  created by the \action.  Consider any $a,b \in M' (\star)$ such that
  $a \neq b$. Because $X \to U$ is mono, if $a,b \in X (\star)$ then
  $f (a) \neq f (b)$. By local 1-injectivity of $U$, if $a,b \in I$
  then $f(a) \neq f(b)$. Finally, if $a \in X(\star)$ and $b \in I$,
  then we have an edge $f(b) \to \mu$ in $G_U$, whereas $f (a)$ is
  initial by~\ref{cond:X}, so $f(a) \neq f(b)$.  This shows that
  $M' \to U$ is mono, which also entails that $Z \to U$ is a mono,
  because $s_1$ is a pushout of the mono $Y_0 \to M_0$.

By~\ref{cond:locinj} and
  Lemma~\ref{maxcore}, $U = M' \cup U'$, i.e., the square
  \begin{center}
    \Diag{%
      \pbk{m-2-1}{m-2-2}{m-1-2} %
    }{%
      Z \& U'  \\
      M'  \& U  %
    }{%
(m-1-1) edge[labelu={}] (m-1-2) %
edge[labell={}] (m-2-1) %
(m-2-1) edge[labeld={}] (m-2-2) %
(m-1-2) edge[labelr={}] (m-2-2) %
    }
  \end{center}
  is a pushout, so $U$ is indeed a composite as claimed, with $Z \into
  M' \otni X$ \anaction by construction. So, it remains to prove that
  $Y \into U' \otni Z$ satisfies the conditions. First, as a
  subpresheaf of $U$ (whose inclusion preserves cores), $U'$ is
  locally 1-injective and has a linear and acyclic causal graph, so
  satisfies~\ref{cond:locinj} and~\ref{cond:causality}. $U'$
  furthermore satisfies~\ref{cond:X} by construction of $Z$ and
  source-linearity of $G_U$, and~\ref{cond:Y} because removing
  $\past{\mu}$ cannot make any non-final agent final.
\end{proof}

Let us conclude this section with a helpful lemma, whose proof relies
on Theorem~\ref{thm:completeness}:
\begin{lem}\label{lem:posetal}
  There is at most one cell filling any diagram
  \begin{center}
    \diag{%
      Y' \& Y \\
      X' \& X %
    }{%
      (m-1-1) edge[labela={k}] (m-1-2) %
      edge[pro,twol={u'}] (m-2-1) %
      (m-2-1) edge[labelb={h}] (m-2-2) %
      (m-1-2) edge[pro,twor={u}] (m-2-2) %
    }
  \end{center}
  in $\D$.
\end{lem}
In order to prove this, let us introduce:
\begin{defi}\label{def:core:ass}
  For any \action $x \in U$, let $\core{x}$, the \emph{core
    associated to $x$}, be the unique core $\mu \in U$ for
  which there exists $f$ in $\C$ such that $\mu \cdot f = x$. If $x$
  is \anagent or a channel, then by definition $\core{x} = x$.
\end{defi}

\begin{proof}[Proof of Lemma~\ref{lem:posetal}]
  By definition, we have cospans $Y' \xto{s'} u' \xot{t'} X'$ and $Y
  \xto{s} u \xot{t} X$.  Suppose we are given $l, l' \colon u' \to u$
  making $(k,l,h)$ and $(k,l',h)$ into cells.  By naturality, $l$ and
  $l'$ are determined by their images on channels, \agents, and cores.
  We show by induction on the ordering induced by $G_{u'}$ that they
  have to agree on these.  For the base case: they have to agree with
  $h$ on initial \agents and channels by definition of cells. For the
  induction step, we proceed by case analysis on the kind of element
  to consider. The image of any source of or channel created by a core
  $\mu$ is uniquely determined by naturality, which leaves the case of
  a core $\mu$, of which we assume that there is \anagent $x$ such
  that $\mu \to x$ in $G_{u'}$ and $l(x) = l'(x)$. The edge $\mu \to
  x$ yields a morphism, say $t$, in $\C$ such that $\mu \cdot t = x$.
  But then by naturality we have $l(\mu) \cdot t = l (x) = l'(x) = l'
  (\mu) \cdot t$.  By linearity of $G_u$ we have $\core{l (\mu)} =
  \core{l' (\mu)}$.  Now let $c_\mu$ denote the object of $\C$ over
  which $\mu$ lies, and let $c'$ be the one over which $\core{l
    (\mu)}$ lies.  By inspection of $\C$, there is exactly one
  morphism $f \colon c_\mu \to c'$, and so we have $l(\mu) = \core{l
    (\mu)} \cdot f = \core{l' (\mu)} \cdot f = l(\mu),$ as desired.
\end{proof}

\subsection{A playground}\label{subsec:playground:pi}
In this section, we finally prove:
\begin{thm} $\D$ forms a playground.  \end{thm} Most axioms have been
proved in previous sections, and we are left with both decomposition
axioms, which are proved in Lemmas~\ref{right decomposition}
and~\ref{lem:decompleft} below, relying on the correctness criterion
of the previous section.

\begin{lem}\label{right decomposition} \ 
  \axref{views:decomp}
  Any double cell as in the center below, where $B$ is a basic
    \action and $M$ is \anaction, decomposes in exactly one of the
    forms on the left and right:
    \begin{mathpar}
      \begin{minipage}[t]{0.18\linewidth}
        \centering \Diag {%
          \twocell[.4][.3]{B}{A}{X}{}{celllr={0.0}{0.0},bend
            right=30,labelbr={\alpha_1}} %
          \twocell[.4][.3]{C}{B}{Y}{}{celllr={0.0}{0.0},bend
            right=20,labelbr={\alpha_2}} %
        }{%
          |(A)| A \& |(X)| X \\
          |(B)| C \& |(Y)| Y \\
          |(C)| D \& |(Z)| Z %
        }{%
          (A) edge[pro] (B) %
          edge (X) %
          (B) edge (Y) %
          edge[pro] (C) %
          (C) edge (Z) %
          (X) edge[pro] (Y) %
          (Y) edge[pro] 
          (Z) %
        }
      \end{minipage}
      \and \xotsdael{} \and
      \begin{minipage}[t]{0.18\linewidth}
        \centering \Diag{%
          \twocellbr{B}{A}{X}{\alpha} %
        }{%
          |(A)| A \& |(X)| X \\
          |(B)| C \& |(Y)| Y \\
          |(C)| D \& |(Z)| Z %
        }{%
          (A) edge[labelu={h}] (X) %
          edge[pro,labell={U}] (B) %
          (B) edge[pro,labell={B}] (C) %
          (X) edge[pro,labelr={V}] (Y) %
          (Y) edge[pro,labelr={M}] (Z) %
          (C) edge[labeld={k}] (Z) %
        }
      \end{minipage}
      \and \xleadsto{} \and
      \begin{minipage}[t]{0.18\linewidth}
        \centering \Diag {%
          \twocell[.4][.3]{B}{A}{X}{}{celllr={0.0}{0.0},bend
            right=30,labelbr={\alpha_1}} %
          \twocell[.3][.4]{C}{Y}{Z}{}{celllr={0.0}{0.0},bend
            right=20,labeld={\alpha_2}} %
        }{%
          |(A)| A \& |(X)| X \\
          |(B)| C \& |(Y)| Y \\
          |(C)| D \& |(Z)| Z. %
        }{%
          (A) edge[pro] (B) %
          edge (X) %
          (B) 
          edge[pro] (C) %
          (C) edge (Z) %
          edge (Y) %
          (X) edge[pro] (Y) %
          (Y) 
          edge[pro] (Z) %
        }
      \end{minipage}
    \end{mathpar}
\end{lem} 

\begin{proof} For any element $a$ over $c \in \C$ of
any presheaf $F \in \Chat$, let its \emph{neighbourhood}
consist of all elements in the image of $a \colon c \to F$.

Let $b \in B$ and $m \in M$ be the unique cores of $B$ and $M$,
respectively.  Let $V_m$ be the neighbourhood of $m$ in $M \vrond V$. 

If $\alpha(b) \in V_m$, let us show that the whole of $U$ is mapped to $V$, and
we are in the left-hand case. It is clear for channels. If there
exists an element $x$ of $U$ of dimension $\geq 1$ mapped to $y$ in
$M - V$, i.e., $M - Y$, then we obtain a path $x \to x'$ to an agent
$x'$ of $C$, in $G_{B \vrond U}$. Via $\alpha$, this yields a path $M - Y \to
Y$ in $G_{M \vrond V}$ between elements of dimension $\geq 1$, a
contradiction.

If now $\alpha(b) \notin V_m$, we show similarly that the whole of $B \vrond U$
is mapped to $V$, because the contrary would imply the existence of a
path $M - Y \to V - Y$ in $G_{M \vrond V}$, which also is a
contradiction.  Hence, we are in the right-hand case.  \end{proof}

\begin{lem}\label{lem:decompleft}\ 
\axref{leftdecomposition}     Any double cell
    \begin{center}
      \Diag{%
        \twocellbr{B}{A}{X}{\alpha} %
      }{%
        |(A)| A \& |(X)| X \\
        \& |(Y)| Y \\
        |(B)| B \& |(Z)| Z %
      }{%
        (A) edge[labelu={h}] (X) %
        edge[pro,labell={U}] (B) %
        (X) edge[pro,labelr={W_1}] (Y) %
        (Y) edge[pro,labelr={W_2}] (Z) %
        (B) edge[labeld={k}] (Z) %
      } \hfil decomposes as \hfil \Diag(1,2){%
        \twocellbr[.3]{C}{A}{X}{\alpha_1} %
        \twocellbr[.3]{B}{C}{Y}{\alpha_2} %
        \twocell[.3]{L}{A}{C}{}{celllr={0}{0},bend
          right,fore,labelbl={\scriptscriptstyle \alpha_3}} %
      }{%
        \& |(A)| A \& |(X)| X \\
        |(L)| \& |(C)| C \& |(Y)| Y \\
        \& |(B)| B \& |(Z)| Z %
      }{%
        (A) edge[labelu={h}] (X) %
        edge[pro] node[pos=.5,anchor=north west] {$\scriptscriptstyle
          {U_1}$} (C) %
        edge[pro,bend right=70,labell={U}] (B) %
        (C) edge[pro] node[pos=.5,anchor=north west]
        {$\scriptscriptstyle {U_2}$} (B) %
        edge[labelu={f}] (Y) %
        (X) edge[pro,labelr={W_1}] (Y) %
        (Y) edge[pro,labelr={W_2}] (Z) %
        (B) edge[labeld={k}] (Z) %
      }
    \end{center}
    with $\alpha_3$ an isomorphism, in an essentially unique way.
\end{lem}

  \begin{figure}[t]
    \centering
      \Diag (.4,.6) {%
        \pbk[.6cm]{w_1}{w}{w_2} %
        \path[->,draw] %
        (C) edge[labelbat={f_m}{0.8}] (Y) %
        (Y) edge (w_1) %
        edge (w_2) %
        (w_1) edge (w) %
        (w_2) edge (w) %
        (C) edge[fore] (u_1) %
        edge[fore] (u_2) %
        (u_1) edge[fore] (u) %
        (u_2) edge[fore] (u) %
        (u) edge[fore,labelu={f}] (w) %
        (u_1) edge[fore,labelu={f_1}] (w_1) %
        (u_2) edge[fore,labeld={f_2}] (w_2) %
        ; %
        \pbk[.6cm]{u}{u_1}{w_1} %
        \pullback[.6cm]{u_1}{u}{u_2}{draw,-,fore} %
        \path[->,draw] %
        (A) edge[] (u_1) %
        edge[bend left,fore] (u) %
        edge[labelu={h}] (X) %
        (X) edge (w_1) %
        edge[bend left,fore] (w) %
        (B) edge[] (u_2) %
        edge[bend right,fore] (u) %
        edge[labeld={k}] (Z) %
        (Z) edge (w_2) %
        edge[bend right,fore] (w) %
        ; %
      }{%
  \& \& |(A)| A \& \& \& \& \& \& \& |(X)| X \& \\ %
  \& \& \& \& \& \& \& \& \& \& \\ %
  \& \& \& \& \& \& \& \& \& \& \\ %
  \& \& |(u_1)| U_1 \& \& \& \& \& \& \& |(w_1)| W_1 \& \\ %
  |(C)| C \& \& \& \& \& \& \& |(Y)| Y \& \& \& \\ %
  \& \& \& \& \& \& \& \& \& \& \\ %
  \& \& \& |(u)| U \& \& \& \& \& \& \& |(w)| W \\ %
  \& |(u_2)| U_2 \& \& \& \& \& \& \& |(w_2)| M \& \& \\ %
  \& \& \& \& \& \& \& \& \& \& \\ %
  \& \& \& \& \& \& \& \& \& \& \\ %
  \& |(B)| B \& \& \& \& \& \& \& |(Z)| Z, \& %
     }{%
      }
    \caption{Proof of Lemma~\ref{lem:decompleft}}
\label{fig:proof:lem:decompleft}
\end{figure}%

\begin{proof}
  Let $\alpha = (h,f,k)$.  We first treat the case where $W_2$ is
  \anaction $M$, recalling Definition~\ref{def:core:ass}.  We
  construct $U_1$ and $U_2$, as depicted in
  Figure~\ref{fig:proof:lem:decompleft}.
First, let $U_1 = U \times_W W_1$, where $W = W_2 \vrond W_1$, and let
$A \to U_1$ denote the induced arrow. By construction, all of
$A \to U_1 \to U$ are monos and, by Lemma~\ref{upper squares all
  pullbacks} and the pullback lemma, $A = U_1 \times_{W_1} X$.

Let us show that the projection $f_1\colon U_1 \to W_1$ preserves
initiality of channels and agents. We proceed by contrapositive:
consider any channel or agent $x \in U_1$. If $x' = f_1 (x)$ is not
initial in $W_1$, then we have an edge $x' \to \mu'$ for some core
$\mu'$ of $W_1$.  But, since $f$ is a morphism between \traces, it
preserves initiality, so $x$ cannot be initial in $U$, hence we find
$x \to \mu$ in $G_U$. By source linearity of $G_W$, $\core{f (\mu)} =
\mu'$, so the \action $f(\mu) \in W$ has antecedents both in $U$ and
$W_1$. By universal property of pullback, there exists \anaction $y
\in U_1$, respectively mapped to $\mu$ and $f(\mu)$, which by
naturality and injectivity of $U_1 \to U$ entails that $x \to y$ in
$G_{U_1}$.  Therefore, $x$ is not initial in $G_{U_1}$, as required.

Let now $C \into U_1$ denote the subpresheaf of $U_1$ consisting of
initial channels and agents (a subpresheaf because if $x$ is an
initial, $n$-ary \agent, then $x \cdot s_i$ is an initial channel for
any $i \in n$, by Theorem~\ref{thm:completeness} and
Lemma~\ref{lem:final}).  Since $f_1$ preserves initiality, $C \to U_1
\to W_1$ factors through $Y \to W_1$, uniquely since the latter is
mono, say as $f_m$ (see Figure~\ref{fig:proof:lem:decompleft}).  By
Theorem~\ref{thm:completeness}, $A \to U_1 \ot C$ is a \trace and
$(h,f_1,f_m)$ defines a morphism to $X \to W_1 \ot Y$.

Let then $U_2 \into U$ denote the subpresheaf of $U$ consisting of elements 
below $C$ in $G_U$, i.e., 
\begin{center}
  ${x \in U_2} \Leftrightarrow {\exists c \in C. c \tostar_{G_U} \core{x}}$.
  \end{center}

  A first observation is that all initial channels and \agents of $U$
  are in $U_2$, so that $B \to U$ factors through $U_2$. Indeed,
  consider any such initial $x$.  By acyclicity of $G_U$, each initial
  element is reachable from some final element, so $x$ is reachable
  from some final $y$. But by source-linearity the corresponding path
  $y \tostar x$ goes through $C$, so we find a path $c \tostar x$ for
  some $c \in C$, as desired.

  Now, because $U_2 \to U$ and $M \to W$ are monos, showing that $f$
  maps all elements $x$ of $U_2$ to $M$ will imply that $U_2 \to U \to
  W$ uniquely factors through $M \to W$. Let us do this by case
  analysis:
  \begin{itemize}
  \item If $x$ is not a channel, then $f$ preserves paths from \agents
    to $x$, so we find some path $f(c) \tostar_{G_W} \core{f(x)}$ with
    $c \in C$ hence $f(c) \in Y$, which implies that $f(x) \in M$ ($f
    (x) \in W_1 - M$ would contradict initiality of $Y$ in $W_1$).
  \item If $x$ is some channel initial in $U$, then since $f$ preserves initiality
    $x$ is mapped to $Z$ hence to $M$.
  \item If finally $x$ is some non-initial channel, then $x \to \mu$
    for some core $\mu \in U$.  Now $\mu \in U_2$, as witnessed by the
    path $c \tostar x \to \mu$.  But then $x = \mu \cdot u$ for some
    morphism $u$ of $\C$, so since by the above $f(\mu) \in M$, we
    have that $f(x) = f (\mu) \cdot u$ is in $M$ too, as desired.
  \end{itemize}
  We thus get a diagram as in Figure~\ref{fig:proof:lem:decompleft},
  which commutes because $M \to W$ is mono.

By Theorem~\ref{thm:completeness}, $C \to U_2 \ot B$ is a \trace, and
$U = U_2 \vrond U_1$, which shows existence of the desired
decomposition.

For any decomposition as in Figure~\ref{fig:proof:lem:decompleft}, we
have $C = U_2 \times_M Y$ by Lemma~\ref{upper squares all pullbacks},
so by Corollary~\ref{cor:cube pullback}, we also have
$U_1 = U \times_W W_1$. Thus, $U_1$ is uniquely determined up to
canonical isomorphism.  But by Theorem~\ref{thm:completeness},
$C \to U_1$ is so too, as the subobject of initial agents and
channels. But then $U_2$ precisely consists of elements below
$C$. Indeed, by finiteness of $G_{U_2}$ and~\ref{cond:Y} in
Theorem~\ref{thm:completeness}, all of $U_2$ clearly lies below $C$.
Conversely, for any $x \in U_1 - U_2$, by finiteness of
$G_{U_1}$ and~\ref{cond:Y} in Theorem~\ref{thm:completeness}, we have
a path $x \to^+ c$ to some $c \in C$, so $x$ cannot lie below $C$
by~\ref{cond:causality}.  Our decomposition is thus unique up to
canonical isomorphism.
\end{proof}

\section{A sheaf model}\label{sec:model}
In the previous sections, we have constructed a double category $\D$
and equipped it with playground structure.  We now instantiate
constructions from~\cite{HirschoDoubleCats} on $\D$, which lead to the
definition of our sheaf model for $\pi$.  We first recall various
notions of strategy in Section~\ref{subsec:behaviours}: naive
strategies, innocent strategies, and behaviours.  Behaviours are
further studied in Section~\ref{subsec:decomp}, where we introduce a
kind of calculus for them.  Using this calculus, we then define our
intepretation of $\pi$ in Section~\ref{subsec:interp:pi}.  Finally, in
Section~\ref{subsec:semantic:fair}, we state our semantic definition
of fair testing equivalence and our main result.

\subsection{\Stratglobales and \stratlocales}\label{subsec:behaviours}
We first recall notions of strategies.
As announced in the introduction, we define a category $\Plays(X)$
combining prefix ordering and isomorphism of \traces: $\Plays(X)$ has
\traces $\trasse \colon Y \proto X$ as objects, and as morphisms
$\trasse \to \trasse'$ all pairs $(w,\alpha)$ with $w \colon Y' \proto
Y$ and $\alpha$ an isomorphism $\trasse \vrond w \to \trasse'$ in the
hom-category $\Dv(Y',X)$, as in
    \begin{center}
      \Diag(.4,.4){%
      }{%
        \& |(Y')| Y' \\
        |(Y)| Y \&  \\
        \& |(X)| X\rlap{,} %
      }{%
        (Y') edge[pro,labell={w},bend right] (Y) %
        (Y) edge[pro,labelbl={\trasse},bend right] (X) %
        (Y') edge[pro,twor={\trasse'},bend left] (X) %
        (Y) edge[cell={0.1},labela={\iso},labelb={\alpha}] (r) %
      }
    \end{center}
considered up to the smallest equivalence relation identifying
$(w,\alpha)$ and $(w',\alpha \rond (u \vrond \gamma))$, for any
$w'\colon Y' \proto Y$ and special $\gamma\colon w' \to w$). Thus,
$u'$ is an extension of $u$ by $w$.
\begin{definition}
  Let the category of \emph{(naive) \stratglobales} on $X$
  be $\FPsh{\Plays(X)}$.
\end{definition}
 \Stratglobales do not yield a
satisfactory model for $\pi$:
\begin{example}\label{ex:noninnocent}
  Consider the \position $X$ with three \agents $x,y,z$ sharing a
  channel $a$, and the following \traces on it: in $\trasse_{x,y}$, $x$
  sends $a$ on $a$, and $y$ receives it; in $\trasse_{x,z}$, $x$ sends $a$
  on $a$, and $z$ receives it; in $i_z$, $z$ inputs on $a$. One may
  define a \stratglobale $S$ mapping $\trasse_{x,y}$ and $i_z$ to a
  singleton, and $\trasse_{x,z}$ to $\emptyset$. Because $\trasse_{x,y}$ is
  accepted, $x$ accepts to send $a$ on $a$; and because $i_z$ is
  accepted, $z$ accepts to input on $a$. The problem is that $S$
  rejecting $\trasse_{x,z}$ roughly amounts to $x$ refusing to synchronise
  with $z$, or conversely.
%
\end{example}

We want to rule out this kind of \stratglobale from our model, by
adapting the idea of innocence. We start by extending $\T(X)$ with
objects representing \traces on \subpositions of $X$.  For this, we
consider the following category $\Plays_X$.  It has as objects pairs
$(\trasse, h)$ of \atrace $\trasse \colon Z \proto Y$ and a morphism
$h \colon Y \to X$ in $\Dh$.  A morphism $(\trasse,h) \to
(\trasse',h')$ consists of a \trace $w \colon T \proto Z$ and a cell
as below left with $h' r = h$. Morphisms are considered up to the
smallest equivalence relation identifying $(w,\alpha)$ with
$(w',\alpha \rond (u \vrond \gamma))$, for any $w'$ and $\gamma$ as
below right.
\begin{center}
    \Diag(.3,.3){%
      \node[inner sep=2pt] (XX) at ($(m-3-2.south) - (0,.3)$) {$X$}
      ; \path[->,draw] %
      (X) edge[labelbl={h}] (XX) %
      (X') edge[labelbr={h'}] (XX) %
      ; %
    }{%
      |(Z)| T \& \& |(Y')| Z' \\
      |(Y)| Z \& \&  \\
      |(X)| Y \& {\ } \& |(X')| Y' }{%
      (Z) edge[pro,labell={w}] (Y) %
      (Y) edge[pro,labell={\trasse}] (X) %
      (Y') edge[pro,twor={\trasse'}] (X') %
      (X) edge[labela={r}] (X') %
      (Z) edge[labela={s}] (Y') %
      (Y) edge[cell={0.2},labela={\alpha}] (r) %
    }
\hfil
    \Diag(.3,.3){%
      \node[inner sep=2pt] (XX) at ($(m-3-3.south) - (0,.3)$) {$X$}
      ; \path[->,draw] %
      (X) edge[labelbl={h}] (XX) %
      (X') edge[labelbr={h'}] (XX) %
      ; %
    }{%
      |(Z')| T' \& |(Z)| T \& \& |(Y')| Z' \\
      \& |(Y)| Z \& \&  \\
      \& |(X)| Y \& {\ } \& |(X')| Y' }{%
      (Z) edge[pro,labelr={w}] node[pos=.47] (w) {} (Y) %
      (Z') edge[pro,bend right=10,labelbl={w'}] node[pos=.3] (bl) {} (Y) %
      edge (Z) %
      (Y) edge[pro,twol={\trasse}] (X) %
      (Y') edge[pro,twor={\trasse'}] (X') %
      (X) edge[labela={r}] (X') %
      (Z) edge[labela={s}] (Y') %
      (Y) edge[cell={0.2},labela={\alpha}] (r) %
      (bl) edge[cell=0,labela={\gamma\ }] (w)
    }
\end{center}


  \begin{example}
    Recalling the right-hand \trace of Figure~\ref{fig:plays}
    (page~\pageref{fig:plays}), say $\trasse \colon Y \proto X$, $y$'s
    first \action is an input on its unique channel $b$.  This yields
    \atrace $\iotaof{1}{1} \colon [2] \proto [1]$. Here is an examlpe morphism
    $(\iotaof{1}{1},y) \to (\trasse,\id_X)$ in $\Plays_X$:
    \begin{center}
      \Diag|baseline=(m-1-1.base)|(.3,.3){%
      \node[inner sep=2pt] (XX) at ($(m-3-2.south) - (0,.3)$)
      {$X\rlap{.}$} ; \path[->,draw] %
      (X) edge[labelbl={y}] (XX) %
      (X') edge[identity] (XX) %
      ; %
    }{%
      |(Z)| [3] \& \& |(Y')| Y \\
      |(Y)| [2] \& \&  \\
      |(X)| [1] \& {\ } \& |(X')| X\rlap{.} %
    }{%
      (Z) edge[pro,labell={\iotaof{2}{2}}] (Y) %
      (Y) edge[pro,labell={\iotaof{1}{1}}] (X) %
      (Y') edge[pro,twor={\trasse}] (X') %
      (X) edge[labela={y}] (X') %
      (Z) edge[labela={y''}] (Y') %
      (Y) edge[cell={0.2}] (r) %
    }
    \end{center}
    We think of it as an occurrence of the \trace $\iotaof{1}{1}$
    in $\trasse$. Thus, morphisms in $\Plays_X$ account both for
    prefix inclusion and for `spatial' inclusion, i.e., inclusion of
    \atrace into some other \trace on a larger
    \position. \end{example}

We now  define \emph{\threads} within $\T_X$:
\begin{definition}\label{def:basicseeds}
  A \emph{\thread} is a \trace isomorphic to some (possibly empty)
  composite of basic \actions (Definition~\ref{def:playground:data}).
  Let $\Views_X$ denote the full subcategory of $\T_X$ spanning pairs
  $(u,h)$ where $u$ is a \thread.
\end{definition}
Intuitively, basic \actions follow exactly one \agent through \anaction.
An object of $\Views_X$ consists of a view, say $v \colon [n'] \proto
[n]$, plus a morphism $h \colon [n] \to X$ in $\Dh$, which by Yoneda
is just \anagent of $X$. So an object of $\V_X$ is just \anagent of
$X$ and a view from it.  
\begin{definition}
  The inclusion $\jj_X \colon \V_X \into \T_X$ induces a Grothendieck
  topology, for which a sieve $((u_i,h_i) \xto{(w_i,\alpha_i)}
  (u,h))_{i \in I}$ of morphisms to some \trace $u$ is \emph{covering}
  iff it contains all morphisms from \threads into $u$. Let the
  category $\SSX \into \FPsh{\T_X}$ of \emph{innocent} strategies be
  the category of sheaves of finite sets for this topology.  Let the
  category $\BB_{X}$ of \emph{\stratlocales} over $X$ be
  $\FPsh{\Views_X}$.
\end{definition}

As announced in the introduction, letting $\ran_{\op{\jj_X}}$ denote
\emph{right Kan extension}~\cite{MacLane:cwm} along the inclusion
$\op{\jj_X} \colon \op{\views_X} \into \op{\plays_X}$, we have:
\begin{prop}\label{prop:beh:inn}
  The embedding $\ran_{\op{\jj_X}}$ induces an equivalence of
  categories $\Beh{X} \equi \SS_X$.
\end{prop}
\begin{proof}
  See~\cite[Lemma~4.34]{HirschoDoubleCats}.
\end{proof}


We thus obtain the innocent \stratglobale $S_B$ associated to a
\stratlocale $B \in \BB_{X}$ by taking its right Kan extension as in
\begin{center}
        \Diag(.6,.5){%
        }{%
       |(VX)| \op{\views_X} \&|(P_X)| \op{\plays_X} \&|(PX)| \op{\plays(X)} \\
        \&|(set)| \set\rlap{.}
      }{%
        (VX) edge[bend right=20,labelbl={B}] (set) %
        edge[into,labela={\ \op{\jj_X}}] (P_X) %
        (PX) edge[linto,labela={\op{\kk_X}}] (P_X) %
        edge[bend left=20,labelbr={\UUU(S_B)}] (set) %
        (P_X) edge[labelon={S_B}] (set) %
      }

\end{center}
Explicitly, using standard results, we obtain the
\emph{end} $$S_B(\trasse,h) = \int_{(v,x) \in \views_X} B
(v,x)^{\plays_X((v,x),(\trasse,h))},$$ which is a kind of generalised
product.  In the boolean setting (functors to $\two$), this end
reduces to the conjunction $\bigwedge_{\ens{(v,x) \in \views_X \aalt
    \exists \alpha \colon (v,x) \to (\trasse,h)}} B(v,x)$, demanding
precisely that all \threads of $\trasse$ are accepted by $B$.  In the
general case, the intuition is that a way of accepting $\trasse$ for
$S_B$ is a compatible family of ways of accepting the \threads of
$\trasse$ for $B$.
The forgetful functor $\UUU$ to naive strategies is then given by
restricting along the opposite of $\kk_X\colon \plays(X) \into
\plays_X$ as above.  Some local information may be forgotten by
$\UUU$, which is hence neither injective on objects, nor full, nor
faithful.  E.g., if two \stratlocales differ on one \agent, but are both empty on
the \threads of another, then both are mapped to the empty naive
\stratglobale.
\begin{example}
  Recalling $X$ and $S$ from Example~\ref{ex:noninnocent}, let us show
  that for any $B \in \BB_{X}$, the associated \stratglobale
  $\UUU(S_B) \in \NaiveStrats{X}$ cannot be $S$. Indeed,
  if $\UUU(S_B)$ was $S$, then because $S$ accepts $u_{x,y}$ and $i_z$, $B$ has 
  to accept the following views: (1) $i_z$, (2) $o_x$, in which $x$ sends $a$ on $a$
  (without any matching input), (3) $i_y$, in which $y$ inputs on $a$,
  and (4) all identity views on $x$, $y$, and $z$. But then $\UUU (S_B)$
  has to accept both $\trasse_{x,y}$ and $\trasse_{x,z}$, because $B$ accepts all
  \threads mapping into them.
\end{example}

\subsection{Decomposing behaviours}\label{subsec:decomp}

In this section, we study behaviours a bit more in depth, which yields
the calculus announced at the beginning of Section~\ref{sec:model}.
The starting point is that the assignment $X \mapsto \Beh{X}$ may be
equipped with useful structure, describing how a behaviour $B$ on some
given \position restricts to any sub\position, and also what
remains of it after a given \action has been played.  Otherwise said,
morphisms of $\D$ act contravariantly on behaviours:
\begin{itemize}
\item horizontal morphisms $h\colon X \to X'$ induce functors
  $\Beh{h}\colon \Beh{X'} \to \Beh{X}$, and
\item vertical morphisms $u\colon Y \proto X$ induce functors
  $\Beh{u}\colon \Beh{X} \to \Beh{Y}$.
\end{itemize}
Furthermore, any cell as below left induces
a natural isomorphism as below right:
\begin{mathpar}
  \Diag{%
    \twocellbr{m-2-1}{m-1-1}{m-1-2}{\alpha} %
  }{%
      Y \& Y' \\
      X \& X' %
    }{%
      (m-1-1) edge[labela={k}] (m-1-2) %
      edge[pro,labell={u}] (m-2-1) %
      (m-2-1) edge[labelb={h}] (m-2-2) %
      (m-1-2) edge[pro,labelr={u'}] (m-2-2) %
    }
    \and
  \Diag{%
    \twocellbr{m-2-1}{m-1-1}{m-1-2}{\iso} %
  }{%
    \Beh{Y} \& \Beh{Y'} \\
      \Beh{X} \& \Beh{X'}\rlap{,} %
    }{%
      (m-1-1) edge[<-,labela={\Beh{k}}] (m-1-2) %
      edge[<-,labell={\Beh{u}}] (m-2-1) %
      (m-2-1) edge[<-,labelb={\Beh{h}}] (m-2-2) %
      (m-1-2) edge[<-,labelr={\Beh{u'}}] (m-2-2) %
    }
\end{mathpar}
which notably says that $B \cdot u' \cdot k \iso B \cdot h \cdot u$
for any behaviour $B \in \Beh{X'}$.  This is worked out in detail and
in full generality in~\cite{HirschoDoubleCats}, and extended to a
pseudo double functor $\op{\D} \to \QCat$, where $\QCat$ denotes
Ehresmann's double category of \emph{quintets} over $\Cat$.  But let
us explain how both actions look like in the present, concrete case.

The first, horizontal action is really easy: any horizontal morphism
$k \colon X' \to X$ acts on a given behaviour $B \in \Beh{X}$ by
returning the behaviour $B \cdot k$ such that for all
$(v,h) \in \V_{X'}$, $(B \cdot k) (v,h) = B (v, k \rond h)$.

\begin{prop}\label{prop:spatialdecomp}
  The functor $\Beh{X} \to \prod_{n, x \colon [n] \to X} \Beh{[n]}$
  given at $(n,x)$ by horizontal action of $x$, i.e., $B \mapsto B
  \cdot x$, is an isomorphism.
\end{prop}
\begin{proof}
  We have $\Views_X \iso \sum_{n,x \colon [n] \to X} \Views_{[n]}$.
\end{proof}
\begin{notation}\label{not:horizaction}
  If $(B_x)_{x \colon [n] \to X}$ is a family of behaviours indexed by
  the agents of $X$, we accordingly denote its unique antecedent by
  $[B_x]_{x \colon [n] \to X}$.
\end{notation}

Vertical action is a bit harder. Let us start by recalling the
following result from~\cite{HirschoDoubleCats}, which generalises
Lemma~\ref{lem:views}:
\begin{lem}\label{lem:views:gen}
  For any \trace{} $P \colon Y \proto X$ and $y \colon d \to Y$ in
  $\Dh$ with $d\in \DI$, there exists an essentially unique cell
\begin{center}
  \diag{%
    |(d)| d \& |(Y)| Y \\
    |(dMy)| d^{y,P} \& |(X)| X, %
  }{%
    (d) edge[labelu={y}] (Y) %
    edge[pro,dashed,twol={v^{y,P}}] (dMy) %
    (Y) edge[pro,twor={P}] (X) %
    (dMy) edge[dashed,labeld={y^P}] (X) %
    (l) edge[cell=.3,dashed,labelu={\scriptstyle \alpha^{y,P}}] (r) %
  }
\end{center}
with $v^{y,P}$ a view.
\end{lem}
\begin{proof}
  By induction on $P$ using
  Lemma~\ref{lem:views}. See~\cite{HirschoDoubleCats} for the precise
  meaning of essential uniqueness in this case.
\end{proof}

This result allows us to describe what remains of a behaviour
after \atrace:
\begin{defi}\label{def:residual}
  Given any behaviour $B \in \Beh{X}$ and \trace{}
  $P\colon Y \proto X$, the \emph{residual} $B \cdot P$ of $B$ along $P$
  is the behaviour determined (up to isomorphism) by
  $$(B \cdot P \cdot y) (w,\id_{[n_y]}) = B (v^{y,P} \vrond w,y^P)$$
  for all \agents{} $y\colon [n_y] \to Y$.
\end{defi}

So we may consider residuals of behaviours along arbitrary \traces{}.
Conversely, any behaviour is determined by its initial states and
residuals. Let us first consider the following special kind of
behaviour:
\begin{defi}\label{def:definite}
  A behaviour $B$ on $[n]$ is \emph{definite} iff
  $B (\idv,\id_{[n]}) \iso 1$.  A behaviour $B$ on an arbitrary
  \position{} $X$ is \emph{definite} iff for all \agents{} $x$ of
  $X$, $B \cdot x$ is definite.  Let $\DD_X$ denote the full
  subcategory of $\Beh{X}$ spanning definite behaviours.
\end{defi}
This means that $B$ has exactly one initial state.
\begin{defi}
  For any $n \in \Nat$ and family $B_b \in \Beh{[n_b]}$ indexed by all
  basic \actions{} $b\colon [n_b] \proto [n]$, let $\tupling{B_b}_b$
  denote the definite behaviour $B$ determined by $B \cdot b = B_b$
  for all $b$.
\end{defi}

By construction, we have:
\begin{prop}
  For any definite behaviour $D \in \Beh{[n]}$,
  $D \iso \tupling{D \cdot b}_b$.
\end{prop}

This extends to arbitrary behaviours on individuals using the fact
that any behaviour on an individual is a coproduct of definite
behaviours.
\begin{notation}\label{not:cardrestr}
  For any behaviour $B = \sum_{k \in \gamma} D_k$ on any $[n]$, let
  $\card{B} = \gamma$ and, for any $k \in \gamma$, let the
  \emph{restriction} $\restr{B}{k}$ of $B$ to $k$ be $D_k$.
\end{notation}
\begin{rem}
  In other words, $k \in \gamma$ is just an element of
  $B (\idv_{[n]},\id_{[n]})$ and $\restr{B}{\state}$ is determined by
  $$\restr{B}{\state} (v,\id_{[n]})
  = \ens{\state' \in B (v, \id_{[n]}) \aalt \state' \cdot {!}_v =
    \state}\rlap{,}$$ where ${!}_v$ denotes the unique morphism
  $(\idv_{[n]},\id_{[n]}) \to (v, \id_{[n]})$ in $\Views_{[n]}$.
\end{rem}
We obtain:
\begin{prop}
  For any behaviour $B \in \Beh{[n]}$,
  $$B \iso \sum_{k \in \card{B}} \tupling{(\restr{B}{k}) \cdot b}_b\rlap{.}$$
\end{prop}

Putting this together with spatial decomposition, we obtain:
\begin{cor}
  For any behaviour $B \in \Beh{X}$,
  $$B \iso
  [\sum_{k \in \card{B \cdot x}} \tupling{(\restrat{B}{x}{k}) \cdot b}_{b\colon [n_b] \shortproto [n_x]}]_{x\colon [n_x] \to X}\rlap{.}$$
  
\end{cor}

The fact that both actions of $\D$ yield a pseudo double functor
$\op{\D} \to \QCat$ essentially boils down to:
\begin{lem}
  For any $B \in \Beh{X}$ and cell
  \begin{center}
    \diag{%
      Y' \& Y \\
      X' \& X\rlap{,} %
    }{%
      (m-1-1) edge[labela={k}] (m-1-2) %
      edge[pro,twol={u'}] (m-2-1) %
      (m-2-1) edge[labelb={h}] (m-2-2) %
      (m-1-2) edge[pro,twor={u}] (m-2-2) %
      (l) edge[cell=.4,labela={\alpha}] (r) %
    }
  \end{center}
  we have $D \cdot h \cdot u' \iso D \cdot u \cdot k$.
\end{lem}
The result is stated (and proved) in full generality
as~\cite[Proposition~4.31]{HirschoDoubleCats}.

\subsection{Interpretation of \texorpdfstring{$\pi$}{π}}\label{subsec:interp:pi}
We now define our interpretation of $\pi$-calculus configurations. We
start with processes and then cover configurations.  Because the
notation for behaviours introduced in the previous section only covers
behaviours on representable \positions{}, while $\pi$-calculus syntax
is name-based, we bridge the gap by keeping track, along the recursive
definition, of a bijection between the set $\gamma$ of free channels
of the considered process and its cardinal $\card{\gamma}$.  So we
define a family of maps in
$$\prod_{\gamma \in \powfin(\Nat)} \DD_{[\card{\gamma}]}^{\picalc_\gamma \times \Bij(\gamma,\card{\gamma})},$$
where $\Bij(A,B)$ denotes the set of all bijections $A \isoto B$.  For
any such $\gamma,P$, and $h$ in the domain, we denote the result by
$\transl{P}_h$, or $\transl{P}_{\gamma,h}$ when needed.
Letting $n = \card{\gamma}$, it is coinductively defined by
$$
\begin{array}{rcll}
  \transl{\sum_i P_i}_h & = & \dsum_i \transl{P_i}_h \\
  \transl{P \para Q}_h & = & \left \langle
    {\begin{array}[c]{rcl}
        \forklof{n} & \mapsto & \transl{P}_h \\
        \forkrof{n} & \mapsto & \transl{Q}_h 
      \end{array}} \right \rangle \\
  \transl{\nu b.P}_h & = & \left \langle
    {\begin{array}[c]{rcl}
        \nuof{n} & \mapsto & \transl{P}_{h'} \\
      \end{array}} \right \rangle \\
  \transl{\tick.P}_h & = & \left \langle
    {\begin{array}[c]{rcl}
        \tickof{n} & \mapsto & \transl{P}_{h} \\
      \end{array}} \right \rangle\\
  \transl{\tau.P}_h & = & \left \langle
    {\begin{array}[c]{rcl}
        \tauuof{n} & \mapsto & \transl{P}_{h} \\
      \end{array}} \right \rangle\\
  \transl{a(b).P}_h & = & \left \langle
    {\begin{array}[c]{rcl}
        \iota_{n, h(a)} & \mapsto & \transl{P}_{h'} \\
      \end{array}} \right \rangle\\
  \transl{\send{a}{b}.P}_h & = & \left \langle
    {\begin{array}[c]{rcl}
        o_{n, h(a),  h(b)} & \mapsto & \transl{P}_{h} \\
      \end{array}} \right \rangle\\
\end{array}$$
where
\begin{itemize}
\item in any list $\tupling{b_1 \mapsto B_1, \ldots, b_m \mapsto
    B_m}$, all unmentioned basic \actions are meant to be mapped to
  the empty behaviour;
\item the \emph{definite sum} $\dsum_i D_i$ of definite behaviours
  $D_i$ is the definite behaviour determined by
  $(\dsum_i D_i) \cdot b = \sum_i (D_i \cdot b)$, for all basic
  \actions $b \colon [n'] \proto [n]$;
\item and $h' \colon \gamma,b \isoto n+1$ maps any $a \in \gamma$ to $h(a)$,
  and $b$ to $n+1$.
\end{itemize}

\begin{example}
  Let us briefly illustrate the translation.  Consider any $h\colon
  \gamma \to \card{\gamma}$ and processes $\gamma,b \vdash P$,
  $\gamma,b \vdash Q$, and $\gamma \vdash R$, with $a \in \gamma$. We
  can form $a(b).P + a(b).Q + \send{a}{a}.R$, which is mapped to
$$\langle \iotaneg{\card{\gamma},h(a)} \mapsto (\transl{P}_{h'} + \transl{Q}_{h'}),
\iotapos{\gamma,h(a),h(a)} \mapsto \transl{R}_h \rangle.$$ To emphasise the
difference, using coproduct of behaviours instead of definite sum in
the translation would yield a behaviour with three distinct initial
states, closer to the internal choice
$a(b).P \oplus a(b).Q \oplus \send{a}{a}.R$ than to the original process.
\end{example}

Generalising this to configurations should really be intuitive: we map
any $\conf{\gamma}{P_1,\ldots,P_n}$ to some behaviour on the
\position{} $X$ with $X(\star) = \gamma$ and $n$ \agents{} of
arity $\card{\gamma}$, given for each \agent{} $i \in n$ by
$\transl{P_i}$.  In order to fully define such \aposition{}, we
need to specify maps $f_i\colon \card{\gamma} \to \gamma$.  We use for
all of them the inverse of the canonical bijection $h_\gamma$ defined by:
\begin{defi}\label{def:h_gamma}
  Let $h_\gamma \colon \gamma \isoto \card{\gamma}$ map each
  $a \in \gamma$ to its position in the ordering induced by the one on
  natural numbers.
\end{defi}
  We call the obtained position $X (\gamma,n)$.
\begin{defi}\label{def:transl}
  Let $\translfun \colon \ob(\Conf) \to \sum_X \Beh{X}$ map
  any configuration $C = \conf{\gamma}{P_1,\ldots,P_n}$
  to the pair $(X(\gamma,n), \transl{C})$, where
  $\transl{C}$ is defined through
  Proposition~\ref{prop:spatialdecomp} by
  $$\transl{C}(\card{\gamma})(i) = \transl{P_i}_{h_\gamma}\rlap{,}$$
  for all $i \in n$.  We implicitly consider processes $P$ over
  $\gamma$ as configurations $\conf{\gamma}{P}$, and hence allow
  ourselves to write $\transl{P}$ for $\transl{\conf{\gamma}{P}}$.
\end{defi}

\subsection{Semantic fair testing}\label{subsec:semantic:fair}
In order to state our main result, it remains to define our semantic
analogue of fair testing equivalence.  It rests on two main
ingredients: a notion of \emph{closed-world} \trace, and an analogue
of \emph{parallel composition} in game semantics.

The intuitive purpose of parallel composition is to let behaviours
interact.  If we partition the \agents of \aposition $X$ into two
teams, we obtain two \subpositions $X_1 \into X \otni X_2$, each
\agent of $X$ belonging to $X_1$ or $X_2$ according to its team. The
crucial fact is that the category $\views_X$ of \threads on $X$ is
isomorphic to the coproduct category $\views_{X_1} +
\views_{X_2}$. Parallel composition of any
$B_1 \in \FPsh{\Views_{X_1}}$ and $B_2 \in \FPsh{\Views_{X_2}}$ is
then simply given by copairing $[B_1,B_2]$ (following
Notation~\ref{not:horizaction}), as in
\begin{center}
  \diag(.6,.6){%
    \op{\views_{X_1}} \&     \op{\views_{X}} \&     \op{\views_{X_2}} \\
    \& \set\rlap{.} %
  }{%
    (m-1-1) edge[into] (m-1-2) %
    edge[labelbl={B_1}] (m-2-2) %
    (m-1-3) edge[linto] (m-1-2) %
    edge[labelbr={B_2}] (m-2-2) %
    (m-1-2) edge[labelon={[B_1,B_2]},dashed] (m-2-2) %
  }
\end{center}

We now describe closed-world \traces, which are then used as a
criterion for success of tests.  Closed-world \actions were
defined (Definition~\ref{def:playground:data}) as those not involving
any interaction with the environment, i.e., formally, pushouts of a
seed of any shape among $\nun$,$\taun$,$\tickn$,$\paran$, and
$\taunamcd$. \Atrace is \emph{closed-world} when it is a composite of
closed-world \actions.  Let $\W(X) \xinto{\ii_X} \Plays(X)$ denote the
full subcategory of $\Plays(X)$ consisting of closed-world \traces,
and let the category of \emph{closed-world \stratglobales} be
$\FPsh{\W(X)}$.
\begin{notation}\label{not:exta}
Recalling the discussion below Proposition~\ref{prop:beh:inn},
we denote by $B \mapsto \exta{X}{B}$ the composite functor
$$\FPsh{\Views_X} \xto{\ran{\op{\jj_X}}}
\FPsh{\Plays_X} \xto{\cob{{\op{\kk_X}}}} \FPsh{\Plays(X)}
\xto{\cob{{\op{\ii_X}}}} \FPsh{\W(X)}\rlap{,}$$ where $\cob{f}$
denotes restriction along $f$.
\end{notation}
A closed-world \trace is \emph{successful} when it contains a $\tick$
\action, and \emph{unsuccessful} otherwise. A state $\state \in
S(\trasse)$ of any $S \in \FPsh{\W(Z)}$ over a
closed-world \trace $\trasse \colon Z' \proto Z$ is successful iff
$\trasse$ is.  Define $\bbot_Z$ as the set of closed-world
\stratglobales $S \in \FPsh{\W(Z)}$ such that any unsuccessful
closed-world state admits a successful extension, i.e. $S \in \bbot_Z$
iff for all unsuccessful $u \in \W(Z)$ and $\state \in S (u)$, there
exists a successful $u'\in \W(Z)$, a morphism $f \colon u \to u'$, and
a state $\state' \in S (u')$ such that $\state' \cdot f = \state$.
Finally, in order to compare \stratlocales for semantic fair testing
equivalence, we specify what a test is for a given \stratlocale $B \in
\BB_X$. A \emph{test} consists of a \position $Y$ and a
\stratlocale $T \in \BB_Y$.  Recalling Definition~\ref{def:interface},
we say that the pair $(X,B)$, with $B \in \Beh{X}$, \emph{should pass}
the test $(Y,T)$ iff $I_X = I_Y$ and $\exta{Z}{[B,T]} \in \bbot{}_Z$,
where $Z$ is the pushout $X +_{I_X} Y$ ($X$ and $Y$ thus form two
teams on $Z$).
At last, we define \emph{semantic fair testing equivalence}, for any
$B \in \BB_X$ and $B' \in \BB_{X'}$:
\begin{definition}\label{def:fair:semantic}
  Let $(X,B) \faireq (X',B')$ iff they should pass the same tests.
\end{definition}

We may at last state:
\begin{theorem}\label{thm:main}
  The translation $\translfun \colon \ob(\Conf) \to \sum_X \Beh{X}$ is
  intensionally fully abstract for $\faireq$, i.e.,
  \begin{itemize}
  \item For all configurations $C_1$ and $C_2$, $C_1
    \faireqof{\picalc} C_2$ iff $\transl{C_1} \faireq \transl{C_2};$
  \item Furthermore, for all \positions{} $X$ and behaviours $B
    \in \Beh{X}$, there exists $C \in \picalc_{X(\star)}$ such that
    $\transl{C} \faireq B$.
  \end{itemize}
\end{theorem}

The proof is the subject of the next section.

\section{Intensional full abstraction}\label{sec:translation}
In the previous section, exploiting the playground structure of $\D$
established in Sections~\ref{sec:fib} and~\ref{sec:playground:pi}, we
have defined and studied the notion of behaviour, into which we have
translated $\pi$-calculus processes and configurations. We have then
defined our semantic analogue of fair testing equivalence and stated
our main result. We now work towards proving it.  In
Section~\ref{subsec:SSS}, we define a graph with testing $\SSS$ whose
vertices are pairs $(X,B)$ with $B$ a definite behaviour on the
\position{} $X$ (Definition~\ref{def:definite}), such that fair
testing equivalence in $\SSS$ coincides with fair testing equivalence
in the model. In order to prove that this is the case, we introduce an
intermediate graph with testing $\CCC$ which is in fact quite
intricate.  We are then in a position where our main result is reduced
to intensional full abstraction of a translation between two graphs
with testing, $\Conf$ and $\SSS$, for which we may hope to apply the
results of Section~\ref{subsec:fair}. In fact, in
Section~\ref{subsec:MMM}, we further reduce to a translation to a more
syntactic graph with testing, $\MMM$.  In
Section~\ref{subsec:fullabs}, we prove intensional full abstraction of
the translation to $\MMM$, from which we deduce
Theorem~\ref{thm:main}. Finally, we generalise to a large class of
testing equivalences in Section~\ref{subsec:gen}.


\subsection{A first graph with testing for behaviours}\label{subsec:SSS}

\begin{defi}\label{defi:SSS}
  Let $\SSS$ denote the graph with vertices in $\sum_X \DD_X$, where
  we recall (Definition~\ref{def:definite}) that $\DD_X$ denotes the
  category of definite behaviours on $X$, and with non-identity edges
  $(X,D)\ot (Y,D')$ all closed-world \actions{} $M \colon Y \proto X$
  such that for all \agents{} $y$ in $Y$, there exists $\state_y \in
  \card{D \cdot M \cdot y}$ such that
$$D' \cdot y \iso \restr{(D \cdot M \cdot y)}{\state_y}\rlap{.}$$

Moreover, let $\pp^\SSS: \SSS\rightarrow\Sierp$ denote the map sending
$(X,D)\xot{M}(X',D')$ to the $\tick$ edge in $\Sierp$ if $M$ is a tick
action and to $\tau$ otherwise.
\end{defi}

Let us now equip $\SSS$ with testing structure.
\begin{defi}\label{def:paraS}
  We define the relation $\para_\SSS$ by $(Z,D) \in
  ((X_1,D_1) \para_{\SSS} (X_2,D_2))$ iff $X_1 (\star) = X_2 (\star)$ and
  there is a pushout square
  \begin{center}
    \Diag{%
      \pbk{m-2-1}{m-2-2}{m-1-2} %
    }{%
      I_{X_1} \& X_2 \\
      X_1 \& Z %
    }{%
      (m-1-1) edge[labela={}] (m-1-2) %
      edge[labell={}] (m-2-1) %
      (m-2-1) edge[labelb={\injl}] (m-2-2) %
      (m-1-2) edge[labelr={\injr}] (m-2-2) %
    }
  \end{center}
  such that $D_1 \iso D \cdot \injl$ and $D_2 \iso D \cdot \injr$,
  or otherwise said $D \iso [D_1,D_2]$.
\end{defi}
Lemma~\ref{lem:freetesting} entails:
\begin{prop}\label{prop:testingS}
  The morphism $\freecat{\pp^\SSS} \colon \freecat{\SSS} \to
  \fcSierp$, with $\para_\SSS$ as testing relation, forms a free graph
  with testing.
\end{prop}
\begin{proof}
  We exhibit a weak bisimulation relating any two such pushouts
  $(Z,D)$ and $(Z',D')$.  The relation containing two such pairs as
  soon as there exists a horizontal isomorphism $h \colon Z \to Z'$
  such that $D \iso D' \cdot h$ does the job.
\end{proof}

The crucial result for proving that fair testing equivalence in $\SSS$
coincides with semantic fair testing equivalence is:
\begin{lem}\label{lem:bbot:bot}
  For any definite behaviour $D \in \DD_X$, we have $D \in \bbot_X$ iff
  $(X,D) \in \bot^\SSS$.
\end{lem}
Our strategy for proving this is to introduce and study an
intermediate graph with testing $\CCC$, which is closer to semantic
fair testing in that its transitions may comprise several \actions.
However, defining $\CCC$ just as $\SSS$ with arbitrary \traces{}
instead of \actions{} would be wrong:
\begin{example}\label{exa:psi:not:surjective}
  Consider the \trace $P = (\tau_0 \vrond \paraof{0})$ consisting of a
  nullary \agent performing a $\tau$ action and then forking. Consider
  now any definite behaviour $D$ such that $D(\tau_0) = D(\tau_0
  \vrond \paralof{0}) = D(\tau_0 \vrond \pararof{0}) = 2$, which maps
  both inclusions $\tau_0 \vrond \paralof{0} \otni \tau_0 \into \tau_0
  \vrond \pararof{0}$ to the identity.  Then $\exta{X}{D}(P) = 2$: it
  consists of pairs $(\state^l, \state^r)$ in $D(\tau_0
  \vrond \paralof{0}) \times D(\tau_0 \vrond \pararof{0})$ whose
  restrictions to $D(\tau_0)$ coincide, which leaves just $(1,1)$ and
  $(2,2)$.  Using the decompositions of Section~\ref{subsec:decomp},
  another way to say this is that $D = \tupling{\tau_0 \mapsto D_1 +
    D_2}$, with $D_i = \tupling{\paralof{0} \mapsto D^l_i, \pararof{0}
    \mapsto D^r_i}$, for $i = 1,2$.  In order for $\CCC$ to correspond
  to the model, assuming $D^l_1 \niso D^l_2$ and $D^r_1 \niso D^r_2$,
  there should be two transitions from $([0],D)$, one to $([0]\para
  [0], [D^l_1,D^r_1])$ and the other to $([0]\para [0],
  [D^l_2,D^r_2])$.  But if we naively generalise
  Definition~\ref{defi:SSS} to arbitrary \traces{}, we obtain
  additional, `incoherent' transitions, to $([0]\para [0],
  [D^l_1,D^r_2])$ and $([0]\para [0], [D^l_2,D^r_1])$.
\end{example}
Instead of relying on $\prod_{y \in \Pl(Y)} \card{D \cdot P
  \cdot y}$ as in Definition~\ref{defi:SSS}, we would like to rely on
$\exta{X}{D}(P)$. This may be done by constructing a map
\begin{equation}
\psi^D_P \colon \exta{X}{D} (P) \to \prod_{y \colon [n_{y}] \to Y} \card{D \cdot P \cdot y},\label{eq:psi}
\end{equation}
whose image will consist precisely of all `coherent' elements.  (From
now on, we omit the superscript $D$ when clear from context.)
Intuitively, this map associates to each global state the
corresponding family of local states.  Furthermore, $\psi_P$ is always
injective, but Example~\ref{exa:psi:not:surjective} shows that it is
not surjective in general. 
In order to construct $\psi_P$, recalling
Definition~\ref{def:residual}, we have by definition: $\prod_{y \colon
  [n_{y}] \to Y} \card{D \cdot P \cdot y} = \prod_{y \colon [n_{y}]
  \to Y} D (v^{y,P},y^P)$.  Furthermore, we also have
$$
\exta{X}{D} (P) \iso \int_{(v,x) \in \V_X} D (v,x)^{\Plays_X
  ((v,x),(P,\id_X))}\rlap{,}$$ which is a subset of $\prod_{(v,x) \in \V_X} D (v,x)^{\Plays_X ((v,x),(P,\id_X))}$.
We may thus define:
$$\psi_P (\state)(y) = \state (v^{y,P},y^P)(\idv_{[n_{y}]},\alpha^{y,P}).$$
Here, $\state (v^{y,P},y^P)$ is in $D
(v^{y,P},y^P)^{\Plays_X ((v^{y,P},y^P),(P,\id_X))}$.  So by
applying it to $\alpha^{y,P}$ viewed as a morphism in
$(v^{y,P},y^P) \to (P,\id_X)$ in $\Plays_X$, we obtain an element
of $D (v^{y,P},y^P) = \card{D \cdot P \cdot y}$, as desired.

We may now define the intermediate graph with testing $\CCC$.  We
first extend the notion of restriction:
\begin{notation}
  We extend Notation~\ref{not:cardrestr}: if $B \in \Beh{X}$ and $\state \in
  \prod_{n, x \colon [n] \to X} \card{B \cdot x}$, let $\restr{B}{\state}$ be
  defined up to isomorphism by
$$\restr{B}{\state} \cdot x = \restrat{B}{x}{\state(x)}.$$
\end{notation}

\begin{defi}\label{def:CCC}
  Let $\CCC$ denote the graph with $\ob (\CCC) = \ob (\SSS)$, and
  where $\CCC ((X',D'),(X,D))$ is the set of closed-world \traces $W
  \colon X' \proto X$ such that there exists a state $\state \in
  \exta{X}{D} (W)$ satisfying $\restrat{D}{W}{\psi_{W}(\state)} \iso
  D'$.
\end{defi}
Thus, $\CCC$ is a generalisation of $\SSS$ from closed-world \actions
to closed-world \traces.  Let us turn it into a graph with testing.
\begin{defi}\label{def:DvW}
  Let $\DvW$ denote the smallest locally full subbicategory of $\Dv$
  containing all closed-world \traces.  The graph morphism
  $\W \to \Sierp$, where we recall that $\W$ denotes the graph of
  closed-world \actions{} (Definition~\ref{def:playground:data}),
  extends to a pseudo functor
  $\pp^\W \colon \DvW \to \freecat{\Sierp}$, which essentially counts
  the number of ticks.  Let $\pp^\CCC\colon \CCC \to \DvW$ denote the
  obvious projection.
\end{defi}

\begin{prop}
  The composite projection
  $\CCC \xto{\pp^\CCC} \DvW \xto{\pp^\W} \fcSierp$, with $\para_\SSS$
  as testing relation, makes $\CCC$ into a graph with testing.
\end{prop}
\begin{proof}
  Just as Lemma~\ref{prop:testingS}.
\end{proof}

As announced, $\CCC$ is an example of a non-free graph with testing.
The rest of this section is devoted proving Lemma~\ref{lem:bbot:bot},
and reducing Theorem~\ref{thm:main} to a statement about $\SSS$.
Lemma~\ref{lem:bbot:bot} follows from the fact that both poles are
equivalent to $\bot^\CCC$, as we now set out to prove.  We start with
a lemma saying that $\CCC$ has essentially the same transitions over
two (specially) isomorphic \traces.

\begin{notation}\label{not:component}
  For any morphism $p \colon G \to H$ in $\Gph$, we denote by $p_{A,B}
  \colon G (A,B) \to H (p (A),p (B))$ the component of $p$ at $A$ and
  $B$.
\end{notation}
\begin{lem}\label{lem:2d}
  For any $(X,D),(X',D') \in \CCC$, if there exists any
  special isomorphism $W_1 \iso W_2$ in $\DvW(X',X)$, we
  have $$\inv{(\pp^\CCC)_{(X',D'),(X,D)}} (W_1) \iso
  \inv{(\pp^\CCC)_{(X',D'),(X,D)}} (W_2).$$
\end{lem}
\begin{proof}
  The given special isomorphism induces by pseudo double functoriality
  of $\Behfun$ an isomorphism $\varphi \colon D \cdot W_1 \isoto D \cdot
  W_2$, hence an isomorphism
 $$\exta{}{\varphi}_{\id_{X'}} \colon \exta{}{D}(W_1) \isoto \exta{}{D}(W_2).$$
 This isomorphism is such that for any $\state$, 
 $$\restrat{D}{W_1}{\psi_{W_1} (\state)}
 \iso \restrat{D}{W_2}{\psi_{W_2}
   (\exta{}{\varphi}_{\id_{X'}}(\state))}.$$ Thus,
 $\restrat{D}{W_1}{\psi_{W_1} (\state)} \iso D'$ iff
 $\restrat{D}{W_2}{\psi_{W_2} (\exta{}{\varphi}_{\id_{X'}}(\state))} \iso D'$.
\end{proof}

In order to relate $\CCC$ to $\SSS$, let us now show that transitions
in $\CCC$ behave well w.r.t.\ composition of \traces.  First,
transitions compose, and second, transitions over any composite
(closed-world) \trace{} $W_1 \vrond W_2$ always decompose into a
transition over $W_1$ followed by one over $W_2$.

\begin{lem}\label{lem:compo}
  For all edges $(X,D) \xot{W} (X',D') \xot{W'} (X'',D'')$ in $\CCC$,
  there is an edge $(X,D) \xot{W \vrond W'} (X'',D'')$.
\end{lem}
\begin{proof}
  Consider $\state \in \exta{X}{D} (W)$ such that
  $\restrat{D}{W}{\psi_{W}(\state)} \iso D'$ and $\state' \in
  \exta{X'}{D'} (W')$ such that $\restrat{D'}{W'}{\psi_{W'}(\state')}
  \iso D''$.  We want to construct an edge $W \vrond W' \colon
  (X'',D'') \to (X,D)$, i.e., find $\state'' \in \exta{X}{D} (W \vrond
  W')$ such that $\restrat{D}{(W \vrond W')}{\psi_{W \vrond W'}
    (\state'')} \iso D''$.  Now, the isomorphism $\varphi \colon
  \restrat{D}{W}{\psi_{W}(\state)} \isoto D'$ yields a state $\state_1
  = \inv{\exta{}{\varphi}_{W'}} (\state') \in
  \exta{X'}{\restrat{D}{W}{\psi_{W}(\state)}} (W')$ such that
  \begin{equation}
\restrat{\restrat{D}{W}{\psi_{W}(\state)}}{W'}{\psi_{W'}(\state_1)}
\iso \restrat{D'}{W'}{\psi_{W'}(\state')}.\label{eq:DWsigmaW'}
\end{equation}
Now we have $$\exta{X'}{\restrat{D}{W}{\psi_{W}(\state)}} (W') \iso
\ens{\state'' \in \exta{X}{D}(W \vrond W') \aalt \restr{\state''}{W} =
  \state},$$ where $\restr{\state''}{W}$ denotes restriction of
$\state''$ along the prefix inclusion $W \into W \vrond W'$.  So the
left-hand side in~\eqref{eq:DWsigmaW'} is just $\restrat{D}{(W \vrond
  W')}{\psi_{W \vrond W'} (\state_1)}$, which yields the desired
transition.
\end{proof}



\begin{lem}
  The projection $\pp^\CCC \colon \CCC \to \DvW$ satisfies the
  following weak Conduché condition: for all $X'' \xproto{W_2} X'
  \xproto{W_1} X$, if there is an edge $(X'',D'') \xto{W_1 \vrond W_2}
  (X,D)$ in $\CCC$, then there exists $D' \in \DD_{X'}$ and edges $(X'',D'') \xto{W_2} (X',D')
  \xto{W_1} (X,D)$.
\end{lem}
\begin{proof}
  Consider any $(X,D) \in \CCC$ and $\state \in \exta{X}{D} (W_1 \vrond
  W_2)$ witnessing the given edge.  Consider also the morphism $u \colon
  W_1 \to (W_1 \vrond W_2)$ given by $(W_2,\id)$, and let $\state_1 =
  \state \cdot u \in \exta{X}{D} (W_1)$.  Let $D_1 =
  \restrat{D}{W_1}{\psi_{W_1} (\state_1)}$.  We have $\state \in
  \ens{\state' \in \exta{X}{D} (W_1 \vrond W_2) \aalt \state' \cdot u
    = \state_1}$, hence $\state \in \exta{X'}{D_1} (W_2)$.
  Furthermore, 
  $$\restrat{D_1}{W_2}{\psi_{W_2} (\state)} \iso
  \restrat{D}{(W_1 \vrond W_2)}{\psi_{W_1 \vrond W_2} (\state)} \iso
  D'',$$
  so we have two edges
  $$(X,D) \xot{(W_1,\state_1)} (X',D_1) \xot{(W_2,\state)} (X'',D'')$$
  as desired.
\end{proof}

The previous result generalises by induction to $n$-ary composites:
\begin{notation}\label{not:right:assoc}
  By default, composition in $\DvW$ associates to the right, i.e., $W
  \vrond W' \vrond W''$ denotes $W \vrond (W' \vrond W'')$.
\end{notation}
\begin{cor}\label{cor:corresp}
  For any path $p$, say
  \begin{equation}
    X = X_0 \xotorp{M_1} X_1 \xotorp{M_2} \ldots X_n = X',\label{eq:path}
  \end{equation}
  in $\W$ and edge $(X',D') \xto{W} (X,D)$ in $\CCC$
  over its right-associated, $n$-ary composition $W = (M_1 \vrond
  (\ldots \vrond M_n))$, there is a path $e = (e_1, \ldots, e_n)$ in
  $\CCC^\star ((X',D'),(X,D))$ such that $$(\pp^\CCC)^\star_{(X',D'),(X,D)} (e) = p.$$
\end{cor}
\begin{proof}
  By induction on $n$.
\end{proof}

Our next goal is to relate transitions in $\CCC$ to sequences of transitions in $\SSS$.
First of all, $\CCC$ and $\SSS$ coincide on \actions{}:
\begin{lem}\label{lem:fibres}
  If $W \colon X' \proto X$ is a closed-world \action{}
  (i.e., has length 1), then for all $D$ and $D'$ both fibres of $\CCC
  ((X',D'),(X,D))$ and $\SSS((X',D'),(X,D))$ over $W$ are equal.
\end{lem}
\begin{proof}
  By~\cite[Proposition 5.23]{HirschoDoubleCats}.
\end{proof}
Now, let us show that for any sequence of closed-world \actions{},
sequences of transitions in $\SSS$ correspond to transitions over the
composite in $\CCC$.  \renewcommand{\truncate}[1]{#1}
\begin{cor}\label{cor:SW}
  For all closed-world paths as in~\eqref{eq:path}, and $(X,D),(X',D') \in \SSS$, 
  we have
  \begin{center}
    $\truncate{\inv{((\pp^\SSS)^\star_{(X',D'),(X,D)})} (p)} \neq
    \emptyset$ \hfil iff \hfil $\inv{(\pp^\CCC_{(X',D'),(X,D)})} (P) \neq
    \emptyset,$
  \end{center}
  for any special isomorphism $P \iso (M_1 \vrond (\ldots \vrond M_n))$.
\end{cor}
\begin{proof}
  Consider any special isomorphism $\alpha \colon P \isoto (M_1 \vrond (\ldots \vrond M_n))$.
  We have 
  \begin{equation*}
  \begin{array}[b]{ll}
    \truncate{\inv{((\pp^\SSS)^\star_{(X',D'),(X,D)})} (p) \neq \emptyset} \\
    \text{iff}\  \inv{((\pp^\CCC)^\star_{(X',D'),(X,D)})} (p)  \neq \emptyset
    & \mbox{(by Lemma~\ref{lem:fibres})} \\
    \text{iff}\  \inv{(\pp^\CCC_{(X',D'),(X,D)})} (M_1 \vrond (\ldots \vrond M_n))  \neq \emptyset
    & \mbox{(by Lemma~\ref{lem:compo} and Corollary~\ref{cor:corresp})}\\
    \text{iff}\  
    \inv{(\pp^\CCC_{(X',D'),(X,D)})} (P)  \neq \emptyset & \mbox{(by Lemma~\ref{lem:2d})}.
\end{array}\tag*{\qEd}
\end{equation*}
\def\popQED{}
\end{proof}
As a corollary, we get that the identity relation on objects is a
strong bisimulation between $\freecat{\SSS}$ and $\CCC$:
\begin{cor}\label{cor:SSS:CCC}
  For all $w \in \Sierp^\star (\star,\star)$ and $(X,D),(X',D') \in
  \SSS$, we have $(X,D) \xOt{w} (X',D')$ in $\SSS$ iff $(X,D)
  \xot{\idfree{w}} (X',D')$ in $\CCC$.
\end{cor}
The last statement is slightly subtle, in that $(X,D) \xot{\idfree{w}}
(X',D')$ denotes a single edge in $\CCC$, lying over the composite
$\idfree{w}$ in $\fcSierp$.
\begin{proof}
  If $(X,D) \xOt{w} (X',D')$ in $\SSS$, then there exists $p \in
  \W^\star$ such that $\freecat{\pp^\W} (p) = \idfree{w}$ and there is
  a path $e \colon (X',D') \to (X,D)$ over $p$ in $\SSS$.  Let $W
  \colon X' \proto X$ denote the composition of $p$.  By
  Corollary~\ref{cor:SW}, we get an edge $(X',D') \to (X,D)$ over $W$
  in $\CCC$. So since $\pp^\W (W) = \idfree{w}$, this gives us the
  expected transition.

  Conversely, if $(X,D) \xot{\idfree{w}} (X',D')$ in $\CCC$, then let
  $W \colon X' \proto X$ denote the corresponding edge in $\DvW$.  In
  particular, we have $\pp^\W (W) = \idfree{w}$.  Decomposing $W$ as
  some path $p$ in $\W$, we obtain by Corollary~\ref{cor:SW} a
  transition sequence $(X,D) \xOt{(\pp^\W)^\star(p)} (X',D')$ in
  $\SSS$. But $\idfree{(\pp^\W)^\star(p)} = \pp^\W (W) = \idfree{w}$,
  as desired.
\end{proof}
As promised, we readily obtain:
\begin{cor}\label{cor:botSC}
We have $\bot^\SSS = \bot^\CCC$.  
\end{cor}

Let us also prove the analogous result with the semantic pole $\bbot$.
\begin{lem}\label{lem:botsemS}
  We have $D \in \bbot_X$ iff $(X,D) \in \bot^\CCC$.
\end{lem}
\begin{proof}
  Assume $D \in \bbot_X$, and consider any $(X,D) \ot (X',D')$.  The
  latter is witnessed by some unsuccessful, closed-world \trace $W
  \colon X' \proto X$, state $\state \in \exta{X}{D}(W)$, and
  isomorphism $h \colon \restrat{D}{W}{\psi_W(\state)} \isoto D'$.

  By hypothesis, $\state$ admits an extension $\state' \in
  \exta{X}{D}(W \vrond W')$ for some successful $W' \colon X'' \proto
  X'$.  Letting $D'' = \restrat{D}{(W \vrond W')}{\psi_{W \vrond W'}
    (\state')}$, we have
  $$D''
  \iso \restrat{(\restrat{D}{W}{\psi_W(\state)})}{W'}{\psi_{W'}
    (\state')} \iso
  \restrat{D'}{W'}{\psi_{W'}(\exta{}{h}_{W'}(\state'))},$$ and hence
  $(X',D') \xot{\tick^n} (X'',D'')$ for some $n > 0$.  This shows that
  $(X,D) \in \bot^\CCC$.

  Conversely, assume $(X,D) \in \bot^\CCC$ and consider any
  unsuccessful, closed-world \trace $W \colon X' \proto X$ and state
  $\state \in \exta{X}{D}(W)$.  Letting $D' =
  \restrat{D}{W}{\psi_W(\state)}$, we have $(X,D) \ot (X',D')$.  By
  hypothesis, we find some transition $(X',D') \xot{\tick} (X'',D'')$,
  witnessed by some successful $W' \colon X'' \proto X'$. Hence, $D''
  \iso \restrat{D'}{W'}{\psi_{W'}(\state')}$ for a certain $\state'
  \in \exta{X'}{D'}(W')$. By definition of $D'$, $\state'$ is a state
  in $\exta{X}{D} (W \vrond W')$ such that $\state' \cdot u = \state$,
  where $u \colon W \to (W \vrond W')$ is $(W',\id)$.  This gives the
  desired successful extension of $\state$, which shows that $D \in
  \bbot_X$.
\end{proof}

Combining the last two results, we may now prove that the semantic
pole coincides with that of $\SSS$:
\begin{proof}[Proof of Lemma~\ref{lem:bbot:bot}]
  We have $D \in \bbot_X$ iff $(X,D) \in \bot^\CCC$ iff $(X,D) \in
  \bot^\SSS$ by Corollary~\ref{cor:botSC} and Lemma~\ref{lem:botsemS}.  
\end{proof}

As expected, this entails preservation and reflection of 
semantic fair testing equivalence:
\begin{cor}\label{cor:bbot:bot}
  For all $D \in \DD_X$ and $D' \in \DD_{X'}$, we have $(X,D) \faireq (X',D')$
  iff $(X,D) \faireqof{\SSS} (X',D')$
\end{cor}
\begin{proof}
  Let us first show that semantic fair testing equivalence may as well
  be defined only with definite tests. Indeed, if $I_X \neq I_{X'}$
  then the result holds trivially. So assuming $I_X = I_{X'}$,
  consider any test $B \in \Beh{Y}$ with $I_X = I_Y$, and, w.l.o.g.,
  $\exta{Z}{[D,B]} \in \bbot_Z$ and $\exta{Z'}{[D',B]} \notin
  \bbot_{Z'}$ with $Z = X +_{I_X} Y$ and $Z' = X' +_{I_{X'}} Y$. Then,
  letting $B = \sum_{i \in \gamma} D_i$ with each $D_i$ definite,
  there exists $i$ such that $\exta{Z}{[D , D_i ]} \in \bbot_Z$ and
  $\exta{Z'}{[D' , D_i ]} \notin \bbot_{Z'}$, hence $D_i$ also
  distinguishes $D$ from $D'$.  

  Returning to our main goal, for any definite test $T \in
  \DD_Y$ with $I_Y = I_X$, by Lemma~\ref{lem:bbot:bot}, $\exta{Z}{[D,
    T]} \in \bbot_Z$ iff $(Z, \exta{Z}{[D, T ]}) = ((X, D) \para (Y, T
  )) \in \bot^\SSS$, which easily entails the result.
\end{proof}

From this we may reduce our main theorem to a result on $\SSS$:
\begin{cor}\label{cor:main:SSS}
  If the translation $\translfun \colon \ob(\Conf) \to \sum_X \DD_{X}$
  is intensionally fully abstract for $\faireqof{\SSS}$, then
  Theorem~\ref{thm:main} holds, i.e., $\translfun$ is also
  intensionally fully abstract for $\faireq$.
\end{cor}
\begin{proof}
  For all configurations $C_1$ and $C_2$, by
  Corollary~\ref{cor:bbot:bot}, $\transl{C_1} \faireq \transl{C_2}$
  iff $\transl{C_1} \faireqof{\SSS} \transl{C_2},$ which holds iff
  $C_1 \faireqof{\Conf} C_2$ by hypothesis.
  
  It remains to prove that surjectivity up to $\faireq$ reduces to
  surjectivity up to $\faireqof{\SSS}$. For this, let us first show
  that any behaviour is $\faireq$-equivalent to some definite one.
  Indeed, consider any $B \in \Beh{X}$.  Letting $B \cdot x = \sum_{i
    \in n_x} D^x_i$ for all \agents{} $x$ in $X$, $B$ is fair testing
  equivalent to the definite behaviour $D$ such that
  $$D \cdot x = \tupling{ \tauuof{n_x} \mapsto \sum_{i \in n_x}
    D^x_i},$$ except if $B \cdot x = \emptyset$ for some $x$.  But in
  the latter case, $B$ is fair testing equivalent to the definite
  behaviour on one nullary player with the same interface which merely
  ticks.

  Thus, we may restrict attention to definite behaviours.  So consider
  any definite $D \in \DD_X$. By hypothesis, there exists $C$ such
  that $\transl{C} \faireqof{\SSS} (X,D)$, hence $\transl{C} \faireq
  (X,D)$ by Corollary~\ref{cor:bbot:bot}, which concludes the proof.
\end{proof}


\subsection{A further graph with testing for behaviours}
\label{subsec:MMM}
In the previous section, we have characterised semantic fair testing
equivalence using the graph with testing $\SSS$, and reduced
intensional full abstraction of $\translfun$ w.r.t.\ $\faireq$ to
intensional full abstraction w.r.t.\ $\faireqof{\SSS}$.  We now define
a further graph with testing, $\MMM$, which will help us bridge the
gap between $\SSS$ and $\pi$-calculus configurations. Indeed, we
define a surjective morphism $\mm \colon \SSS \to \MMM$ over $\Sierp$,
and we then prove that $\mm$ is intensionally fully abstract
(Proposition~\ref{prop:am}), from which we deduce that our main
result follows from intensional full abstractness of the composite
translation $\TT = \mm \rond \translfun$ (Lemma~\ref{lem:MMM}).

Recall from Definition~\ref{def:multisets} that $\Multiset{-}$ denotes
the finite multiset monad on sets.
\begin{defi}
  Let the set $\MMM_0$ of \emph{mixed behaviours} be
$$\sum_{\gamma \in \powfin(\Nat)} \Multiset{(\sum_{n \in \Nat} \DDn \times \gamma^n)}.$$
The graph $\MMM$ over $\Sierp$ is inductively defined by the rules in
Figure~\ref{fig:mmm}, where $S[\gamma_1 \subseteq \gamma_2]$ denotes
pointwise composition of the substitution component with the inclusion
$h\colon \gamma_1 \subseteq \gamma_2$, i.e., each $D[\sigma]$ is
replaced by $D[h \rond \sigma]$.

Let $\pp^\MMM \colon \MMM \to \Sierp$ denote the projection.
\end{defi}
\begin{figure}[ht]
  \begin{mathpar}
    \inferrule{%
      i \in \card{D \cdot \paralof{n}} \\
      j \in \card{D \cdot \pararof{n}} %
    }{%
      \conf{\gamma}{D[\sigma]} \xtrans{}{\id} %
      \conf{\gamma}{\restrat{D}{\paralof{n}}{i}[\sigma],\restrat{D}{\pararof{n}}{j}[\sigma]} %
    }
\and %
    \inferrule{
      i \in \card{D \cdot \tauuof{n}} 
    }{\conf{\gamma}{D[\sigma]} \xtrans{}{\id}
      \conf{\gamma}{\restrat{D}{\tauuof{n}}{i}[\sigma]}} 
\and %
    \inferrule{
      i \in \card{D \cdot \tickof{n}} 
    }{\conf{\gamma}{D[\sigma]} \xtrans{}{\tick}
      \conf{\gamma}{\restrat{D}{\tickof{n}}{i}[\sigma]}} 
\and %
    \inferrule{
      i \in \card{D \cdot \nuof{n}} \\
      a \notin \gamma
    }{%
      \conf{\gamma}{D[\sigma]} \xtrans{}{\id}
      \conf{\gamma,a}{\restrat{D}{\nuof{n}}{i}[n+1 \xto{\sigma+\name{a}} \gamma,a]}} 
\and %
 \inferrule{
      i \in \card{D_1 \cdot \iotaof{n_1}{a_1}} \\
      j \in \card{D_2 \cdot \outof{n_2}{a_2}{b_2}} \\
      \sigma_1 (a_1) = \sigma_2 (a_2) %
    }{ %
      \conf{\gamma}{D_1[\sigma_1],D_2[\sigma_2]} \xtrans{}{\id}
      \conf{\gamma}{ %
        \restrat{D_1}{\iotaof{n_1}{a_1}}{i}[n_1 + 1 \xto{[\sigma_1,\name{\sigma_2(b_2)}]} \gamma], %
        \restrat{D_2}{\outof{n_2}{a_2}{b_2}}{j}[\sigma_2] %
      }} 
\and %
    \inferrule{\conf{\gamma_1}{S_1} \xtrans{}{\alpha}
      \conf{\gamma_2}{S_2}}{ \conf{\gamma_1}{S \cup S_1} \xtrans{}{\alpha}
      \conf{\gamma_2}{S[\gamma_1 \subseteq \gamma_2] \cup S_2}}
\and %
    \inferrule{ }{\conf{\gamma}{S} \xtrans{}{\id} \conf{\gamma}{S}}
  \end{mathpar}
  \caption{Transitions in $\MMM$}
  \label{fig:mmm}
\end{figure}

\begin{notation}\label{not:MMM}
  Similarly to the notation for configurations, we denote 
  \begin{center}
  $(\gamma, [(n_1,D_1,\sigma_1), \ldots,(n_p,D_p,\sigma_p)])$ \hfil by \hfil
  $\conf{\gamma}{D_1[\sigma_1], \ldots, D_p[\sigma_p]}.$
\end{center}
\end{notation}

For the testing structure of $\MMM$, we mimick Definition~\ref{def:atpi} and put:
\begin{defi}\label{def:atm}
  For any $\conf{\gamma}{S}, \conf{\gamma'}{S'} \in \MMM$, let
  $\conf{\gamma}{S} @ \conf{\gamma'}{S'}$ denote $\conf{\gamma}{S \cup
    S'}$ if $\gamma = \gamma'$ and be undefined otherwise.
  Let furthermore  $\epsilon_\gamma = \conf{\gamma}{}$.
\end{defi}
By Lemma~\ref{lem:freetesting}, we have:
\begin{prop}
  The morphism $\freecat{\pp^\MMM} \colon \freecat{\MMM} \to
  \fcSierp$, with the graph of $@$ as testing relation, forms a free
  graph with testing.
\end{prop}

Let us now reduce Theorem~\ref{thm:main} to a statement about $\MMM$.
In order to do this, we will use Lemma~\ref{lem:fairness} and hence
need to exhibit a fair relation $\ob(\SSS) \proto \MMM_0$. We use the
graph of the following map:
\begin{defi}
  Let $\mm \colon \ob (\SSS) \to \MMM_0$ map any $(X,D)$ to the mixed
  behaviour
$$\conf{X (\star)}{[(D \cdot x)[\sigma_x] \aalt (n,x) \in \Pl (X)]},$$
where $\sigma_x$ is the map $n \xto{[\name{x \cdot s_i}]_{i \in n}}
X(\star).$
\end{defi}
We now need to show that $\mm$ yields a fair relation. 
Most points are direct, like totality or 
the fact that $(X,D) \coh (X',D')$ iff $\mm(X,D) \coh \mm(X',D')$.
Furthermore, we have:
\begin{prop}\label{prop:mm:surj}
  The map $\mm \colon \ob(\SSS) \to \MMM_0$ is surjective. Let $\aa$
  denote any section of $\mm$.
\end{prop}

The point about bisimilarity is trickier:
\begin{prop}\label{prop:mm:bisim}
  We have $(X,D) \bisimsierp \mm(X,D)$ for all $(X,D) \in \ob(\SSS)$.
\end{prop}
\begin{proof}
  It is enough to prove that (the graph of) $\mm$ is a strong
  bisimilarity up to strong bisimilarity. For this, let us record that
  clearly for any isomorphism $h\colon X \to Y$ of positions and $D
  \in \DD_Y$, we have $(Y, D) \bisimsierp (X, D \cdot h)$ in $\SSS$.
  Let us call $\III$ (for isomorphism) the relation given by all pairs
  $((Y,D),(X, D \cdot h))$, so that we have ${\III} \subseteq
  {\bisimsierp}$.  By case analysis, we can show that $\mm$ is a
  strong bisimulation up to $\III$, i.e., for all $(X,D) \xot{\alpha}
  (Y,D')$, there exists $(Y,D') \mathrel{\III} (Z,D'')$ such that
  $\mm(X,D) \xot{\alpha} \mm (Z,D'')$, as below left. And more
  tightly, for all $\mm(X,D) \xot{\alpha} M'$, there exists $(X,D)
  \xot{\alpha} (X',D')$ such that $\mm(X',D') = M'$, as below right.
  \begin{center}
    \Diag(.5,0){%
      \path (m-1-1.base) -- node[anchor=base] {$\xmapsto{\mm}$} (m-1-7.base) ;
    }{%
      (X,D) \& \& \& \& \& \& \mm (X,D) \\
      (Y,D') \& \III \&  (Z,D'') \& \xmapsto{\mm} \& \& \& \mm (Z,D'') %
    }{%
      (m-1-1)
      edge[<-,labell={\alpha}] (m-2-1) %
      (m-1-7) edge[<-,dashed,labelr={\alpha}] (m-2-7) %
    }
    \hfil
    \diag(.5,0){%
      (X,D) \& \xmapsto{\mm} \& \mm (X,D) \\
      (Y,D') \& \xmapsto{\mm} \& M' %
    }{%
      (m-1-1)
      edge[<-,dashed,labell={\alpha}] (m-2-1) %
      (m-1-3) edge[<-,labelr={\alpha}] (m-2-3) %
    } %
  \end{center}
  The right-hand diagram is a tedious yet straightforward case
  analysis. The left-hand one is also a tedious case analysis, whose
  main point is that for all transitions $(X,D) \xot{\alpha} (Y,D')$,
  some renaming of elements of $X$ may take place, which cannot happen
  in any transition from $\mm(X,D)$.  So in each case we need to find
  the $Z$ and corresponding $D''$ which avoids such renaming.  In
  fact, this goes by indentifying the right transition from $\mm(X,D)$
  and showing that the obtained $M'$ is of the form $\mm(Z,D'')$ for
  $(Y,D') \mathrel{\III} (Z,D'')$.  Let us do one case.  If $(X,D)
  \xot{\tick} (Y,D')$, then we have an iso $h \colon X
  \isoto Y$, and there is some agent $(n_0,x_0)$ in $X$ and $i$ such
  that $D' \cdot h(x_0) = \restrat{D \cdot x_0}{\tickof{n_0}}{i}$
  and furthermore for all $(n,x) \neq (n_0,x_0)$ in $\Pl (X)$, we
  have $D' \cdot h(x) = D \cdot x$. So, letting $y = h(x)$ for all
  such $x$, we indeed have
  $$M = \conf{\gamma}{(D \cdot x_0)[\sigma_{x_0}] \cons [(D \cdot x)[\sigma_x] \aalt x \neq x_0]} %
  \xtranszo{}{\tick} \conf{\gamma}{(D' \cdot y_0)[\sigma_{x_0}] \cons
    [(D \cdot x)[\sigma_x] \aalt x \neq x_0]} = M'$$%
  with $M = \mm(X,D)$, and furthermore letting $D''$ be determined by
  $D'' \cdot x = D' \cdot h(x) = D \cdot x$ for $x \neq x_0$ and $D''
  \cdot x = D' \cdot h(x_0)$, we have $D'' = D' \cdot h$ and hence
  $(Y,D') \mathrel{\III} (X,D'')$ with $\mm(X,D'') = M'$, as desired.
\end{proof}

We are now ready to show:
\begin{prop}\label{prop:am}
  The map $\mm\colon \ob(\SSS) \to \ob (\MMM)$ is intensionally fully
  abstract for fair testing equivalence.
\end{prop}
\begin{proof}
  As announced, for preservation and reflection of fair testing
  equivalence we apply Lemma~\ref{lem:fairness}: we have established
  all necessary hypotheses, except the last one, which follows by
  choosing a pushout with the same set of channels as its summands.

  For surjectivity up to $\faireqof{\MMM}$, consider any
  $\conf{\gamma}{S} \in \MMM_0$. Because $\aa$ is a section of
  $\mm$, we have $\mm (\aa (\conf{\gamma}{S})) = \conf{\gamma}{S}$,
  hence $\mm (\aa (\conf{\gamma}{S})) \faireqof{\MMM}
  \conf{\gamma}{S}$, thus providing the desired antecedent.
\end{proof}

Let us now reduce Theorem~\ref{thm:main} to its analogue about
$\MMM$. 
\begin{defi}
  Let $\TT\colon \ob(\Conf) \to \ob(\MMM)$ denote the composite
  $$\ob (\Conf) \xto{\translfun} \ob (\SSS) \xto{\mm} \ob(\MMM)\rlap{.}$$
\end{defi}
Concretely, we have
$$\TT\conf{\gamma}{P_1,\ldots,P_n} = \conf{\gamma}{\transl{P_1}_{h_\gamma}[\inv{h_\gamma}], \ldots,
  \transl{P_n}_{h_\gamma}[\inv{h_\gamma}]}\rlap{.}$$ 
\begin{lem}\label{lem:MMM}
  The translation $\translfun$ from Definition~\ref{def:transl} is
  intensionally fully abstract for $\faireqof{\SSS}$ if $\TT$ is for
  $\faireqof{\MMM}$.
\end{lem}
\begin{proof}
  Assuming $\TT$ is intensionally fully abstract, then for all
  configurations $C$ and $C'$, we have that $C \faireq C'$ iff $\TT(C)
  \faireqof{\MMM} \TT(C')$.  But by Proposition~\ref{prop:am}, we
  have
  \begin{center}
    $\TT(C) = \mm \transl{C} \faireqof{\MMM} \mm \transl{C'} = \TT(C')$ 
    \hfil iff \hfil
    $\transl{C} \faireqof{\SSS} \transl{C'}$, 
  \end{center}
  hence  $C \faireq C'$ iff $\transl{C} \faireq{\SSS} \transl{C'}$.

  Finally, for any $(X,D) \in \SSS$, by intensional full abstractness
  of $\TT$, we find a configuration $C$ such that $\TT(C)
  \faireqof{\MMM} \mm(X,D)$, hence by Proposition~\ref{prop:am} again
  $\transl{C} \faireq (X,D)$.
\end{proof}

\subsection{Intensional full abstraction}
\label{subsec:fullabs}
We at last prove our main result, by proving intensional full
abstractness of $\TT$.  Our strategy is to define a relation $\myexp
\colon \Conf \modto \MMM$ over $\Sierp$ which
\begin{itemize}
\item relates any configuration to its image under $\TT$, and
\item is surjective, i.e., relates any mixed behaviour to some configuration.
\end{itemize}
We will then show that this relation is a weak bisimulation over $\Sierp$, 
which will entail the result.

Let us start with the second point, and define a map
$\ZZ \in \ob(\MMM) \to \ob(\Conf)$, which associates a configuration
to each mixed behaviour.  We first coinductively define $\zeta$ for
definite behaviours on representable \positions by:
$$\begin{array}{rcll}
  \zeta (n \vdashdefinite D) & = & \left (
  \hspace*{-2em} \begin{array}{rcrll} && & \sum_{
              i \in \card{D \cdot \paralof{n}},
              j \in \card{D \cdot \pararof{n}} %
            }
          \tau.
            ( \zeta(n \vdashdefinite \restr{(D \cdot \paralof{n})}{i}) 
            \para  \zeta(n \vdashdefinite \restr{(D \cdot \pararof{n})}{j}) )
          \\
          &&+ & \sum_{i \in \card{D \cdot \tauuof{n}}} 
          \tau.
          \zeta(n \vdashdefinite \derivrest{D}{\tauuof{n}}{i})  \\
          &&+ & \sum_{i \in \card{D \cdot \tickof{n}}} 
          \tick.
          \zeta(n \vdashdefinite \derivrest{D}{\tickof{n}}{i})  \\
          &&+ & \sum_{i \in \card{D \cdot \nuof{n}}} 
          \nu (n+1). \zeta(n+1 \vdashdefinite \restr{(D \cdot \nuof{n})}{i}) \\
          &&+ & \sum_{a \in n, i \in \card{D \cdot \iotaof{n}{a}}} 
          a (n+1). \zeta(n+1 \vdashdefinite \derivrest{D}{\iotaof{n}{a}}{i}) \\ 
          &&+ & \sum_{a,b \in n, i \in \card{D \cdot \outof{n}{a}{b}}} 
          \send{a}{b}.
          \zeta(n \vdashdefinite \derivrest{D}{\outof{n}{a}{b}}{i}) 
        \end{array} \right )\rlap{,}
    \end{array}$$
    where $n \vdashdefinite D$ means $D \in \DD_{[n]}$.

    Except perhaps for the first term of the sum, this should be
    rather natural: each definite behaviour on a representable
    \position{} corresponds to a guarded sum, with one term for
    each state over each basic \action{} -- the translation is direct.
    The twist in the first term is due to the fact that forking is not
    a guard in $\pi$, so $P \para Q$ cannot appear in any guarded
    sum. But $\tau.(P \para Q)$ can, and it is clearly weakly
    bisimilar to $P \para Q$, so this is precisely what $\zeta$ does.

Let us now extend $\zeta$ to arbitrary configurations:
\begin{defi}
  Let $\ZZ \colon \ob(\MMM) \to \ob (\Conf)$ be defined by
$$\ZZ\conf{\gamma}{D_1[\sigma_1],\ldots,D_n[\sigma_n]} 
=
\conf{\gamma}{\zeta(D_1)[\sigma_1],\ldots,\zeta(D_n)[\sigma_n]}\rlap{.}$$
\end{defi}
\begin{rem}
  Let us emphasise that brackets on the right denote proper
  substitutions, while on the left $D_i[\sigma_i]$ is just syntactic
  sugar for $(n_i,D_i,\sigma_i)$ by Notation~\ref{not:MMM}.
\end{rem}
In order for $\ZZ$ to return an antecedent up to $\faireqof{\MMM}$, we
immediately observe that for any $D$ and $i \in \card{D
  \cdot \paralof{n}}$ and $j \in \card{D \cdot \pararof{n}}$,
\begin{itemize}
\item on the one hand $D$ has a silent transition to $\deriv_{i,j} D$,
  the behaviour on $[n] \para [n]$ such that $(\deriv_{i,j} D) \cdot
  x_1 = \restr{(D \cdot \paralof{n})}{i}$ and $(\deriv_{i,j} D) \cdot
  x_2 = \restr{(D \cdot \pararof{n})}{j}$ (where $x_1$ and $x_2$
  denote the two \agents of $[n] \para [n]$);
\item on the other hand $\zeta (D)$ has a silent transition 
  to $$\zeta'_{i,j} (D) = ( \zeta(\restr{(D \cdot \paralof{n})}{i}) 
            \para  \zeta(\restr{(D \cdot \pararof{n})}{j}) ),$$
            which then has a further silent transition to 
            the two-process configuration consisting of 
            $\zeta(\restr{(D \cdot \paralof{n})}{i})$ and 
            $\zeta(\restr{(D \cdot \pararof{n})}{j})$.
\end{itemize}
Thus, when we try to relate $D$ and $\zeta (D)$, the
transition $\zeta (D) \xtrans{}{\id} \zeta'_{i,j} (D)$ has to be
matched by the former transition $D \xtrans{}{\id} \deriv_{i,j} D$. So
our relation $\myexp$ should somehow include pairs $(\zeta'_{i,j}
(D),\deriv_{i,j} D)$.

\begin{defi}
  Let the relation $\myexp \colon \Conf \modto \MMM$ over $\Sierp$ be
  defined inductively by the rules in Figure~\ref{fig:myexp}.
\end{defi}
(In the last rule, $\epsilon_\gamma$ is understood in $\Conf$ on the
left, and in $\MMM$ on the right.)
\begin{figure}[ht]
  \begin{mathpar}
    \inferrule{\gamma' \vdash P \\ h \colon \gamma' \isoto n \\ \sigma
      \colon n \to \gamma %
    }{ %
      P[\sigma \rond h] \relmyexp \transl{P}_h[\sigma] } \and
    \inferrule{n \vdashdefinite D \\ \sigma \colon n \to \gamma %
    }{ %
      \zeta(D)[\sigma] \relmyexp D[\sigma] } \and \inferrule{n
      \vdashdefinite D^1,D^2 \\ \sigma \colon n \to \gamma %
    }{ %
      (\zeta(D^1) \para \zeta(D^2))[\sigma] \relmyexp
      D^1[\sigma],D^2[\sigma] %
    } %
    \and
    \inferrule{ %
      C \relmyexp M \\
      D \relmyexp N %
    }{ %
      C @ D \relmyexp M @ N %
    }
    \and
    \inferrule{ }{ %
      \epsilon_\gamma \relmyexp \epsilon_\gamma
    }~\cdot
  \end{mathpar}
  \caption{The relation $\myexp$}
\label{fig:myexp}
\end{figure}

\begin{lem}\label{lem:myexptotsurj}
  We have $C \relmyexp \TT(C)$ for all $C$ and $\ZZ(M) \relmyexp M$
  for all $M$, and so $\myexp$ is total and surjective.
\end{lem}
\begin{proof}
  By construction.
\end{proof}

\begin{lem}\label{lem:expansion}
  The relation $\myexp$ is a weak bisimulation over $\Sierp$.
\end{lem}
\begin{proof}
  First, we observe that $\myexp$ may equivalently be defined by first
  letting $\myexpo$ be generated by all rules except the last two,
  and then adding the rule
  \begin{mathpar}
    \inferrule{C_1 \relmyexpo M_1 \\ \ldots \\ C_n \relmyexpo M_n}{
      C_1 @ \ldots @ C_n \relmyexp M_1 @ \ldots @ M_n}~(n \in \Nat)\cdot
  \end{mathpar}
  The advantage of this presentation is that proofs of $C \relmyexp M$
  all have depth at most 1.

  The rest is then a tedious case analysis, which we defer to
  Appendix~\ref{app:proof}. In summary, we easily observe that the
  `forwards' clause of weak bisimulation is satisfied, the only
  subtlety being that heating $(\zeta(D^1) \para \zeta(D^2))[\sigma]$
  should be matched by the identity edge on the corresponding
  behaviour.  Furthermore, the `backwards' clause is also easily
  satisfied, the only subtlety being that if the considered behaviour
  performs a transition involving one or several
  $[D^1[\sigma],D^2[\sigma]]$'s, related on the left to $(\zeta
  (D^1) \para \zeta (D^2))[\sigma]$'s, then all of the latter first
  have to heat to $[\zeta (D^1)[\sigma], \zeta (D^2)[\sigma]]$, and
  only then perform the matching transition.
\end{proof}

This easily entails:
\begin{lem}\label{lem:main:M}
  For all $C_1,C_2,M_1,M_2$, if $C_1 \relmyexp M_1$ and $C_2 \relmyexp
  M_2$, then $C_1 \faireq C_2$ iff $M_1 \faireq M_2$.
\end{lem}
\begin{proof}
  The relation $\myexp$ is weakly fair, the only difficult points being 
  proved by Lemmas~\ref{lem:myexptotsurj} and~\ref{lem:expansion} above.
  We thus conclude by Corollary~\ref{cor:fairness:weak}.
\end{proof}

\begin{thm}\label{thm:main:M}
  The map $\TT\colon \ob (\Conf) \to \ob (\MMM)$ is intensionally fully abstract.
\end{thm}
\begin{proof}
  Lemma~\ref{lem:main:M} directly entails preservation and reflection
  of fair testing equivalence.  Regarding surjectivity up to
  $\faireqof{\MMM}$, for any $M \in \MMM_0$, we have $M
  \faireqof{\MMM} \TT(\ZZ((M)))$.  Indeed, for any $M'$, we have
  \begin{center}
    \begin{minipage}[c]{0.9\linewidth}
      \raggedright
      \hphantom{iff} $M @ M' \in \bot$
     \\ 
    iff $\ZZ(M) @ \ZZ(M') \in \bot$ 
    \\ \hfil
    (because $\ZZ(M) @ \ZZ(M') \relmyexp M @ M'$ and by Lemma~\ref{lem:expansion}) \\ 
    iff $\TT (\ZZ (M)) @ M' \in \bot$ 
    \\ \hfil
    (because $\ZZ(M) @ \ZZ(M') \relmyexp \TT(\ZZ(M)) @ M'$ and by Lemma~\ref{lem:expansion}),
    \end{minipage}
  \end{center}
  as desired.
\end{proof}

We are now able to prove our main result:
\begin{proof}[Proof of Theorem~\ref{thm:main}] By
  Theorem~\ref{thm:main:M}, $\TT$ is intensionally fully abstract for
  $\faireqof{\MMM}$, so by Lemma~\ref{lem:MMM} $\translfun$ is
  intensionally fully abstract for $\faireqof{\SSS}$.  We thus
  conclude by Corollary~\ref{cor:main:SSS}.
\end{proof}

\subsection{Generalisation}\label{subsec:gen}
We now show that our main results generalise beyond fair testing
equivalence.  Indeed, let us put:
\begin{defi}\label{def:pole}
  A \emph{pole} is a property of states over $\fcSierp$ which is
  stable under strong bisimilarity.
\end{defi}
There is a slight size issue in this definition, as it quantifies over
elements of all graphs over $\fcSierp$. The reader may understand this
using whatever fix they prefer, e.g., using a universe or some modal
logic.
\begin{example}
  Consider any $x \in G$ over $\fcSierp$.
  We have
  \begin{itemize}
  \item $x$ is in the pole for fair testing equivalence 
    iff for all $x \ot x'$ there exists $x' \xot{\tick} x''$;
  \item $x$ is in the
    pole for may testing equivalence iff there exists $x
    \xot{\tick} x'$.
  \end{itemize}
  Must testing equivalence is less easy to capture, for reasons
  explained in~\cite{2011arXiv1109.4356H}.  Here is an exotic, yet
  perhaps relevant pole: $x$ is in it iff for all finite,
  not-necessarily silent transition sequences $x \xotstar{} x'$, there
  exists $x' \xot{\tick} x''$. In other words, $x$ never loses the
  ability to tick. The induced equivalence is clearly at least as fine
  as fair testing equivalence, but we leave open the question of
  whether or not it is strictly finer.
\end{example}

\begin{defi}
  For any such pole $\bot$, let $\sim_\bot$ denote the testing
  equivalence induced by replacing $\bot^G$ by $\bot$ in the
  definition of fair testing equivalence
  (Definition~\ref{def:faireq}). 
\end{defi}
Semantic testing equivalence may then be taken to be testing
equivalence on $\CCC$ (Definition~\ref{def:CCC}), and we get the exact
analogue of Theorem~\ref{thm:main} (without changing the model in any
way).

\section{Conclusion and future work}\label{sec:conclu}
We have described our playground for $\pi$ and the induced sheaf
model, which we have proved intensionally fully abstract for a wide
range of testing equivalences.

Regarding future work: our proof that \traces form a playground uses a
new technique based on factorisation systems.  Since submission of
this paper, we have designed~\cite{TOtoEH,TOtoEHCALCO} a general
setting where this technique applies, and used it to bridge the gap
between our notion of plays based on string diagrams and the standard
one based on justified sequences~\cite{OngTsukada}.  We also consider
applying our notion of \trace to error
diagnostics~\cite{DBLP:conf/rv/GosslerMR10} or efficient machine
representation of reversible $\pi$-calculus
processes~\cite{DBLP:conf/lics/CristescuKV13}.  Longer-term directions
include applying the approach to more complex calculi, e.g., calculi
with passivation~\cite{DBLP:conf/fossacs/LengletSS09} or functional
calculi, and eventually consider some full-fledged functional language
with concurrency primitives.  Finally, deriving the complex notion of
trace evoked in Section~\ref{sec:contrib} from the one exposed here is
akin to deriving \ltss{} from reduction
rules~\cite{DBLP:conf/concur/LeiferM00,modularLTS}.  Since the issue
still seems easier on \traces than on a full operational semantics
specification, this might be worthwile to investigate further.  In the
same vein, the emphasis we put on traces suggests that we might be
able to deduce properties of type systems (soundness, progress, etc)
or compilers (correctness) from corresponding properties on traces.

\bibliographystyle{plainnat}
\bibliography{../common/bib}

\appendix

\section{Proof of Lemma~\ref{lem:expansion}}\label{app:proof}
\newcommand{\hs}{\hspace{0.5cm}}
\newcommand{\newcase}[1]{\hspace{\parskip}\newline\noindent\textbf{#1}}

One wants to check two properties:
\begin{itemize}
\item[(LH)] for all transitions $C' \xsnart{A}{a} C$ with $C\relmyexp
  M$, there exists $M' \xsnart{A}{a} M$ with $C' \relmyexp M'$;
\item[(RH)] for all transitions $M \xtrans{A}{a} M'$ with $C \relmyexp
  M$, there exists $C \xTrans{A}{a} C'$ with $C' \relmyexp M'$.
    \end{itemize}
    The attentive reader will have noticed that (LH) imposes $y$ to
    answer with a single transition. This means we actually prove that
    $\myexp$ is an \emph{expansion}~\cite[Chapter
    6]{SangioRutten}. Any expansion being in particular a weak
    bisimulation, this suffices.

    \begin{notation}
      We sometimes cast processes $P$ (resp.\ pairs $D[\sigma]$) over
      $\gamma$ into configurations $\conf{\gamma}{P}$ (resp.\ mixed
      behaviours $\conf{\gamma}{D[\sigma]}$). We proceed similarly for
      multisets of processes.
    \end{notation}

We start by proving (LH) for all cases, before proving that (RH) holds as well.

\newcase{Synchro, (LH).} We begin by the case of a synchronisation, i.e., when one has a transition
$$C = \conf{\gamma}{a(b).P+_{k_{1}}R_{1}, \send{a}{c}.Q+_{k_{2}}R_{2}} @ C_{0} %
\xtrans{}{id} %
\conf{\gamma}{P[b\mapsto c],Q} @ C_{0} = C'.$$
We want to show that there exists a transition $M\xtrans{}{id} M'$
with $C' \relmyexp M'$.

We write $P_{1}=a(b).P+_{k_{1}}R_{1}$ and
$P_{2}=\send{a}{c}.Q+_{k_{2}}R_{2}$. Neither of them are of the form
$(\cdot \para \cdot)$ so they can only be related to mixed behaviours
using the first two rules. Therefore, four sub-cases should be
considered, as detailed in Figure \ref{figsynchrolh}.  If we are in
case $i_1$ for $P_1$ and $i_2$ for $P_2$, then we have two mixed
behaviours $D_1[\sigma_1]$ and $D_2[\sigma_2]$ such that $n_i
\vdashdefinite D_i$, $\sigma_i \colon n_i \to \gamma$, and $P_i
\relmyexp D_i[\sigma_i]$ for $i = 1,2$, plus $M =
\conf{\gamma}{D_1[\sigma_1],D_2[\sigma_2]} @ M_0$ with $C_0 \relmyexp
M_0$.

\begin{itemize}
\item[$\bullet$]\textbf{Case 1 for both $P_{1}$ and $P_{2}$.}
We have 
$M=\transl{P'_{1}}_{h_{1}}[\sigma_{1}]@\transl{P'_{2}}_{h_{2}}[\sigma_{2}]@ M_{0}$, and 
there is a transition
$$M \xtrans{}{id} \transl{P'}_{h'_{1}}[\sigma'_{1}]@\transl{Q'}_{h_{2}}[\sigma_{2}]@ M_{0},$$
where $h'_{1}$ is $\gamma'_{1},b'\xto{h_{1}+!}n_{1}+1$ and $\sigma'_{1}$ is 
$n_{1}+1\xto{[\sigma_{1},c]}\gamma$.\\
Since $\sigma'_{1}\rond h'_{1}$ equals
$\gamma'_{1}\xto{h_{1}+!}n_{1}+1\xto{\sigma_{1}+\ceil{b}} (\gamma,b) \xto{b\mapsto
  c}\gamma$, we have that $P'[\sigma'_{1}\rond
h'_{1}]=P'[(\sigma_{1}+\ceil{b})\rond(h_{1}+!)][b\mapsto c]=P[b\mapsto
c]$, and therefore $P[b\mapsto c]\relmyexp
\transl{P'}_{h'_{1}}[\sigma'_{1}]$. Moreover, it is clear that
$Q=Q'[\sigma_{2}\rond h_{2}]\relmyexp
\transl{Q'}_{h_{2}}[\sigma_{2}]$, and finally $[P[b\mapsto c],Q]@
C_{0}\relmyexp M'$.

\item[$\bullet$]\textbf{Case 1 for $P_{1}$, Case 2 for $P_{2}$.}
We have a transition
$$M \xtrans{}{id} \transl{P'}_{h'_{1}}[\sigma'_{1}]@\restr{(D_{2}\cdot\outof{n_{2}}{a_{2}}{c_{2}})}{j}[\sigma_{2}]@ M_{0} = M',$$
where $h'_{1}$ is $\gamma'_{1},b'\xto{h_{1}+!}n_{1}+1$ and $\sigma'_{1}$ is $n_{1}+1\xto{[\sigma_{1},c]}\gamma$.\\
As is the previous case, one can check that $P[b\mapsto c]\relmyexp \transl{P'}_{h'_{1}}[\sigma'_{1}]$. Moreover, $Q\relmyexp \restr{(D_{2}\cdot\outof{n_{2}}{a_{2}}{c_{2}})}{j}[\sigma_{2}]$ and therefore $[P[b\mapsto c],Q]@ C_{0} \relmyexp M'$.

\item[$\bullet$]\textbf{Case 2 for $P_{1}$, Case 1 for $P_{2}$.}
We have a transition
$$M \xtrans{}{id} \restr{(D_{1}\cdot \iotaof{n_{1}}{a_{1}})}{i}[n_{1}+1\xto{[\sigma_{1},\ceil{c}]}\gamma]@\transl{Q'}_{h'_{2}}[\sigma'_{2}]@ M_{0} = M'.$$
As in the first case, we have $Q=Q'[\sigma_{2}\rond h_{2}]\relmyexp \transl{Q'}_{h_{2}}[\sigma_{2}]$. 
Furthermore, since 
$$\begin{array}{l}
\zeta(\restr{(D_{1}\cdot \iotaof{n_{1}}{a_{1}})}{i})[n_{1}+1\xto{[\sigma_{1},\ceil{c}]}\gamma]\\
\hs\hs=\zeta(\restr{(D_{1}\cdot \iotaof{n_{1}}{a_{1}})}{i})[n_{1}+1\xto{\sigma_{1}+\ceil{b}}\gamma,b][b\mapsto c]\\
\hs\hs=P[b\mapsto c],
\end{array}$$ 
we have that $[P[b\mapsto c],Q]@C_{0}\relmyexp M'$.

\item[$\bullet$]\textbf{Case 2 for both $P_{1}$ and $P_{2}$.}
We have a transition
$$M \xtrans{}{id} \restr{(D_{1}\cdot \iotaof{n_{1}}{a_{1}})}{i}[n_{1}+1\xto{[\sigma_{1},\ceil{c}]}\gamma]@\restr{(D_{2}\cdot\outof{n_{2}}{a_{2}}{c_{2}})}{j}[\sigma_{2}]@ M_{0} = M'.$$
As in the previous cases, one can show that $[P[b\mapsto c],Q]@C_{0}\relmyexp M'$.
\end{itemize}

\begin{figure}
\begin{tabular}{|l|p{.4\linewidth}|p{.45\linewidth}|}
\hline
& Case 1 & Case 2\\\hline
$P_{1}$ &
\begin{minipage}[t]{1.0\linewidth}
  There exist $\gamma'_{1} \vdash P'_{1} = a'_1(b').P'
  +_{k_1} R'_1$ and $h_{1} \colon \gamma'_{1} \isoto n_{1}$ such that
  \begin{mathpar}
    P_1 = P'_1[\sigma_1 \rond h_1]
    \and
    \sigma_1 (h_1 (a'_1)) = a
    \and
    P = P'[\gamma'_{1},b'\xto{h_{1}+!}n_{1}+1\xto{\sigma_{1}+\name{b}}\gamma,b]
    \and
    D_1 = \transl{P'_1}_{h_1}
    \and
    R_1 = R'_{1}[\sigma_{1}\rond h_{1}].
  \end{mathpar}\\[-2em]
\end{minipage}
&
\begin{minipage}[t]{1.0\linewidth}
  There exist $n_1 \vdashdefinite D_1$, $a_1 \in n_1$, and $i \in
  \card{D_1 \cdot \iotaof{n_1}{a_1}}$ such that
  \begin{mathpar}
    \sigma_1(a_1) = a 
    \and
    P_1 = \zeta (D_1)[\sigma_1]
    \and
    {P=\zeta (\restrat{D_1}{\iotaof{n_1}{a_1}}{i})[n_1+1 \xto{\sigma_1 + \name{b}} \gamma,b]}.
  \end{mathpar}
\end{minipage}
\\\hline
$P_{2}$ &
\begin{minipage}[t]{1.0\linewidth}
  \vspace*{-.7em} %
  There exist $\gamma'_{2} \vdash P'_{2} = \send{a'_{2}}{c'}.Q' +_{k_2} R'_2$ 
  and $h_{2} \colon \gamma'_{2} \isoto n_{2}$
  such that
  \begin{mathpar}
    P_2 = P'_2[\sigma_2 \rond h_2]
    \and
    Q = Q'[\gamma'_{2},b\xto{h_{2}+1}n_{2}+1\xto{\sigma_{2}+\name{b}}\gamma,b]
    \and
    \sigma_2 (h_2 (a'_2)) = a
    \and
    \sigma_2 (h_2 (c'_2)) = c
    \and
    D_2 = \transl{P'_2}_{h_2}
    \and
    R_2 = R'_{2}[\sigma_{2}\rond h_{2}].
  \end{mathpar}\\[-2em]
\end{minipage}
&
\begin{minipage}[t]{1.0\linewidth}
  \raggedright
  \vspace*{-.7em} %
  There exist $n_2 \vdashdefinite D_2$, $a_2,c_2 \in n_2$, and \linebreak $j \in \card{D_2 \cdot \outof{n_2}{a_2}{c_2}}$ such that
  \begin{mathpar}
    \sigma_2(a_2) = a 
    \and
    \sigma_2(c_2) = c
    \and
    P_2 = \zeta (D_2)[\sigma_2]
    \and
    {Q=\zeta (\restrat{D_2}{\outof{n_2}{a_2}{c_2}}{j})[\sigma_2]}.
  \end{mathpar}
\end{minipage}
\\\hline
\end{tabular}
\caption{Synchro, (LH) cases}\label{figsynchrolh}
\end{figure}

\noindent\textbf{Heating, (LH).} We now consider the case of heating, i.e., when one has a transition
$$C = (P\para Q)@C_{0} \xtrans{}{\tau}[P,Q]@ C_{0} = C'.$$ 
We want to show that there exists a transition $M\xtrans{}{\tau} M'$ with $[P,Q]@ C_{0} \relmyexp M'$.

We now have to consider a few cases, depending on which rule is applied
for $P\para Q$ in the proof of $(P\para Q)@C_{0} \relmyexp
M$. Notice that $P\para Q$ cannot be of the form
$\zeta(D)[\sigma]$. We are therefore left with two cases, depending
on whether the first or third rule is applied.
\begin{itemize}
\item[$\bullet$] If the first rule is applied, we find
$M_0$, $\gamma' \vdash P'\para Q'$, $h \colon \gamma' \isoto n$, and $\sigma
      \colon n \to \gamma$
      such that
      $P\para Q = (P'\para Q')[\sigma \rond h]$ and $M = \transl{P'\para Q'}_h[\sigma] @ M_0$,
      with $C_0 \relmyexp M_0$.
Letting $D=\transl{P'\para Q'}_h$, we notice that
$$D=\tupling{ \forklof{n}  \mapsto  \transl{P'}_h,
        \forkrof{n}  \mapsto  \transl{Q'}_h}.$$
Thus, there is a transition
$$M\xtrans{}{\tau}M'=[\transl{P'}_h[\sigma],\transl{Q'}_h[\sigma]]@M_{0}$$
with $[P, Q]@C_{0} \relmyexp M'$, as desired.

\item[$\bullet$] If the third rule is applied,
we find $M_0$, 
$n \vdashdefinite D^1,D^2$ and $\sigma \colon n \to \gamma$ such that %
$$P\para Q=(\zeta(D^1) \para \zeta(D^2))[\sigma]$$ 
and $M = [D^1[\sigma],D^2[\sigma]] @ M_0$.
We notice that $[P, Q]=[\zeta(D^1)[\sigma],\zeta(D^2)[\sigma]]$, so
$[P, Q]\relmyexp [D^1[\sigma],D^2[\sigma]]$, hence $[P, Q] @ C_{0}
\relmyexp M$.  The identity transition $M \xtrans{}{\tau} M$ thus fits
our needs.
\end{itemize}

\newcase{Nu, (LH).} We now consider the case of a $\nu$ rule, i.e., when one has a transition
$$C = \conf{\gamma}{\nu a.P +_{k} R}@C_{0} \xtrans{}{\tau}\conf{\gamma,a}{P} @ C_{0}[\gamma\subset \gamma,a] = C'.$$
We want to show that there exists a transition $M \xtrans{}{\tau} M'$
with $\conf{\gamma,a}{P} @ C_{0}[\gamma\subset \gamma,a] \relmyexp
M'$.  We notice that $\nu a.P+_{k} R$ cannot be obtained from the
third rule. We thus consider two cases corresponding to the first and
second rules.
\begin{itemize}
\item[$\bullet$] If the first rule is applied, there exist $M_0$, 
$\gamma' \vdash \nu a. P' +_{k} R'$, $h \colon \gamma' \isoto n$, and $\sigma
      \colon n \to \gamma$ such that
      $M = \transl{\nu a. P' +_{k} R'}_h[\sigma] @ M_0$ and
      $$\nu a. P +_{k} R %
      =(\nu a. P' +_{k} R')[\sigma \rond h].$$ We write
      $D=\transl{\nu a. P' +_{k} R'}_h$ as $\tupling{\nuof{n} \mapsto
        \transl{P'}_{h'}} \dplus \transl{R'}_{h}$, where $h'$ is
      $\gamma',a \xto{h+1} n+1$.
Thus, there is a transition
$$M\xtrans{}{\tau} M'=\conf{\gamma,a}{\transl{P'}_{h'}[n+1\xto{\sigma+\name{a}}\gamma,a]}@
M_{0}[\gamma\subset\gamma,a]$$
and, modulo the fact that $C\relmyexp M$ implies $C[\sigma]\relmyexp M[\sigma]$, we have $P@C_{0}[\gamma\subset \gamma,a]\relmyexp\transl{P'}_{h'}[\sigma+\name{a}]
@M_{0}[\gamma\subset\gamma,a]$ since $P=P'[\gamma',a\xto{h'}n+1\xto{\sigma+1}\gamma,a]$, as desired.

\item[$\bullet$] If the second rule is applied, there exist $M_0$,
  $n\vdashdefinite D$, and $\sigma \colon n \to \gamma$ such that 
$M = \zeta(D)[\sigma] @ M_0$ and %
$$\nu a. P +_{k} R=\zeta(D)[\sigma].$$
Thus, there exists $i\in\card{D\cdot \nuof{n}}$ such that $\zeta(\restr{(D\cdot\nuof{n})}{i})[n+1\xto{\sigma+\ceil{a}}\gamma,a]=P$. There is thus a transition
$$M\xtrans{}{\tau} M'=\conf{\gamma,a}{\restr{(D\cdot\nuof{n})}{i}[\sigma+\name{a}]} @ 
M_{0}[\gamma\subset\gamma,a]$$
with  $P@C_{0}[\gamma\subset\gamma,a]\relmyexp M'$, as desired.
\end{itemize}

\newcase{Tick and Tau, (LH).} We now consider the cases $\tick$ and
$\tau$, i.e., when one has a transition
$$C = (\xi.P+_{k} R)@C_{0} \xtrans{}{\xi} P @ C_{0},$$
where $\xi \in \ens{\tick,\tau}$.  We want to show that there exists a
transition $M\xtrans{}{\xi} M'$ with $P@ C_{0} \relmyexp M'$.

Once again, the third rule could not have been applied, and we are
left with two cases corresponding to the first and second rules.
\begin{itemize}
\item[$\bullet$] If the first rule is applied, then we find $M_0$,
$\gamma' \vdash \xi.P'+_{k} R'$, $h \colon \gamma' \isoto n$, and $\sigma
      \colon n \to \gamma$ such that
      $M = \transl{\xi. P' +_{k} R'}_h[\sigma] @ M_0$, $C_0 \relmyexp M_0$, and
      $$\xi.P+_{k} R=(\xi. P' +_{k} R')[\sigma \rond h].$$
      We write $D=\transl{\xi.P'+_{k} R'}_{h}$ as $\tupling{\xi_{n}
        \mapsto \transl{P'}_{h}} \dplus \transl{R'}_{h}$.  Thus, there
      is a transition
$$M=D[\sigma] @ M_0 \xtrans{}{\xi} \transl{P'}_{h}[\sigma]@M_{0} = M',$$
with $P=P'[\sigma\rond h]$ and thus $P@C_{0}\relmyexp M'$, as desired.

\item[$\bullet$] If the second rule is applied, then we find $M_0$, $n
  \vdashdefinite D$, and $\sigma \colon n \to \gamma$ such that %
  $C_0 \relmyexp M_0$,
  $M = D[\sigma] @ M_0$, and 
  $$\xi. P +_{k} R=\zeta(D)[\sigma].$$
  Then, there exists $i\in \card{D\cdot \xi_{n}}$ such that
  $\zeta(\restr{(D\cdot \xi_{n})}{i})[\sigma]=P$. Hence, there is a
  transition
$$M\xtrans{}{\xi} \restr{(D\cdot\xi_{n})}{i}[\sigma]@M_{0} = M'$$
with $P@C_{0}\relmyexp M'$, as desired.
\end{itemize}

We have thus proved that (LH) holds. We now proceed with the case analysis for (RH).

\newcase{Synchro, (RH).}
We start, once again, with the case of the synchronisation, i.e., when $M$ has the shape
$\conf{\gamma}{D_1[\sigma_1],D_2[\sigma_2]} @ M_0$ and we consider a silent transition to
$$M' = \conf{\gamma}{ %
  \restrat{D_1}{\iotaof{n_1}{a_1}}{i}[n_1 + 1
  \xto{[\sigma_1,\name{\sigma_2(b_2)}]} \gamma], %
  \restrat{D_2}{\outof{n_2}{a_2}{b_2}}{j}[\sigma_2] %
} @ M_0,$$ where $n_i \vdashdefinite D_i$ and $\sigma_i \colon n_i \to \gamma$ for $i = 1,2$, 
$\sigma_1(a_1)=a=\sigma_2(a_{2})$, and
$\sigma_2(c_2)=c$.  We want to show that there exists a transition
$C\xtrans{}{\id} C'$ with $C'\relmyexp M'$.  There are exactly 10
cases here. Firstly, we have the case where the third rule is
applied to $D_{1}[\sigma_1], D_2[\sigma_2]$.  Otherwise, one could
have used each of the three rules for $\relmyexp_{0}$ for each of
$D_{1}$ and $D_{2}$, yielding nine cases. We start with the first
case, and then treat the nine others.
\begin{itemize}
\item[$\bullet$] If the third rule is applied on
  $D_{1}[\sigma_1],D_2[\sigma_2]$ (hence $n_1=n_2=n$ and
  $\sigma_1=\sigma_2=\sigma$), then we find $C_0$ such that  
$C_0 \relmyexp M_0$ and
$C = (\zeta(D_1) \para \zeta(D_2))[\sigma] @ C_0$.
Then, we have a transition 
$$(\zeta(D_1) \para \zeta(D_2))[\sigma] @ C_0 \xtrans{}{\id}
[\zeta(D_1)[\sigma_{1}],\zeta(D_2)[\sigma_2]] @ C_0.$$ Since the
latter configuration is again related to $M$, this reduces to the case
where the second rule is applied for both $D_{1}[\sigma_1]$ and $D_{2}[\sigma_2]$.

\item[$\bullet$] If the third rule is applied for $D_{1}[\sigma_1]$
  and any of the three rules is applied for $D_{2}[\sigma_2]$, then we
  find $P$, $M_1$, $M_2$, $C_0$, and $n_1 \vdashdefinite D_3$ such
  that $M_1$ has length $1$ if the third rule is also used for $D_2$
  and $0$ otherwise, $M_0 = D_3[\sigma_3] @ M_1 @ M_2$, $P \relmyexp
  D_2[\sigma_2] @ M_1$, $C_0 \relmyexp M_2$, and $C = (\zeta(D_1) \para
  \zeta(D_3))[\sigma_1] @ P @ C_0$.

  Thus, there is a transition $$(\zeta(D_1) \para \zeta(D_3))[\sigma_1]
  @ P @ C_0
  \xtrans{}{\id}[\zeta(D_1)[\sigma_{1}],\zeta(D_3)[\sigma_1]] @ P @
  C_0$$ which reduces this case to the one where the second rule is
  applied to $D_1$.

\item[$\bullet$] If the third rule is applied for $D_{2}[\sigma_2]$
  and any of the first two rules is applied for $D_{1}$, then we find
  $M_1$, $P$, $C_0$, and $n_2 \vdashdefinite D_3$ such that $M_0 =
  D_3[\sigma_3] @ M_1$, $P \relmyexp D_1[\sigma_1]$, $C_0 \relmyexp
  M_1$, and $C = (\zeta(D_2) \para \zeta(D_3))[\sigma_2] @ P @ C_0$.
  Again, we are reduced to the case where the second rule is applied
  for $D_2$, using the transition
  $$(\zeta(D_2) \para
  \zeta(D_3))[\sigma_2] @ P @ C_0 
  \xtrans{}{\id}
  [\zeta(D_2)[\sigma_{2}],\zeta(D_3)[\sigma_2]] @ P @ C_0.$$
\end{itemize}

In the remaining cases, we have $C = P^0_1 @ P^0_2 @ C_0$ with $P^0_i
\relmyexp D_i[\sigma_i]$ for $i = 1,2$, and $C_0 \relmyexp M_0$.  We
thus have four cases, described in Figure~\ref{figsynchrorh} just as
we did in Figure~\ref{figsynchrolh}.
\begin{figure}
\begin{tabular}{|l|p{.4\linewidth}|p{.45\linewidth}|}
\hline
& Case 1 & Case 2\\\hline
$D_{1}$ &
\begin{minipage}[t]{1.0\linewidth}
  There exist $\gamma'_{1} \vdash P_{1} = a'_1(b).P'_1
  +_{k_1} R_1$ and $h_{1} \colon \gamma'_{1} \isoto n_{1}$ such that
  \begin{mathpar}
    P^0_1 = P_1[\sigma_1 \rond h_1]
    \and
    h_1 (a'_1) = a_1
    \and
    D_1 = \transl{P_1}_{h_1}
    \and
    i \in \card{D_1 \cdot \iotaof{n_1}{a_1}}
    \and 
    \restrat{D_1}{\iotaof{n_1}{a_1}}{i} = \transl{P'_1}_{h'_1}
  \end{mathpar}
  where $h'_1$ is $\gamma'_1, b \xto{h_1 + 1} n_1 + 1$.\\[-.5em]
\end{minipage}
&
\begin{minipage}[t]{1.0\linewidth}
$P^0_1 = \zeta(D_1)[\sigma_1]$.
\end{minipage}
\\\hline
$P_{2}$ &
\begin{minipage}[t]{1.0\linewidth}
  There exist $\gamma'_{2} \vdash P_{2} = \send{a'_{2}}{c'}.P'_2 +_{k_2} R_2$ 
  and $h_{2} \colon \gamma'_{2} \isoto n_{2}$
  such that
  \begin{mathpar}
    P^0_2 = P_2[\sigma_2 \rond h_2]
    \and
    h_2 (a'_2) = a_2
    \and
    h_2 (c'_2) = c_2
    \and
    D_2 = \transl{P_2}_{h_2}
    \and
    j \in \card{D_2 \cdot \outof{n_2}{a_2}{c_2}} 
    \and
    \restrat{D_2}{\outof{n_2}{a_2}{c_2}}{j} = \transl{P'_2}_{h_2}.
  \end{mathpar}\\[-2em]
\end{minipage}
&
\begin{minipage}[t]{1.0\linewidth}
$P^0_2 = \zeta(D_2)[\sigma_2]$.
\end{minipage}
\\\hline
\end{tabular}
\caption{Synchro, (RH) cases}\label{figsynchrorh}
\end{figure}

\begin{itemize}
\item[$\bullet$] If the second rule is applied for both $D_1$ and
  $D_{2}$, then there is a transition
$$C\xtrans{}{\id}C'=\zeta(\restr{(D_1\cdot\iotaof{n_1}{a_1})}{i})[[\sigma_1,\ceil{c}]]
@\zeta(\restr{(D_2\cdot\outof{n_2}{a_2}{c_2})}{j})[\sigma_2]@M_0$$
with $C'\relmyexp M'$ as desired.

\item[$\bullet$] If we apply the second rule for $D_1$ and the first
  rule for $D_{2}$, then there is a transition
$$C\xtrans{}{\id}C'=
\zeta(\restr{(D_1\cdot\iotaof{n_1}{a_1})}{i})[[\sigma_1,\ceil{c}]]@P'_2[\sigma_2\rond h_2]@C_0$$
with $C'\relmyexp M'$ as desired.

\item[$\bullet$] If the first rule is applied for both $D_{1}$ and
  $D_{2}$, then we have a synchronisation between $P_1[\sigma_1 \rond
  h_1]$ and $P_2[\sigma_2 \rond h_2]$.  In order to determine the
  result of this synchronisation we need to choose a representative
  for $P_1[\sigma_1 \rond h_1]$, i.e., pick a channel for what $b$
  becomes after substitution. A reasonable choice here is $\gamma+1$,
  we choose
$$P_1[\sigma_1 \rond h_1] = a(\gamma+1).(P'_1[(\sigma_1+1) \rond h'_1])+_{k_1} R_1[\sigma_1 \rond h_1].$$
We thus have a transition to
$$C' = P'_1[(\sigma_1+1) \rond h'_1][\gamma + 1 \mapsto c] @ 
P'_2[\sigma_2 \rond h_2] @ C_0.$$ 
But the diagram
\begin{center}
  \diag{%
    n_1+1 \& \& \gamma+1 \\
    \& \& \gamma %
  }{%
    (m-1-1) edge[labela={\sigma_1 + 1}] (m-1-3) %
    edge[labelbl={[\sigma_1,\name{c}]}] (m-2-3) %
    (m-1-3) edge[labelr={[\gamma+1 \mapsto c]}] (m-2-3) %
  }
\end{center}
commutes, so 
$C' = P'_1[[\sigma_1,\name{c}] \rond h'_1] @ 
P'_2[\sigma_2 \rond h_2] @ C_0.$
Finally, we also know: $$M'
\begin{array}[t]{l}
=
(\restr{(D_1\cdot\iotaof{n_1}{a_1})}{i})[[\sigma_1,\name{c}]] @
\restr{(D_2\cdot\outof{n_2}{a_2}{c_2})}{j}[\sigma_2] @ M_0 \\
=
\transl{P'_1}_{h'_1}[[\sigma_1,\name{c}]] @
\transl{P'_2}_{h_2}[\sigma_2] @ M_0,
\end{array}
$$
which entails $C' \relmyexp M'$ as desired.

\item[$\bullet$] If we apply the first rule for $D_1$ and the second
  rule for $D_2$, then, choosing the same representative as before for
  $P_1[\sigma_1 \rond h_1]$, there is a transition
$$C\xtrans{}{\id} P'_1[(\sigma_1+1) \rond h'_1][\gamma+1 \mapsto c] 
@ \zeta(\restr{(D_2 \cdot \outof{n_2}{a_2}{c_2})}{j})[\sigma_2] @ C_0
= C'$$ satisfying $C'\relmyexp M'$ (for the same reason as in the last
case).
\end{itemize}

\newcase{Fork, (RH).}  We consider the case of a forking \action,
i.e., $M = \conf{\gamma}{D[\sigma]} @ M_0$ with $n \vdashdefinite D$
and $\sigma \colon n \to \gamma$, and we have a transition
$$\conf{\gamma}{D[\sigma]} @ M_0 \xtrans{}{\id} %
\conf{\gamma}{\restrat{D}{\paralof{n}}{i}[\sigma],\restrat{D}{\pararof{n}}{j}[\sigma]}
@ M_0$$
for some $i$ and $j$. We want to show that there exists a transition $C
\xtrans{}{\id} C'$ with $C'\relmyexp M'$. We proceed by case analysis
on the rule applied for $D[\sigma]$ in the proof of $C \relmyexp M$.
\begin{itemize}
\item[$\bullet$] If the first rule is applied, then we find $C_0 \relmyexp M_0$,
$\gamma' \vdash P=P_{1}\para P_{2}$, and $h \colon \gamma' \isoto n$
such that $D = \transl{P}_{h}$, $C = P[\sigma \rond h] @ C_0$,
$\restr{(D\cdot\forklof{n})}{i}=\transl{P_1}_{h}$, and  $\restr{(D\cdot\forkrof{n})}{j}=\transl{P_2}_{h}$.

Thus, there is a transition 
$$P[\sigma\rond h] @ C_0 \xtrans{}{\id} P_{1}[\sigma\rond h]@P_{2}[\sigma\rond h] @ C_0 = C'$$
with $C' \relmyexp M'$ as desired.

\item[$\bullet$] If the second rule is applied, then 
we find $C_0 \relmyexp M_0$ such that $C = \zeta(D)[\sigma] @ C_0$.
Thus, $\zeta(D)[\sigma]$ has the 
shape 
\begin{equation}
\tau.(\zeta (\restr{(D\cdot\forklof{n})}{i} \para \zeta (\restr{(D\cdot\forkrof{n})}{j})))[\sigma] 
+_k R\label{eq:shape:zeta}
\end{equation}

so we have
$$\begin{array}{rcl}
  \zeta(D)[\sigma] @ C_0 & 
  \xtrans{}{\id}&
  (\zeta(\restr{(D\cdot\forklof{n})}{i})\para\zeta(\restr{(D\cdot\forkrof{n})}{j}))[\sigma] @ C_0 = C' 
\end{array}$$
with $C' \relmyexp M'$ (using the third rule) as desired.

\item[$\bullet$] If the third rule is applied, then we find $n
  \vdashdefinite D'$, $M_1$, and $C_0$ such that $C_0 \relmyexp M_1$, $C =
  (\zeta(D) \para \zeta(D'))[\sigma] @ C_0$, and
  $M_{0}=D'[\sigma]@M_{1}$.  But then, as in the previous case,
  $\zeta(D)[\sigma]$ has the shape~\eqref{eq:shape:zeta} and we have
  transitions:
$$
\begin{array}{rcl}
(\zeta(D) \para \zeta(D'))[\sigma] @ C_0 
& \xtrans{}{\id} &
\zeta(D)[\sigma]@ \zeta(D')[\sigma] @ C_0 \\
& \xtrans{}{\id}^{2} & 
  \zeta(\restr{(D\cdot\forklof{n})}{i})[\sigma]
  @ \zeta(\restr{(D\cdot\forkrof{n})}{j})[\sigma]
  @ \zeta(D')[\sigma]
  @ C_0 \\
  & = & C'
\end{array}
$$
with $C' \relmyexp M'$ as desired.
\end{itemize}

\newcase{Nu, (RH).}
We consider the case of a $\nu$ rule, i.e., one has a transition
$$C = \conf{\gamma}{D[\sigma]} @ M_{0} \xtrans{}{\id} %
\conf{\gamma}{\restrat{D}{\nuof{n}}{i}[n+1\xto{\sigma+\ceil{a}}\gamma,a]}
@ M_{0}[\gamma\subset\gamma,a]$$ with $n \vdashdefinite D$ and $\sigma
\colon n \to \gamma$. We want to show that there exists a transition
$C\xtrans{}{\id} C'$ with $C'\relmyexp M'$. We again proceed by case
analysis on the rule applied for $D[\sigma]$.

\begin{itemize}
\item[$\bullet$] If the first rule is applied, then we find
$C_0 \relmyexp M_0$, $\gamma' \vdash P=\nu a.P'+_{k} R$, and $h \colon \gamma' \isoto n$ such that
$C = P@C_0$,
$D = \transl{P}_{h}$
and $\restr{(D\cdot\nuof{n})}{i}=\transl{P'}_{h+1}$.  There are thus transitions
$$C=P[\sigma\rond h] @ C_0 
\xtrans{}{\id} P'[\gamma',a\xto{h+1}n+1\xto{\sigma+\ceil{a}}\gamma,a] @ C_{0}[\gamma\subset\gamma,a] = C'$$ 
with $C' \relmyexp M'$ as desired.

\item[$\bullet$] If the second rule is applied, then we find
$C_0 \relmyexp M_0$ such that 
$C = \zeta(D)[\sigma] @ C_0$, and
there is a transition
$$\zeta(D)[\sigma] @ C_0 
\xtrans{}{\id}
\zeta(\restr{(D\cdot\nuof{n})}{i})[n+1\xto{\sigma+\ceil{a}}\gamma,a]@C_{0}[\gamma\subset\gamma,a] = C'$$
with $C' \relmyexp M'$ as desired.

\item[$\bullet$] If the third rule is applied, then we find $M_1$,
  $C_0$, and $n \vdashdefinite D'$ such that $M_0 = D'[\sigma] @ M_1$,
  $C_0 \relmyexp M_1$, and $C = (\zeta(D) \para \zeta(D'))[\sigma] @
  C_0$.  But then we have
$$\begin{array}{l}
 (\zeta(D) \para \zeta(D'))[\sigma] @ C_0 \\
 {} \xtrans{}{\id} 
 \zeta(D)[\sigma]@\zeta(D')[\sigma] @ C_0 \\
 {} \xtrans{}{\id} \zeta(\restr{(D\cdot\nuof{n})}{i})[n+1\xto{\sigma+\ceil{a}}\gamma,a]@\zeta(D')[\sigma][\gamma\subset\gamma,a]@C_{0}[\gamma\subset\gamma,a] = C'
 \end{array}$$
 and $C' \relmyexp M'$ as desired.
\end{itemize}

\newcase{Tick and Tau, (RH).}
We now consider the cases $\tick$ and $\tau$, i.e., when one has a transition
$$ \conf{\gamma}{D[\sigma]} @M_{0} \xtrans{}{\xi} %
\conf{\gamma}{\restrat{D}{\xi_{n}}{i}[\sigma]} @ M_{0}$$ where
$\xi\in\{\tick,\tau\}$, with $n \vdashdefinite D$ and $\sigma \colon n
\to \gamma$.  We want to show that there exists a transition
$C\xtrans{}{\xi} C'$ with $C'\relmyexp M'$. Again, we proceed by case
analysis on the rule applied for $D[\sigma]$.

\begin{itemize}
\item[$\bullet$] If the first rule is applied, we find
$C_0 \relmyexp M_0$,
$\gamma' \vdash P=\xi.P'+_{k} R$, and $h \colon \gamma' \isoto n$
such that $C = P[\sigma \rond h] @ C_0$, $D = \transl{P}_{h}[\sigma]$,
and $\restr{(D\cdot\xi_{n})}{i}=\transl{P'}_{h}$.
But then we have 
$$P[\sigma\rond h] @ C_0 \xtrans{}{\xi} P'[\sigma\rond h] @ C_0 \relmyexp \transl{P'}_{h}[\sigma] @ M_0,$$
as expected.
\item[$\bullet$] If the second rule is applied, then 
$C = \zeta(D)[\sigma] @ C_0$ for some $C_0 \relmyexp M_0$
and we have
$$\zeta(D)[\sigma] @ C_0 
\xtrans{}{\xi}
\zeta(\restr{(D\cdot \xi_{n})}{i})[\sigma] @ C_0 
\relmyexp \restrat{D}{\xi_{n}}{i}[\sigma] @ M_0,$$
as desired.
\item[$\bullet$] Finally, if the third rule is applied, then
we find $C_0$, $M_1$, and $n \vdashdefinite D'$ such that
$C = (\zeta(D) \para \zeta(D'))[\sigma] @ C_0$,
$M_0 = D'[\sigma] @ M_1$, and $C_0 \relmyexp M_1$.
But then, we have
$$(\zeta(D) \para \zeta(D'))[\sigma] @ C_0
\xtrans{}{\id}
\zeta(D)[\sigma] @ \zeta(D')[\sigma] @ C_0 
\xtrans{}{\xi}
\zeta(\restr{(D\cdot \xi_{n})}{i})[\sigma] @ \zeta(D')[\sigma] @ C_0 = C'$$
with $C' \relmyexp M'$ as desired.
\end{itemize}


\end{document}